\documentclass[11pt]{article} %

\usepackage{url}
\usepackage{amsmath}
\usepackage{amsthm}
\usepackage{amssymb}
\usepackage{mathtools}
\usepackage{thmtools}
\usepackage{dsfont}
\usepackage{graphicx}
\usepackage{subcaption}
\usepackage{caption}
\usepackage[square,semicolon,numbers]{natbib}

\usepackage{booktabs}    %
\usepackage{multirow}    %
\usepackage{array}

\usepackage[margin=1in]{geometry}
\usepackage{charter}

\usepackage{amsfonts,amssymb}
\usepackage{bbm}
\usepackage{mathrsfs}
\usepackage{float}
\usepackage{bm}

\usepackage{amsthm,thmtools,thm-restate}
\usepackage{mathtools} %
\usepackage{amsmath} %
\usepackage{algorithm}
\usepackage[noend]{algpseudocode}
\usepackage{etoolbox}

\makeatletter
\newcommand{\linelabel}[1]{%
  \edef\@currentlabel{\theALG@line}%
  \label{#1}%
}
\makeatother
\algrenewcommand\algorithmicrequire{\textbf{Input:}}
\algrenewcommand\algorithmicensure{\textbf{Output:}}

\usepackage{tabularx}
\usepackage{tabularx}

\usepackage[dvipsnames]{xcolor} %
\definecolor{ptblue}{RGB}{15,76,129} %
\definecolor{ptemerald}{HTML}{009473} %
\definecolor{bluegray}{rgb}{0.4, 0.6, 0.8}
\definecolor{ptilluminating}{HTML}{F5DF4D} %
\definecolor{ptgray}{HTML}{939597}
\usepackage{cprotect}
\definecolor{cobalt}{rgb}{0.0, 0.28, 0.67}
\usepackage{hyperref}
\hypersetup{
linktocpage,
colorlinks=true,
citecolor=cobalt, %
urlcolor=ptblue, %
linkcolor=cobalt, %
}

\usepackage{cleveref}

\usepackage{thm-restate}
\usepackage{enumitem}
\usepackage{verbatim}

\DeclareMathOperator*{\argmax}{arg\,max}
\DeclareMathOperator*{\argmin}{arg\,min}

\theoremstyle{plain}
\newtheorem{theorem}{Theorem}
\newtheorem{proposition}[theorem]{Proposition}
\newtheorem{lemma}[theorem]{Lemma}
\newtheorem{corollary}[theorem]{Corollary}
\newtheorem{claim}[theorem]{Claim}
\theoremstyle{definition}
\newtheorem{definition}{Definition}

\newtheorem{example}{Example}

\theoremstyle{remark}
\newtheorem{remark}{Remark}
\theoremstyle{definition}

\crefname{assumption}{assumption}{assumptions}
\Crefname{assumption}{Assumption}{Assumptions}

\newcommand{\RR}{\mathbb{R}}

\title{\bfseries Packing Linear Programs and Fractional Knapsack using Comparison Oracles}

\author{
Ritabrata Barat\thanks{Indian Institute of Science. \texttt{ritobrotobarat.100@gmail.com}}\hfill
\and
Siddharth Barman\thanks{Indian Institute of Science. \texttt{barman@iisc.ac.in}}\hfill
\and
Nirjhar Das\thanks{Indian Institute of Science. \texttt{nirjhardas@iisc.ac.in}}\hfill
\and
Sukruta Midigeshi\thanks{Work done while at Microsoft Research, India. \texttt{sukruta.midigeshi@gmail.com}}
}
\date{}

\begin{document}

\maketitle

\begin{abstract}
We study the problem of recovering the objective of a packing linear program when the algorithm has access only to comparison information about optimal solutions under varying constraint matrices. This model is motivated by recent work on optimization with comparison oracles~\cite{gupta2026combinatorial} and on learning from preference feedback~\cite{kaufmann2025a}. It also strengthens revealed-preference and inverse-optimization frameworks by replacing direct observations of optimal solutions with ordinal information obtained via comparison queries.

We focus on the important special case of the fractional knapsack problem, in which the packing linear program has a single budget constraint specified by item prices (sizes), and the underlying linear objective is determined by the item values. The fractional knapsack formulation captures a monopoly-pricing setting in which a seller seeks to infer a buyer' unknown (additive) valuations for divisible items from comparison information alone. Here, the algorithm (seller) can query an oracle with two price vectors, and the oracle returns which of the corresponding optimal solutions has larger total packing of the items. Our results also apply to oracles that report which of the corresponding optimal solutions attains the larger objective value. Such comparison oracles provide a meaningful abstraction of discrete-choice surveys, where buyers repeatedly choose between alternatives offered at different prices.

For fractional knapsack, we develop a polynomial-time algorithm that recovers the item values up to scale using $O(n \log(1/\delta) + B^2)$ comparison queries, where $n$ is the number of items, $B$ is the known capacity of the knapsack, and $\delta$ is grid resolution of the values. We complement this result with an $\Omega(n \log(1/\delta))$ lower bound on the query complexity of fractional knapsack. 

A key insight of our work is that, in the comparison-oracle model, fractional knapsack is essentially as general as packing linear programs: the algorithm for fractional knapsack solves the packing setting by treating a row of the constraint matrix as the price vector and setting the remaining rows to zero, while the $\Omega(n \log(1/\delta))$ query lower bound continues to hold for packing linear programs as well. Consequently, the upper bound obtained via the knapsack reduction is essentially best possible for packing linear programs, up to a linear-factor gap.

Finally, we extend our knapsack algorithm to the profit-maximization setting, yielding a comparison-oracle analogue of the revealed-preference result of \cite{amin2015onlineProfitMax}.
\end{abstract}

\newpage
\thispagestyle{empty}

\tableofcontents

\thispagestyle{empty}
\newpage

\setcounter{page}{1}

\section{Introduction}
Packing linear programs constitute a central primitive in combinatorial optimization, approximation algorithms, and operations research \cite{toth20172,hochbaum1997approximating,plotkin1995fast,young2001sequential}. They model resource-allocation problems in which one seeks to maximize a nonnegative linear objective subject to nonnegative capacity (budget) constraints, and arise in settings ranging from fractional knapsack \cite{korte2008combinatorial} and multi-commodity flows \cite{garg2007faster} to ad allocation \cite{mehta2010online} and scheduling \cite{li2023nearly}.

In many recent applications, the feasible region---encoded by a constraint matrix---is under the algorithm designer's control, while the objective coefficients represent latent values, utilities, or profits that are not directly observable; see the monopoly-pricing and profit-maximization settings below for specific examples. This raises the question of whether one can recover the underlying objective of a packing LP by probing the system with different constraints and observing the induced optimal solution. Indeed, this question is closely related to the literature on inverse optimization~\cite{chan2025inverse} and to revealed-preference models such as those studied in \cite{amin2015onlineProfitMax}. These frameworks, however, typically assume access to the induced optimal solutions in their entirety. In many modern settings, by contrast, such cardinal information about the underlying optimal solution is unavailable, and one has access only to ordinal feedback. This viewpoint is central to recent work on optimization using comparison oracles~\cite{gupta2026combinatorial,kane2017active} and on learning from preference feedback, including RLHF settings~\cite{kaufmann2025a}, where the algorithm receives pairwise comparisons rather than explicit numerical information. Comparison-based feedback is also a more realistic model than exact quantitative data in settings such as preference elicitation \cite{peters2021preference,wang2026preference}, dueling-bandit models \cite{bengs2021preference}, and ranking systems \cite{hofmann2013balancing,szorenyi2015online}.

Motivated by such considerations, we study the problem of recovering the objective of a packing linear program when the algorithm has access only to comparison information about optimal solutions under varying constraint matrices. Specifically, we consider a model in which the algorithm may query two constraint matrices $\mathbf{P}, \mathbf{Q} \in [0,1]^{m\times n}$. For an a priori unknown but fixed objective vector $\mathbf{v} \in \mathbb{R}^n_{>0}$ and known $\mathbf{b} \in \mathbb{R}^m_{>0}$, let $\mathbf{x}^*(\mathbf{P})$ and $\mathbf{x}^*(\mathbf{Q})$ denote optimal solutions to the packing programs
\[
\max_{\mathbf x\in[0,1]^n} \ \mathbf v^\top \mathbf x \quad \text{s.t.} \quad \mathbf{P} \mathbf{x} \le \mathbf{b} 
\qquad\text{and}\qquad
\max_{\mathbf x\in[0,1]^n} \ \mathbf v^\top \mathbf x \quad \text{s.t.} \quad \mathbf{Q} \mathbf{x} \le \mathbf{b},
\]
respectively. The oracle then reports which of $\mathbf{x}^*(\mathbf{P})$ and $\mathbf{x}^*(\mathbf{Q})$ has the larger sum of coordinates. Our results apply more generally when the oracle reports whether $\mathbf{w}^\top \mathbf{x}^*(\mathbf{P}) \geq \mathbf{w}^\top \mathbf{x}^*(\mathbf{Q})$, for a fixed but unknown comparison weight vector $\mathbf{w} \in \mathbb{R}_{> 0}^n$. The special case $\mathbf{w} = \mathbf{v}$ case corresponds to an oracle that reports which of the corresponding optimal solutions attains the larger objective value. The comparison-oracle model thus strengthens revealed-preference and inverse-optimization frameworks by replacing direct observations of optimal solutions with comparison queries between the optimal solutions induced by two different constraints. 

We focus in particular on the important special case of the fractional knapsack problem, in which the packing linear program has a single budget constraint specified by item prices (sizes), and the---a priori unknown---linear objective is determined by the item values. Formally, for a queried price vector $\mathbf{p} \in [0,1]^n$ and a known capacity $B \in \mathbb{R}_{>0}$, the underlying problem is $\max_{\mathbf x\in[0,1]^n} \ \mathbf v^\top \mathbf x \ \ \text{ s.t. } \  \mathbf p^\top \mathbf x\le B.$ 

This formulation captures a monopoly-pricing setting in which a seller seeks to infer an additive buyer's unknown valuations $\mathbf{v} \in \mathbb{R}^n_{>0}$ for divisible items from comparison information alone. The algorithm (seller) can query an oracle with two price vectors $\mathbf{p}, \mathbf{q} \in [0,1]^n$, and the oracle reports which of the corresponding optimal solutions, $\mathbf{x}^*(\mathbf{p})$ or $\mathbf{x}^*(\mathbf{q})$, has larger total consumption (equivalently, larger total packing of the items). As noted above, our results also apply to oracles that report which of the corresponding optimal solutions attains the larger objective value. Such oracles provide a meaningful abstraction of discrete-choice surveys \cite{louviere2000stated,train2009discrete}, where buyers repeatedly choose between alternatives offered at different prices.

\medskip
 
\noindent {\bf Our Results.} For fractional knapsack, we develop a polynomial-time algorithm that recovers the item values up to scale using $O(n \log(1/\delta) + B^2)$ comparison queries, where $n$ is the number of items, $B$ is the known capacity of the knapsack, and $\delta$ is the grid resolution of the values (\Cref{theorem:query-upper-bound-KwC}). We work with scaled versions of the knapsack and packing problems, in which the value vector $\mathbf v \in [0,1]^n$ and the queried price vectors satisfy $\mathbf p \in [0,1]^n$. Under such constrained price vectors, an essentially necessary condition for recovery is $B<n$ (see \Cref{remark:bounded-assumption} in \Cref{section:notation}), and hence the stated query upper bound is at most quadratic in $n$.

We complement this algorithmic result by establishing that recovering the values in the fractional knapsack setting necessarily requires $\Omega(n \log(1/\delta))$ comparison queries (\Cref{query-lower-bound}). Since $B<n$ and $n<1/\delta$ (see \Cref{footnote:n-less-than-1/delta}), the upper and lower query bounds have a multiplicative gap of at most $O(n/\log n)$. Furthermore, the bounds are essentially tight whenever $B=o\left(\sqrt{n\log(1/\delta)}\right)$.

A key insight of our work is that, in the comparison-oracle model, fractional knapsack is essentially as general as packing linear programs. In particular, our algorithm for fractional knapsack can be used directly to solve the packing setting by treating a row of the constraint matrix $\mathbf{P}\in[0,1]^{m\times n}$ as the price vector $\mathbf{p}\in[0,1]^n$ and setting the remaining rows to zero (\Cref{thm:plco-upper-bound}). Moreover, the $\Omega(n\log(1/\delta))$ query lower bound continues to hold for packing linear programs as well (\Cref{thm:plco-lower-bound}). Consequently, the query upper bound obtained via this knapsack reduction is essentially the best possible for packing linear programs, up to an $O(n/\log n)$-factor gap.

Finally, we extend our knapsack algorithm to the profit-maximization setting, thereby obtaining a comparison-oracle analogue of the revealed-preference result of \cite{amin2015onlineProfitMax}. In the profit-maximization problem, a seller posts prices $\mathbf p \in [0,1]^n$ for $n$ divisible items to a buyer with unknown additive valuations $\mathbf v \in [0,1]^n$ and seeks to maximize the resulting profit. Given a price vector $\mathbf p$, the buyer chooses a value-maximizing bundle $\mathbf x^*(\mathbf p)$ subject to her budget---equivalently, solves a fractional knapsack problem---and the seller’s profit is then $\mathbf p^\top \mathbf x^*(\mathbf p) - \mathbf c^\top \mathbf x^*(\mathbf p)$, where $\mathbf c \in [0,1]^n$ is a known cost vector. We show that profit maximization can be solved in polynomial-time using $O(n \log(1/\delta) + B^2)$ comparison queries (\Cref{theorem:profit-maximization}).

\medskip

\noindent {\bf Our Techniques.} Our algorithm for fractional knapsack proceeds in three phases. Our algorithm for fractional knapsack proceeds in three phases. We begin by examining the optimal packing under the uniform price vector $\mathbf 1_n$. Utilizing the greedy packing structure of the optimal solution $\mathbf x^*(\mathbf 1_n)$, and querying one-coordinate perturbations of $\mathbf 1_n$, we partition the items into three sets: the integrally packed items $S$ under price vector $\mathbf 1_n$ along with the fractionally packed item $f$  and the unpacked items $U$ (\Cref{subsection:identify-tripartition}). We then recover the scaled values of items in $U\cup \{f\}$ using carefully chosen one-coordinate perturbations and utilizing how such perturbations affect the greedy packing order in optimal solutions (\Cref{subsection:Uvalues}). 

The technically most involved part is recovering the values of the items in $S$ (\Cref{subsection:Svalues}). Our key novel construct here is a family of \emph{triad prices}---parameterized by a scalar $\alpha$---that are designed so that the oracle’s response exhibits a sharp transition when $\alpha$ crosses the ratio $\widehat v / v_{\overline s}$, where $\widehat v$ is the value of the least-valued recovered item and $v_{\overline s}$ is the value of the least-valued unrecovered item. By combining a binary search over $\alpha$ with monotonicity properties of triad prices, we identify this threshold and thereby recover the scaled value of $\overline s$. Repeating this procedure recovers the scaled values of all items in $S$, and hence of the entire instance. Altogether, this yields a polynomial-time algorithm with query complexity $O(n\log(1/\delta)+B^2)$.

Our results crucially rely on understanding how the optimal fractional knapsack solution changes under carefully selected perturbations of the price vectors. As mentioned, a key technical contribution is the construction and analysis of triad prices, which allow us to convert comparison-oracle responses into a threshold test for unknown value ratios. Their design is intricate: the prices must simultaneously preserve enough structure of the (greedy) optimal solutions to isolate the least-valued unrecovered item ($\overline s$), while also ensuring that the induced threshold behavior is monotone and sharp enough to support an efficient binary search. 

After localizing the scaled values (i.e., value ratios) to sufficiently small intervals, we determine them exactly via known continued-fraction-based method for computing the simplest intervening rational; see \cite{murakami2010continued} and Appendix~\ref{appendix:Stern-Brocot}.

We obtain the query complexity lower bound for knapsack and packing via a novel reduction from the \textsc{Index} problem in the one-way communication model (\Cref{sec:lower_bound}).

\medskip
\noindent
{\bf Additional Related Work.} The problems studied in this work are related to the literature on inverse optimization; see, e.g.,~\cite{chan2025inverse} for a survey. In a standard inverse-optimization formulation for linear programs, Ahuja and Orlin~\cite{ahuja2001inverse} study the problem of minimally perturbing a given objective vector so that a prescribed feasible solution becomes optimal, and show that the resulting inverse problem can itself be formulated as a linear program when perturbations are measured in the \(\ell_1\) or \(\ell_\infty\) norm. Our comparison-oracle model is more parsimonious: rather than observing a solution itself, the learner receives only ordinal information comparing optimal solutions. Consequently, the result of Ahuja and Orlin~\cite{ahuja2001inverse} does not, on its own, yield a recovery algorithm in the comparison-query setting. Similarly, in the setting of fractional knapsack and profit maximization, the current model strengthens the revealed-preference framework of \cite{amin2015onlineProfitMax}: rather than observing the buyer's demand bundle at posted prices, the seller receives only ordinal information comparing optimal bundles. We also note that, unlike prior works~\cite{zadimoghaddam2012efficiently,beigman2006learning,balcan2014learning}, which study revealed preferences in statistical learning settings, our focus is on the query complexity of algorithms that can issue adaptive queries.

\section{Notation and Preliminaries}
\label{section:notation}
This work develops query efficient algorithms for fractional knapsack, packing linear programs, and profit maximization using comparison oracles. We will write $\mathbf{e}_j$ for the $j$-th standard basis vector in $\mathbb{R}^n$ and $\mathbf{1}_n$ for the all-ones vector.

In the fractional knapsack setup, the objective is to infer the a priori unknown values $v_1,\ldots, v_n \in (0,1]$ of $n$ items that can be fractionally packed in a knapsack of known capacity $B \in \RR_{>0}$.\footnote{Note that $B$ itself can be obtained through a single revealed-preference query, which returns the entire optimal fractional knapsack solution for given prices.} The algorithm has access only to comparison information about the optimal fractional solutions \(\mathbf x^*(\mathbf p)\) under different price vectors $\mathbf{p}$.

Specifically, for any price vector $\mathbf p=(p_1, \ldots, p_n) \in [0,1]^n$, where $p_i \in [0,1]$ denotes the price of item $i$, write \(\mathbf x^*(\mathbf p)\) to denote the
optimal solution returned by the standard greedy algorithm for the fractional knapsack problem: 
$%
    \max_{\mathbf x\in[0,1]^n} \ \mathbf v^\top \mathbf x
    \ \text{ s.t. } 
    \mathbf p^\top \mathbf x\le B $. %
That is, given $\mathbf p$, the optimal solution $\mathbf x^*(\mathbf p) \in [0,1]^n$ is obtained by ordering the items in decreasing order of their value-to-price ratios $v_i/p_i$ and selecting them greedily until the budget constraint becomes tight, possibly including the final item fractionally; ties in the value-to-price ratios are broken according to an unknown fixed ordering of the items.  An item $i$ is said to be integrally packed if $x^*_i(\mathbf p) = 1$, unpacked if $x^*_i(\mathbf p) = 0$, and fractionally packed if $0 < x^*_i(\mathbf p) <1$.

The algorithm has access to a comparison oracle that, when queried with any two price vectors $\mathbf p,\mathbf q\in[0,1]^n$, returns the relative ordering of the weighted total packings $\mathbf w^\top \mathbf x^*(\mathbf p)$ and $\mathbf w^\top \mathbf x^*(\mathbf q)$; that is, returns $\mathrm{sign}(\mathbf w^\top \mathbf x^*(\mathbf p) - \mathbf w^\top \mathbf x^*(\mathbf q))$, where $\mathbf w \in \mathbb{R}^n_{>0}$ is a fixed but unknown weight vector with positive components. Note that when $\mathbf w$ is the all-ones vector, i.e., $\mathbf{w}=\mathbf{1}_n$, then oracle reports which of the two optimal solutions, $\mathbf x^*(\mathbf p)$ and $\mathbf x^*(\mathbf q)$, has the larger sum of coordinates. Also, the special case $\mathbf{w} = \mathbf{v}$ case corresponds to an oracle that reports which of the corresponding optimal solutions attains the larger objective value. We write $\mathcal{O}_{\mathbf w}(\mathbf p,\mathbf q)$ for the output of the comparison oracle on price vectors $\mathbf p$ and $\mathbf q$, namely
$\mathcal{O}_{\mathbf w}(\mathbf p,\mathbf q) = \operatorname{sign} \left(\mathbf w^\top \mathbf x^*(\mathbf p) \ - \ \mathbf w^\top \mathbf x^*(\mathbf q) \right)$ %
Note that although the oracle's responses depend on $\mathbf w$, the algorithm itself does have the knowledge of $\mathbf{w}$. %

In this work, we address instances in which the item values are distinct, lie in the range $(0,1]$, and are integer multiples of a known parameter $\delta \in (0,1)$;\footnote{In fact, our results hold even if we have a bound on $\delta$ and its exact value is not known.} that is, $v_i \neq v_j$ for all $i \neq j$, and $v_i \in \delta \mathbb{Z}_{>0}$ for every $i$.\footnote{The distinct value condition necessitates that $\delta < 1/n.$\label{footnote:n-less-than-1/delta}} We now formally define the \emph{Knapsack using Comparison Oracles} problem.
\begin{definition}[Knapsack using Comparison Oracle (KCO)]
\label{definition:knapsack-with-comparison}
For an unknown value vector $\mathbf{v} \in (0,1]^n$ with distinct coordinates $v_i \in \delta \mathbb{Z}_{>0}$, a known knapsack capacity $B > 0$ satisfying $B \leq (n - 1) \left(\frac{\min_{i \in [n]} v_i}{\max_{i \in [n]} v_i} - \frac{\delta^2}{n^2}\right)$, and given query access to comparison oracle $\mathcal{O}_\mathbf{w}: [0,1]^n \times [0,1]^n \to \{-\mathsf 1, \mathsf 0, \mathsf 1\}$, the \emph{Knapsack using Comparison Oracle} ({KCO}) problem is to recover a vector $\widehat{\mathbf{v}}$ such that $\widehat{\mathbf{v}} = \mathbf{v}/\lambda$ for some scalar $\lambda > 0$.
\end{definition}
\begin{remark}
The recovery of \(\mathbf v\) is possible only up to a positive scaling factor. Indeed, for any \(\lambda>0\), replacing \(\mathbf v\) by \(\lambda \mathbf v\) does not alter the optimal solution for any \(\mathbf p\), since
$
    \argmax_{\mathbf x\in[0,1]^n:\mathbf p^\top \mathbf x\le B} \allowbreak
    \mathbf v^\top \mathbf x
    =
    \argmax_{\mathbf x\in[0,1]^n:\mathbf p^\top \mathbf x\le B}
    (\lambda\mathbf v)^\top \mathbf x.
$
Consequently, the comparison oracle produces identical responses under \(\mathbf v\) and \(\lambda\mathbf v\). Therefore, no algorithm with access only to comparison queries can distinguish between these two value vectors, and recovery is possible only up to positive scaling.
\end{remark}

We also note that the queried price vectors are constrained to lie in $[0,1]^n$. This convention was likewise adopted in \cite{amin2015onlineProfitMax} and excludes unrealistic scenarios in which the price of an item can be made arbitrarily large.

\begin{remark}
\label{remark:bounded-assumption}
In defining the KCO problem, we focus on instances satisfying two conditions. The first is that all item values are distinct. This is a non-degeneracy assumption that appears in prior work on knapsack and profit maximization~\cite{amin2015onlineProfitMax}. The second condition is that the knapsack capacity is essentially bounded by $n \frac{\min_i v_i}{\max_j v_j}$. In Appendix~\ref{appendix:bounded-capacity}, we show that this assumption is necessary in the following sense: if the bound is violated even by a factor of $(1+o(1))$, then there exist fractional knapsack instances for which no algorithm can recover the value vector $\mathbf v$ up to scaling using only the comparison oracle. This impossibility holds both when $\min_i v_i/\max_j v_j$ is large (in particular, $1-o(1)$) and when it is small (specifically, $o(1)$).
\end{remark}

\section{Fractional Knapsack using Comparison Oracle}
\label{section:knapsack-with-comparison}

This section develops a polynomial-time algorithm for the Knapsack using Comparison Oracle problem and establishes an upper bound on its query complexity. The main result of this section is the following theorem.

\begin{restatable}{theorem}{TheoremMainKCO} 
\label{theorem:query-upper-bound-KwC}
The Knapsack using Comparison Oracle problem admits a polynomial-time algorithm with query complexity $O(n \log(1/\delta) + B^2)$.  
\end{restatable}

Throughout this section, we assume, without loss of generality, that the knapsack capacity $B$ is non-integral. This assumption is used solely to ensure that, under the price vector $\mathbf{1}_n$, the optimal solution $\mathbf{x}^*(\mathbf{1}_n)$ necessarily contains a fractionally packed item. In particular, since $B \notin \mathbb{Z}$, the greedy procedure that computes $\mathbf{x}^*(\mathbf{1}_n)$ yields an item $f$ satisfying $0 < x_f^*(\mathbf{1}_n) < 1$. The assumption $B \notin \mathbb{Z}$ is indeed without loss of generality: if $B \in \mathbb{Z}$, we can instead work with the price vector $\rho \mathbf{1}_n$, where $\rho \in (0,1]$ is chosen arbitrarily close to $1$. In this case, $\mathbf{x}^*(\rho \mathbf{1}_n)$ contains a fractionally packed item; see \Cref{lemma:fractional-item-exists-at-uniform-size} in Appendix \ref{appendix:BZwlog}. Moreover, all arguments presented in this section for the price vector $\mathbf{1}_n$ continue to hold for $\rho \mathbf{1}_n$, with $\rho$ sufficiently close to $1$, hence it suffices to address the case $B \notin \mathbb{Z}$.
Also, our algorithm uses randomness solely for tie-breaking, its guarantees hold almost surely.

Our algorithm for proving \Cref{theorem:query-upper-bound-KwC} consists of three phases. 
\begin{itemize}
    \item The first phase (\Cref{alg:get-partition}) identifies the support of the optimal solution under the price vector $\mathbf{p}=\mathbf{1}_n$. Specifically, for the optimal solution $\mathbf{x}^*(\mathbf{1}_n)$, we determine---in a query-efficient manner---the set $S$ of integrally packed items, the unique fractionally packed item $f$, and the set $U$ of unpacked items. That is, $S= \left\{i \in [n] \mid \mathbf x_i^*(\mathbf{1}_n)=1 \right\}$, $U= \{j \in [n] \mid \mathbf x_j^*(\mathbf{1}_n)=0\}$, while $f \in [n]$ is the unique item for which $0<x_f^*(\mathbf{1}_n)<1$.\footnote{As noted above, the condition $B\notin \mathbb{Z}$ ensures the existence of a fractionally packed item $f$.}
     
    This phase builds on the insight that a small decrease in the price of an item $u \in U$ (i.e., an unpacked item) does not change the optimal solution (see \Cref{claim:item-in-U-iff-oracle-returns-equal}), whereas small decrements in prices of a packed or the fractional item changes the optimal solution in a way that is detectable via comparison queries (\Cref{claim:fractional-item-has-higher-oracle-value-than-items-in-S}). This phase is detailed in \Cref{alg:get-partition}, and \Cref{lemma:identification-of-sets-S-U-f} establishes that it correctly identifies the sets $S$, $U$, and the fractional item $f$.

    \item With the tripartition $S$, $U$, and $\{f\}$ in hand, the second phase (\Cref{alg:get-value-ratio-for-items-in-U}) recovers the values $v_u$ for all $u \in U$ up to a common scaling factor. For a given item $u \in U$, we gradually decrease its price $p_u$ while keeping the prices of all other items fixed. This increases its value-to-price ratio $v_u/p_u$. The first point at which this ratio exceeds that of the fractional item $f$ can be detected via comparison queries, as the optimal bundle changes precisely at that point. By decreasing $p_u$ in sufficiently small increments and utilizing the fact that the values lie on the $\delta$-grid, the ratio $v_u/v_f$ can be uniquely determined from the observed crossing. The details of this phase are given in \Cref{alg:get-value-ratio-for-items-in-U}, and its guarantee is provided by \Cref{lemma:estimating-profit-of-items-in-U}.

     \item The technically more involved third phase (\Cref{alg:get-value-ratio-of-items-in-S}) recovers the remaining value ratios for the integrally packed items $s \in S$. At this stage, the values of the items in $U \cup \{f\}$ are known up to a common scaling factor. The algorithm then proceeds iteratively, recovering the scaled values of the items in $S$ one at a time. In each iteration, the algorithm leverages the already recovered values of the items in $U \cup \{f\}$, together with the values obtained in previous iterations, to construct a carefully chosen {\it triad price vector} (\Cref{Definition:TriadPrices}) that identifies the least-valued item among the currently unresolved ones in $S$. Repeating this procedure eventually recovers the ratios $v_i / \min_{j \in [n]} v_j$ for all $i \in S$. The correctness of this phase is established in \Cref{lemma:performance-guarantee-of-alg-for-computing-value-ratio-S}.
\end{itemize}

Combining the three phases yields a recovery of the entire value vector up to a common scaling factor. As shown in the following sections, the first, second, and third phases use $O(n)$, $O(n\log(1/\delta))$, and $O(\lfloor B\rfloor\log(1/\delta) + \lfloor B\rfloor^2)$ queries, respectively. Summing these bounds establishes the claimed query complexity of the overall algorithm.

\subsection{Identifying the Optimal Solution Partition}
\label{subsection:identify-tripartition}
This section develops an algorithm (\Cref{alg:get-partition}) that, for the optimal solution $\mathbf{x}^*(\mathbf{1}_n)$ corresponding to the all-ones price vector, identifies---using $O(n)$ comparison queries---the set $S = \{i \in [n] \mid x_i^*(\mathbf{1}_n)=1\}$ of integrally packed items, the set $U = \{i \in [n] \mid x_i^*(\mathbf{1}_n)=0\}$ of unpacked items, and the unique fractionally packed item $f$ satisfying $0 < x_f^*(\mathbf{1}_n) < 1$.

A key observation is that, for a suitably small $\theta>0$, decreasing the price of a single item $i$ from $1$ to $1-\theta$ affects the optimal solution in a predictable manner determined by whether the item lies in $S$, $U$, or is the unique fractional item $f$. \Cref{lemma:opt-changes} quantifies this change in the optimal solution for $\theta \coloneqq \frac{1}{2}\min\{\delta,\ 1-B+\lfloor B\rfloor\}$. 

Using \Cref{lemma:opt-changes}, we obtain Corollaries~\ref{claim:item-in-U-iff-oracle-returns-equal} and~\ref{claim:fractional-item-has-higher-oracle-value-than-items-in-S}, which characterize the oracle responses induced by single-item price perturbations. %
\Cref{claim:item-in-U-iff-oracle-returns-equal} asserts that an item \(u\) belongs to \(U\) if and only if a $\theta$ decrease in its price leaves the oracle response unchanged. Note that \(\mathcal{O}_{\mathbf w}(\mathbf{1}_n,\mathbf{1}_n-\theta \mathbf e_u)\) returns the sign of \(\mathbf w^\top \mathbf x^*(\mathbf{1}_n)-\mathbf w^\top \mathbf x^*(\mathbf{1}_n-\theta \mathbf e_u)\). Therefore, whenever the optimal solution remains unchanged, the oracle's response is \(\mathsf{0}\). We use this characterization to find $U$ in Line \ref{line:find-U} of \Cref{alg:get-partition}. Furthermore, \Cref{claim:fractional-item-has-higher-oracle-value-than-items-in-S} enables us to identify the fractional item among the items in $[n] \setminus U = S \cup \{f\}$ by querying the oracle on the price vectors ${\bf 1}_n-\theta{\bf e}_i$ and ${\bf 1}_n-\theta{\bf e}_j$ for any distinct $i,j \in S \cup \{f\}$. The while-loop in the algorithm utilizes this characterization to recover $f$ and the set $S$. 

\begin{lemma}
\label{lemma:opt-changes}
Let $\theta=\frac{1}{2}\min\{\delta,\;1-B+\lfloor B\rfloor\}$. Then the optimal solution under the perturbed price vector $\mathbf{1}_n-\theta\mathbf{e}_i$ satisfies:
\begin{align*}
\mathbf{x}^*(\mathbf{1}_n-\theta\mathbf{e}_i)
=
\begin{cases}
\mathbf{x}^*(\mathbf{1}_n) & \ \ i\in U, \\ 
\mathbf{x}^*(\mathbf{1}_n) + \theta \ \mathbf{e}_f & \ \ i\in S, \\ 
\mathbf{x}^*(\mathbf{1}_n)+\widetilde{\theta} \ \mathbf{e}_f & \ \ i=f,
\end{cases}    
\end{align*}
where $\widetilde{\theta} = \frac{\theta(B-\lfloor B\rfloor)}{1-\theta}$.
\end{lemma}
\begin{proof}
Under the price vector $\mathbf{1}_n$, all items have identical prices. Consequently, the greedy construction of the optimal solution $\mathbf{x}^*(\mathbf{1}_n)$ implies that $v_s > v_f > v_u$ 
for every $s \in S$ and $u \in U$.\footnote{Recall that the item values are distinct.} We prove the lemma through a case-analysis. We have three cases depending on whether the item is in $U$, $S$, or is the fractional item $f$.

\noindent
\textbf{Case 1.} Consider any item $u \in U$ and note that for such an item we have $v_f > v_u$. We next show that, even after reducing the price of $u$ from $1$ to $1-\theta$, the fractional item $f$ remains ahead of $u$ in the value-to-price-ratio
ordering used by the greedy algorithm. Indeed, since $v_u < v_f$ and both values lie on the $\delta$-grid, we have $v_u \leq v_f-\delta \leq v_f(1-\delta)$, where the second inequality follows from $v_f \leq 1$. Therefore, $v_u/(1-\delta) \leq v_f$. Since $\theta < \delta$, it follows that $v_u/(1-\theta) < v_u/(1-\delta) \leq v_f$. Thus, decreasing the price of $u$ from $1$ to $1-\theta$ does not alter the position of $u$ relative to the fractional item $f$ in the value-to-price-ratio ordering. Consequently, the greedy ordering induced by $\mathbf{1}_n-\theta\mathbf{e}_u$ coincides with that of $\mathbf{1}_n$, and hence $\mathbf{x}^*(\mathbf{1}_n-\theta\mathbf{e}_u) = \mathbf{x}^*(\mathbf{1}_n)$.

\noindent
\textbf{Case 2.} Consider any item $i \in S$. Under the price vector $\mathbf{1}_n-\theta\mathbf{e}_i$, the value-to-price ratio of item $i$ increases from $v_i$ to $v_i/(1-\theta)$, whereas the
ratios of all other items remain unchanged. Since $i$ already precedes $f$ in the greedy ordering under $\mathbf{1}_n$, after this price perturbation, item $i$ continues to appear before $f$ in the greedy ordering. 

Hence, all items in $S$ continue to be packed integrally. However, under the price vector $\mathbf{1}_n-\theta\mathbf{e}_i$, item $i$ has price $1-\theta$ rather than $1$. Consequently, after 
packing all items in $S$, the remaining capacity becomes 
\begin{align}
B-\sum_{s\in S} p_s = B-(|S|-\theta) = B-\lfloor B\rfloor+\theta \label{ineq:resd-cap}    
\end{align}
The final equality uses the identity $|S|=\lfloor B \rfloor$, which follows from the greedy construction of the optimal solution. Since $\theta < 1-(B-\lfloor B\rfloor) = 1- x_f^*(\mathbf 1_n)$, the residual capacity $B-\lfloor B\rfloor+\theta$ remains strictly smaller than $1$. Hence, the next item in the greedy ordering is again $f$, which is packed fractionally with $ x_f^*(\mathbf 1_n-\theta\mathbf e_i) = B-\lfloor B\rfloor+\theta$. That is, $ x_f^*(\mathbf 1_n-\theta\mathbf e_i) = x_f^*(\mathbf 1_n)+\theta$, while the packing levels of all other items are unchanged: $x_a^*(\mathbf 1_n-\theta\mathbf e_i) = x_a^*(\mathbf 1_n)$ for all $a \neq f$. In other words, $\mathbf x^*(\mathbf 1_n - \theta \mathbf e_i) = \mathbf x^*(\mathbf 1_n) + \theta \mathbf e_f$.

\noindent
\textbf{Case 3.} Finally, consider the item $i=f$. Under the price vector $\mathbf{1}_n-\theta\mathbf{e}_f$, the value-to-price ratio of $f$  becomes $v_f/(1-\theta)$. Since $v_s>v_f$ for every $s\in S$ and $\theta<\delta$, the same grid-separation argument as above implies that $v_f/(1-\theta)<v_s$ for every $s\in S$. Hence, the greedy ordering remains unchanged: all items in $S$ are packed integrally before $f$ is considered. After packing the items in $S$, the remaining capacity is $B-\lfloor B\rfloor$, while the price of $f$ is $1-\theta$. Therefore,
\[
x_f^*(\mathbf 1_n-\theta\mathbf e_f) = \frac{B-\lfloor B\rfloor}{1-\theta} = x_f^*(\mathbf 1_n) + \frac{B-\lfloor B\rfloor}{1-\theta} - (B-\lfloor B\rfloor) = x_f^*(\mathbf 1_n) + \widetilde{\theta}.
\]
Here, the penultimate equality uses the fact that $x_f^*(\mathbf 1_n)=B-\lfloor B\rfloor$, whereas the final equality follows from the definition $\widetilde{\theta} = \frac{\theta(B-\lfloor B\rfloor)}{1-\theta}$. Furthermore, since $\theta < 1- \left( B-\lfloor B\rfloor \right)$, we have $x_f^*(\mathbf 1_n-\theta\mathbf e_f) = \frac{B-\lfloor B\rfloor}{1-\theta} <1$, and hence $f$ remains fractionally packed, and all items in $U$ remains unpacked. Therefore, $x^*_a(\mathbf 1_n - \theta \mathbf e_a) = x^*_a(\mathbf 1_n)$ for all $a \neq f$. Combining these equations, we obtain $\mathbf x^*(\mathbf 1_n - \theta \mathbf e_i) = \mathbf x^*(\mathbf 1_n) + \widetilde{\theta} \mathbf e_f$. This establishes the case $i=f$ and completes the proof of the lemma. 
\end{proof}

\Cref{claim:item-in-U-iff-oracle-returns-equal,claim:fractional-item-has-higher-oracle-value-than-items-in-S} hold for the choice $\theta = \frac{1}{2}\min\{\delta, 1 - B + \lfloor B \rfloor\}$.

\begin{corollary}
\label{claim:item-in-U-iff-oracle-returns-equal}
An item $u \in  U$ if and only if $\mathcal{O}_{\bf w}({\bf 1}_n, {\bf 1}_n - \theta {\bf e}_u) = \mathsf{0}$.
\end{corollary}
\begin{proof} 
From \Cref{lemma:opt-changes}, we have $\mathbf x^*(\mathbf 1_n - \theta \mathbf e_u) = \mathbf x^*(\mathbf 1_n)$ for every $u \in U$. Hence, 
\begin{align*}
    \mathbf w^\top \big(\mathbf x^*(\mathbf 1_n - \theta \mathbf e_u) - \mathbf x^*(\mathbf 1_n)\big) = 0~.
\end{align*}
For the converse direction, consider any item $i \notin U$, i.e., $i \in S \cup \{f\}$. By \Cref{lemma:opt-changes}, we have $\mathbf x^*(\mathbf 1_n-\theta\mathbf e_i) = \mathbf x^*(\mathbf 1_n)
+ \eta \mathbf e_f$, where $\eta=\theta$ if $i\in S$ and $\eta=\widetilde{\theta}$ if $i=f$. Therefore, 
\[\mathbf w^\top \bigl(\mathbf x^*(\mathbf 1_n-\theta\mathbf e_i) -\mathbf x^*(\mathbf 1_n) \bigr) = \mathbf w^\top(\eta\mathbf e_f) = \eta w_f \neq 0. \]
This completes the proof.
\end{proof}

\Cref{claim:item-in-U-iff-oracle-returns-equal} enables \Cref{alg:get-partition} to correctly identify the set $U$ using comparison queries. The next corollary shows how to identify the fractional item $f$ and, consequently, the set $S$.

\begin{corollary}
\label{claim:fractional-item-has-higher-oracle-value-than-items-in-S}
For any distinct items $i,j \in S \cup \{f\} = [n] \setminus U$, the oracle distinguishes the fractional item from the integrally packed ones as follows:
\begin{enumerate}
    \item[(i)] $\mathcal{O}_{\bf w}({\bf 1}_n-\theta{\bf e}_i, \ {\bf 1}_n-\theta{\bf e}_j)=\mathsf{-1}$ if and only if $i=f$ and $j\in S$.
    \item[(ii)] $\mathcal{O}_{\bf w}({\bf 1}_n-\theta{\bf e}_i, \ {\bf 1}_n-\theta{\bf e}_j)=\mathsf{1}$ if and only if $i\in S$ and $j=f$.
    \item[(iii)] $\mathcal{O}_{\bf w}({\bf 1}_n-\theta{\bf e}_i,\ {\bf 1}_n-\theta{\bf e}_j)=\mathsf{0}$ if and only if both $i, j \in S$.
\end{enumerate}
\end{corollary}
\begin{proof} 
For part (i) of the lemma, consider $i = f$ and $j \in S$. For the optimal solutions under the price vectors ${\bf 1}_n - \theta {\bf e}_i $ and ${\bf 1}_n - \theta {\bf e}_j$, we have 
\begin{align*}
    {\bf w}^\top \left( {\bf x}^*({\bf 1}_n - \theta {\bf e}_f) - {\bf x^*}({\bf 1}_n - \theta {\bf e}_j) \right) &= {\bf w}^\top \Bigl( {\bf x}^*({\bf 1}_n - \theta {\bf e}_f) - {\bf x^*}({\bf 1}_n) + {\bf x^*}({\bf 1}_n)  - {\bf x^*}({\bf 1}_n - \theta {\bf e}_j) \Bigr)\\
    & = {\bf w}^\top \Bigl( {\bf x}^*({\bf 1}_n - \theta {\bf e}_f) - {\bf x^*}({\bf 1}_n) \Bigr) \ + \ {\bf w}^\top \Bigl({\bf x^*}({\bf 1}_n)  - {\bf x^*}({\bf 1}_n - \theta {\bf e}_j) \Bigr) \\
    &= \mathbf w^\top (\widetilde{\theta} \mathbf e_f) - \mathbf w^\top (\theta \mathbf e_f) \tag{via~\Cref{lemma:opt-changes}}\\
    &= \theta w_f \cdot \left( \frac{B - \lfloor B \rfloor}{1 - \theta} - 1 \right) \tag{recall $\widetilde{\theta} = \frac{\theta (B - \lfloor B \rfloor)}{1 - \theta}$}  \\
    & < 0.
\end{align*}
The last inequality follows from the fact that $\theta < 1 - B + \lfloor B \rfloor$. Hence, $\mathcal O_{\bf w}({\bf 1}_n - \theta {\bf e}_f, {\bf 1}_n - \theta {\bf e}_j) = \mathsf{-1}$.

For part (ii), observe that $\mathrm{sign}(a) = - \mathrm{sign}(-a)$ for all $a \in \mathbb{R}$. Therefore, for any $i \in S$,
\begin{align*}
    \mathcal O_{\mathbf w}(\mathbf 1_n - \theta e_i, \mathbf 1_n - \theta \mathbf e_f) &= \mathrm{sign}\left(\mathbf w^\top \left(\mathbf x^*(\mathbf 1_n - \theta \mathbf e_i) - \mathbf x^*(\mathbf 1_n - \theta \mathbf e_f)\right)\right) \\
    &= - \mathrm{sign}\left(\mathbf w^\top \left(\mathbf x^*(\mathbf 1_n - \theta \mathbf e_f) - \mathbf x^*(\mathbf 1_n - \theta \mathbf e_i)\right)\right) \\
    &= - \mathcal O_{\mathbf w}(\mathbf 1_n - \theta e_f, \mathbf 1_n - \theta \mathbf e_i) \\
    &= \mathsf{1} \tag{by part (i)}
\end{align*}

Finally, for part (iii), consider any two items $i, j \in S$. By \Cref{lemma:opt-changes}, $\mathbf x^*(\mathbf 1_n - \theta \mathbf e_i) = \mathbf x^*(\mathbf 1_n) + \theta \mathbf e_f = \mathbf x^*(\mathbf 1_n - \theta \mathbf e_j)$. Therefore, $\mathbf w^\top (\mathbf x^*(\mathbf 1_n - \theta \mathbf e_i) - \mathbf x^*(\mathbf 1_n - \theta \mathbf e_j)) = 0$.
 
We have established the ``if'' direction of each part of the corollary. The corresponding ``only if'' directions follow immediately by exclusion. For instance, the ``only if'' direction of part~(i) follows
from the ``if'' directions of parts~(ii) and~(iii), and similarly for the remaining parts. This completes the proof of the corollary.
\end{proof}

\begin{algorithm}[ht]
\caption{Find tripartition $S$, $U$, and $\{f\}$ for the optimal solution $\mathbf{x}^*(\mathbf{1}_n)$}
\label{alg:get-partition}

\begin{algorithmic}[1]

\Require Items $[n]$, parameter $\theta \coloneqq \frac{1}{2}\min\{\delta,1-B+\lfloor B\rfloor\}$, and access to oracle $\mathcal{O}_{\mathbf w}(\cdot,\cdot)$.
\Ensure Sets $S,U$ and fractional item $f$.

\State Initialize subsets $S=\emptyset$ and $U=\emptyset$.

\State For each item $a\in[n]$, query price vectors $\mathbf 1_n$ and $\mathbf 1_n-\theta\mathbf e_a$. If the oracle response $\mathcal{O}_{\mathbf w}(\mathbf 1_n,\mathbf 1_n-\theta\mathbf e_a)=0$, then update $U\leftarrow U\cup\{a\}$.
\linelabel{line:find-U}

\While{$[n]\setminus(U\cup S)\neq\emptyset$}

    \State If $|[n]\setminus(U\cup S)|=1$, then set $f$ to be that unique item and exit, returning the tripartition $S$, $U$, and $f$.

    \State Otherwise, select any two items $a,b\in[n]\setminus(U\cup S)$ and based on the oracle response
    \(
    r\coloneqq
    \mathcal{O}_{\mathbf w}
    (\mathbf 1_n-\theta\mathbf e_a,
    \mathbf 1_n-\theta\mathbf e_b)
    \)
    perform one of the following three updates.

    \State If response $r=\mathsf 0$, then assign both items to $S$, i.e., update $S\leftarrow S\cup\{a,b\}$.

    \State If $r=\mathsf{-1}$, then set $f=a$, and return the tripartition with $f$, $U$, and $S\coloneqq[n]\setminus(U\cup\{f\})$.

    \State If $r=\mathsf 1$, then set $f=b$, and return the tripartition with $f$, $U$, and $S\coloneqq[n]\setminus(U\cup\{f\})$.

\EndWhile

\end{algorithmic}
\end{algorithm}

We establish the correctness of \Cref{alg:get-partition} in the following lemma.

\begin{lemma}
\label{lemma:identification-of-sets-S-U-f}
For the optimal solution $\mathbf x^*(\mathbf{1}_n)$ under the price vector $\mathbf{1}_n$, \Cref{alg:get-partition} identifies the set of integrally packed items $S = \{i \in [n]: x^*_i(\mathbf{1}_n) = 1\}$, the set of unpacked items $U = \{i \in [n]: x^*_i(\mathbf{1}_n) = 0\}$, and the fractionally packed item $f$ using $O(n)$ comparison queries.
\end{lemma}
\begin{proof}
The query complexity follows immediately. Line~\ref{line:find-U} performs at most $n$ comparison queries, and the while-loop executes $O(n)$ times, performing one query per
iteration. Therefore, \Cref{alg:get-partition} uses $O(n)$ comparison queries overall.

By \Cref{claim:item-in-U-iff-oracle-returns-equal}, $\mathcal{O}_{\bf w}({\bf 1}_n,\; {\bf 1}_n-\theta{\bf e}_a) =\mathsf{0}$ if and only if $a \in U$. Consequently, the set constructed in 
Line~\ref{line:find-U} is exactly $U$.

Once the algorithm has correctly identified the set $U$, the remaining items are precisely those in $S \cup \{f\} = [n]\setminus U$. The correctness of the while-loop then follows from \Cref{claim:fractional-item-has-higher-oracle-value-than-items-in-S}, which characterizes the oracle responses for pairs of items in $S \cup \{f\}$. Consequently, the while-loop correctly identifies the $f$ and $S$. This completes the proof of the lemma. 
\end{proof}

Having recovered the tripartition $S,U,\{f\}$, we now turn to recovering the item values up to a common scaling factor. This recovery is carried out in two subsequent phases. \Cref{subsection:Uvalues} recovers the scaled values of all items $u \in U$, while \Cref{subsection:Svalues} recovers the scaled values of the items $s \in S$.

\subsection{Recovering Scaled Values of Unpacked Items}
\label{subsection:Uvalues}
This section presents an algorithm (\Cref{alg:get-value-ratio-for-items-in-U}) for recovering, up to a scaling factor, the values of items that remain unallocated under the price vector $\mathbf 1_n$. An insight underlying our query-efficient algorithm is that, for a suitably chosen parameter $\gamma>0$, decreasing the price of an unpacked item $u\in U$ alone by sufficiently many multiplicative steps of $\gamma$ induces a predictable change in the optimal solution (\Cref{lemma:x*-under-price-1-t-gamma-e_u}). We make this result algorithmically useful by showing that  this change is detectable via comparison-oracle queries (\Cref{claim:for-items-in-U-equal-iff-t-less-than-t*}).

We now state and prove the results outlined above. Throughout this section, let %
\begin{align} 
\gamma_1 := \frac{\delta^2}{2}\left(1-\frac{\delta^2}{2}\right), \qquad \gamma_2 := \frac{\delta^2}{2}\left(1-\frac{\delta^2}{3}\right), \qquad \gamma \sim \operatorname{Unif}[\gamma_1,\gamma_2]. \end{align}

Our analysis focuses on prices of the form $1 - t \gamma$ for integer values of $t$. The following lemma shows that, for every item $u \in U$, the value ratio $v_u/v_f$ lies in one of the intervals induced by grid points $1 - t \gamma$, where $t$ ranges over an appropriately bounded set of integers. 

\begin{lemma}
\label{lemma:unique-t*-for-U-switching}
For every item $u \in U$, there exists a unique integer $t^*_u \in \left\{1, 2, \ldots, \lfloor \frac{2}{\delta^2} \rfloor \right\}$ such that the value ratio $\frac{v_u}{v_f}$ lies in the interval $\bigl(1 - t^*_u \gamma, \ 1 - (t^*_u - 1) \gamma\bigr]$.
\end{lemma}
\begin{proof}
Given that $u \in U$ is unpacked under the price vector $\mathbf 1_n$, we have $v_u < v_f$. Furthermore, since $v_u \ge \delta$ and $v_f \leq 1$, the ratio $v_u/v_f$ lies in the interval $[\delta,1)$. Hence, for the lemma, it suffices to consider grid points $1-t\gamma$ covering this range. To this end, using the lower bound $\gamma \geq \frac{\delta^2}{2}(1 - \frac{\delta^2}{2})$, we obtain  
\begin{align*}
    1 - \left\lfloor \frac{2}{\delta^2} \right\rfloor \gamma \leq 1 - \left(\frac{2}{\delta^2} - 1\right) \frac{\delta^2}{2} \left(1 - \frac{\delta^2}{2}\right) = 1-\left(1-\frac{\delta^2}{2}\right)^2<\delta, 
\end{align*}
where the final inequality uses the fact that $\delta \in (0,1)$. 
Hence, 
\begin{align*}
    [\delta, 1) \subset \left(1 - \left\lfloor \frac{2}{\delta^2}\right\rfloor \gamma, \ 1\right] = \bigsqcup_{t=1}^{\lfloor 2/\delta^2\rfloor} (1 - t \gamma, 1 - (t-1)\gamma].
\end{align*}
This containment guarantees the existence of an integer $t \in \left\{1,2,\ldots,\left\lfloor \frac{2}{\delta^2}\right\rfloor\right\}$ such that $\frac{v_u}{v_f} \in \bigl(1-t\gamma, \ 1-(t-1)\gamma\bigr]$. We define $t_u^*$ to be this integer $t$. The lemma stands proved. 
\end{proof}

Note that under the price vector $\mathbf 1_n - t_u^*\gamma \mathbf e_u$, the value-to-price ratio of item $u$, namely $\frac{v_u}{1-t_u^*\gamma}$, exceeds that of item $f$, which is equal to $v_f$. Furthermore, for every integer $t < t_u^*$, under the price vector $\mathbf 1_n - t\gamma \mathbf e_u$ the value-to-price ratio of item $u$ is at most that of item $f$. Consequently, when considering price vectors of the form $\mathbf 1_n - t \gamma \mathbf e_u$, the relative ordering of items $u$ and $f$ in the greedy construction changes precisely at $t=t_u^*$. The following lemma characterizes the resulting change in the optimal solution and establishes a monotonicity property that is crucial for query efficiency.

\begin{lemma}
\label{lemma:x*-under-price-1-t-gamma-e_u}
Let $u \in U$, and let $t^*_u$ be defined as in \Cref{lemma:unique-t*-for-U-switching}. Also let $\widehat{s} = \argmin_{s \in S} v_s$ be the least-valued item in $S$. Then, for every integer $t \in \left\{1, 2, \ldots, \lfloor \frac{2}{\delta^2} \rfloor \right\}$, we have 
\begin{align*}
    \mathbf x^*(\mathbf 1_n - t \gamma \mathbf e_u) 
    = \begin{cases}
        \mathbf x^*(\mathbf 1_n), & t < t^*_u\\
        \mathbf x^*(\mathbf 1_n) + \min \left\{1, \frac{B - \lfloor B \rfloor}{1 - t \gamma} \right\} \mathbf e_u + \eta \mathbf e_f, & t \geq t^*_u, \text{ and } 1 - t\gamma > {v_u}/{v_{\widehat{s}}} \\
        \mathbf x^*(\mathbf 1_n) +  \mathbf e_u + \min\{0, B - \lfloor B \rfloor + t\gamma - 1\} \mathbf e_{\widehat{s}} + \eta \mathbf e_f, & t \geq t^*_u, \text{ and } 1 - t\gamma < {v_u}/{v_{\widehat{s}}}~.
    \end{cases}
\end{align*}
where $\eta \coloneqq \max \{ \lfloor B \rfloor - B, \ t \gamma - 1\}$.
\end{lemma}

\begin{proof}
    Fix an item $u \in U$. We will establish this lemma by analyzing the three different cases.

    \noindent
    \textbf{Case 1}: $t < t^*_u$. The definition of $t_u^*$ implies that $\frac{v_u}{1-(t_u^*-1)\gamma} < v_f$. Therefore, for every integer $t \leq t_u^*-1$, we have $\frac{v_u}{1-t\gamma} < v_f$. It follows that, for all integers $t < t_u^*$, under the price vector $\mathbf 1_n - t\gamma \mathbf e_u$, the value-to-price ratio of item $u$ is strictly smaller than that of item $f$. Furthermore, all items other than $u$ have the same value-to-price ratios under the price vectors $\mathbf 1_n - t\gamma \mathbf e_u$ and $\mathbf 1_n$. Hence, the greedy ordering---and therefore the optimal solution---remains unchanged under the price vector $\mathbf 1_n - t\gamma \mathbf e_u$. Thus, for every integer $t \in \left\{1,2,\ldots,\left\lfloor \frac{2}{\delta^2}\right\rfloor\right\}$ satisfying $t < t_u^*$, we have $\mathbf x^*(\mathbf 1_n - t\gamma \mathbf e_u)=\mathbf x^*(\mathbf 1_n)$. 

     \medskip
     
    \noindent
    \textbf{Case 2}: $t \geq t_u^*$ and $1 - t\gamma > \frac{v_u}{v_{\widehat{s}}}$. Consider the price vector $\mathbf 1_n - t\gamma \mathbf e_u$. Since $t \geq t_u^*$, the value-to-price ratio of item $u$ satisfies $\frac{v_u}{1-t\gamma} \geq \frac{v_u}{1-t_u^*\gamma} > v_f$. In addition, the condition $1-t\gamma > {v_u}/{v_{\widehat{s}}}$ implies that $\frac{v_u}{1-t\gamma} < v_{\widehat{s}}$. Therefore, in the greedy construction of the optimal solution, item $u$ appears after all items in $S$ but before item $f$. Consequently, all items in $S$ are packed integrally before item $u$, while item $f$ is considered after item $u$.

    In particular, after all items in $S$ are packed, the remaining knapsack capacity is $B-\lfloor B\rfloor$. Therefore, the quantity of item $u$ packed is $\min\left\{1,\frac{B-\lfloor B\rfloor}{1-t\gamma}\right\}$. If item $u$ is packed integrally, the remaining capacity becomes $B-\lfloor B\rfloor-(1-t\gamma)=B-\lfloor B\rfloor+t\gamma-1$; otherwise, if $u$ is packed fractionally, the remaining capacity is $0$. Since both $B-\lfloor B\rfloor$ and $t\gamma$ are strictly less than $1$, we have $B-\lfloor B\rfloor+t\gamma-1<1$. As item $f$ has price $1$, the remaining capacity can be exhausted by packing at most a fractional amount of $f$. Hence, the quantity of item $f$ packed is $\max\{0,\,B-\lfloor B\rfloor+t\gamma-1\}$. Furthermore, all items in $U \setminus \{u\}$ remain unpacked. 
    
    Finally, recall that under the price vector $\mathbf 1_n$, the item $u$ is unpacked, and item $f$ is packed to the extent $B - \lfloor B \rfloor$. Therefore 
    \begin{align*}
        \mathbf x^*(\mathbf 1_n - t \gamma \mathbf e_u) = \mathbf x^*(\mathbf 1_n) + \min\left\{1, \frac{B - \lfloor B \rfloor}{1 - t \gamma}\right\} \mathbf e_u  + \max\{ \lfloor B \rfloor - B, \ t\gamma - 1\} \mathbf e_f~. 
    \end{align*}
    
  \medskip
  
    \noindent
    \textbf{Case 3:} $t \geq t^*_u$ and $1 - t \gamma < \frac{v_u}{v_{\widehat{s}}}$. Consider the price vector $\mathbf 1_n - t \gamma \mathbf e_u$. Since $\frac{v_u}{1 - t \gamma} > v_{\widehat{s}} > v_f$, item $u$ appears before both $\widehat{s}$ and $f$ in the greedy ordering. After all the items in $S \setminus \{\widehat{s}\}$ have been packed, the remaining knapsack capacity is $$B - \sum_{k \in S\setminus \{\widehat{s}\}} p_k = B - (|S| - 1) = B - \lfloor B \rfloor + 1.$$

    Since the price of item $u$ is $1 - t \gamma < 1$, item $u$ is packed integrally. Thus, after all the items in $\left( S \cup \{u \} \right) \setminus \{\widehat{s}\}$ have been packed integrally, the remaining capacity is  
    $$B - \lfloor B \rfloor + 1 - (1 - t\gamma) = B - \lfloor B \rfloor + t\gamma.$$
    
    Here, the quantity of item $\widehat{s}$ packed is $\min\{1, B - \lfloor B \rfloor + t\gamma\}$. Also, note that if $\widehat{s}$ is packed integrally, then the remaining capacity is $B - \lfloor B \rfloor + t\gamma - 1$, else if $\widehat{s}$ is packed fractionally the remaining capacity is $0$. That is, after packing $\widehat{s}$, the remaining capacity is $$\max\{0, B - \lfloor B \rfloor + t\gamma - 1\} $$ 
    
    This quantity is less than $1$; both $B - \lfloor B \rfloor$ and $t\gamma$ are less than $1$. Also, given that the price of item $f$ is $1$, this remaining capacity is filled by a (possibly zero) fraction of $f$. All items in $U \setminus \{u\}$ therefore remain unpacked.

    Finally, recall that under the price vector $\mathbf 1_n$, item $\widehat{s}$ is integrally packed, item $f$ is packed in quantity $B - \lfloor B \rfloor$ and item $u$ is unpacked. Comparing the two optimal solutions yields 
      \begin{align*}
        \mathbf x^*(\mathbf 1_n - t \gamma \mathbf e_u) = \mathbf x^*(\mathbf 1_n) + \mathbf e_u + \min\{0, B - \lfloor B \rfloor + t\gamma - 1\} \mathbf e_{\widehat{s}} + \max\{\lfloor B \rfloor - B, t\gamma - 1\} \mathbf e_f~.
    \end{align*}

    The lemma stands proved.
\end{proof}

The following corollary characterizes the oracle's response to price vectors of the form ${\bf 1}_n - t\gamma{\bf e}_u$, thereby enabling a binary search for the thresholds $t_u^*$ defined in \Cref{lemma:unique-t*-for-U-switching}.

\begin{corollary}
\label{claim:for-items-in-U-equal-iff-t-less-than-t*}
    For every item $u \in U$, let $t^*_u$ be as defined in \Cref{lemma:unique-t*-for-U-switching}. Then, for every integer $t \in \left\{1, 2, \ldots, \lfloor \frac{2}{\delta^2} \rfloor \right\}$, it holds that $\mathcal{O}_{\bf w}({\bf 1}_n, {\bf 1}_n - t \gamma {\bf e}_u) = \mathsf{0} $ if and only if $t < t^*_u$.
\end{corollary}
\begin{proof}
    Fix an item $u \in U$. For the ``if'' direction, \Cref{lemma:x*-under-price-1-t-gamma-e_u} ensures that $\mathbf x^*(\mathbf 1_n - t \gamma \mathbf e_u) = \mathbf x^*(\mathbf 1_n)$ for every $t < t^*_u$. Hence, $\mathcal O_{\bf w}(\mathbf 1_n, \mathbf 1_n - t \gamma \mathbf e_u) = \mathrm{sign}\bigl(\mathbf w^\top\mathbf x^*(\mathbf 1_n) \ - \ \mathbf w^\top\mathbf x^*(\mathbf 1_n - t \gamma \mathbf e_u)\bigr) = 0,$ as claimed.

    For the ``only if'' direction, note that for $t \geq t^*_u$, two cases arise: either $1 - t \gamma > {v_u}/{v_{\widehat{s}}}$ or $1 - t \gamma < {v_u}/{v_{\widehat{s}}}$, where $\widehat{s} = \argmin_{s \in S} v_s$. Also, let  $\eta = \max\{\lfloor B \rfloor - B, t\gamma - 1\}$.

    \noindent
    \textbf{Case 1}: $1 - t \gamma > \frac{v_u}{v_{\widehat{s}}}$. \Cref{lemma:x*-under-price-1-t-gamma-e_u} give us 
    \begin{align*}
        \mathbf w^\top \big(\mathbf x^*(\mathbf 1_n) - \mathbf x^*(\mathbf 1_n - t \gamma \mathbf e_u)) &= - \mathbf w^\top \left( \min\left\{1, \frac{B - \lfloor B \rfloor}{1 - t\gamma}\right\} \mathbf e_u + \eta \mathbf e_f \right) \\
        &= - w_u \min\left\{1, \frac{B - \lfloor B \rfloor}{1 - t\gamma}\right\} - \eta w_f \\
        &= \begin{cases}
            - w_u + w_f(1 - t \gamma), \qquad\qquad \text{if } 1 - t \gamma < B - \lfloor B \rfloor,\\
            -\left(\frac{w_u}{1 - t \gamma} - w_f\right) (B - \lfloor B \rfloor), \quad\text{otherwise}.
        \end{cases}
    \end{align*}
    The above quantity is \emph{nonzero} almost surely: since $\gamma$ is a continuous random variable, for every integer $t \in \{1,\ldots, \lfloor\frac{2}{\delta^2}\rfloor\}$ we have $1 - t\gamma \neq \frac{w_u}{w_f}$ with probability one. 
    Therefore, $\mathbf w^\top \big(\mathbf x^*(\mathbf 1_n) - \mathbf x^*(\mathbf 1_n - t \gamma \mathbf e_u)) \neq 0$, which implies that $\mathcal{O}_{\bf w}({\bf 1}_n, {\bf 1}_n - t \gamma {\bf e}_u) \neq \mathsf{0}$.

    \medskip
    \noindent
    \textbf{Case 2}: $1 - t \gamma < \frac{v_u}{v_{\widehat{s}}}$. By \Cref{lemma:x*-under-price-1-t-gamma-e_u},
    \begin{align*}
        \mathbf w^\top \big(\mathbf x^*(\mathbf 1_n) - \mathbf x^*(\mathbf 1_n - t \gamma \mathbf e_u)) &= - \mathbf w^\top \left(\mathbf e_u + \min\{0, B - \lfloor B \rfloor + t\gamma - 1\} \mathbf e_{\widehat{s}} + \eta \mathbf e_f\right) \\
        &= - w_u - w_{\widehat{s}} \min\{0, B - \lfloor B \rfloor + t\gamma - 1\} - \eta w_f \\
        &= \begin{cases}
             - w_u + w_f (1 - t \gamma),\quad\qquad\qquad\qquad\qquad\quad \text{if } 1 - t \gamma < B - \lfloor B \rfloor,\\
            - w_u + (w_f - w_{\widehat{s}}) (B - \lfloor B \rfloor) - w_{\widehat{s}} (1 - t \gamma),\qquad\quad\text{otherwise}.
        \end{cases}
    \end{align*}
    The above quantity is nonzero almost surely. Indeed, in either case it is an affine function of $1-t\gamma$ with a nonzero coefficient. Since $\gamma$ is a continuous random variable, the probability that this expression vanishes for any integer $t \in \{1,\ldots, \lfloor\frac{1}{\delta^2}\rfloor\}$ is zero. Consequently, $\mathbf w^\top \bigl(\mathbf x^*(\mathbf 1_n) - \mathbf x^*(\mathbf 1_n - t\gamma \mathbf e_u)\bigr) \neq 0$ almost surely, and therefore $\mathcal{O}_{\bf w}({\bf 1}_n, {\bf 1}_n - t \gamma {\bf e}_u) \neq \mathsf{0}$.

    This completes the proof of the corollary.
\end{proof}

Recall that the distinct item values $v_1,\ldots,v_n \in (0,1]$ lie on the $\delta$-grid, i.e., each value $v_i = p \delta$ for some integer $p \in  \left\{1, 2, \ldots,  \lfloor \frac{1}{\delta} \rfloor \right\}$. Hence, any two distinct values differ by at least $\delta$. The following proposition shows that any two distinct ratios of these values are separated by at least $\delta^2$.

\begin{proposition}
\label{lemma:unique-p/q-in-interval-of-length-delta^2}
Let integer \(N=\lfloor \frac{1}{\delta} \rfloor\) and define $\mathcal R_\delta = \left\{ \frac{p}{q} : \text{integers } p,q \in \{1,2, \ldots, N\} \right\}$. Then, any interval of length less than $\delta^2$ contains at most one element of $\mathcal R_\delta$. 
\end{proposition}
\begin{proof}
Consider two distinct elements of \(\mathcal R_\delta\), say $\frac{p}{q}\neq \frac{p'}{q'},$ where \(p,q,p',q'\in \{1,\ldots, N\} \). Then
 $\left|\frac{p}{q}-\frac{p'}{q'}\right|
    =
    \frac{|pq'-p'q|}{qq'}.
    $
Since \(\frac{p}{q}\neq \frac{p'}{q'}\), the integer \(pq'-p'q\) is nonzero, so
$
    |pq'-p'q|\ge 1.
$
Also \(q,q'\le N\), and hence \(qq'\le N^2\). Therefore
$
    \left|\frac{p}{q}-\frac{p'}{q'}\right|
    \ge
    \frac{1}{N^2}
    \geq
    \delta^2.
$
Thus, any two distinct elements of \(\mathcal R_\delta\) are separated by distance at least \(\delta^2\). Hence, any interval of length less than $\delta^2$ can contain at most one element of \(\mathcal R_\delta\). This completes the proof. 
\end{proof}

\begin{algorithm}[ht]
\caption{Recovering scaled values of items in $U$}
\label{alg:get-value-ratio-for-items-in-U}

\begin{algorithmic}[1]

\Require Items $U$ and parameter $\gamma>0$.
\Ensure Value ratios $\frac{v_u}{v_f}$ for every item $u \in U$.

\For{every item $u \in U$}
\State Perform a binary search over $t\in\{1,2,\ldots,\lfloor 2/\delta^2\rfloor\}$ using comparison queries with price vectors $\mathbf 1_n$ and $\mathbf 1_n-t\gamma\mathbf e_u$. Let $t^*_u$ be the smallest value of $t$ for which $\mathcal O_{\mathbf w}(\mathbf 1_n,\ \mathbf 1_n-t\gamma\mathbf e_u)\neq \mathsf{0}$.

\State Find the unique rational number $\frac{p}{q}$ in the interval $(1-t^*_u \ \gamma_2,\, 1-(t^*_u-1)\gamma_1)$ satisfying $p,q\in\{1,\ldots,\lfloor1/\delta\rfloor\}$ by invoking the simplest intervening rational algorithm, which runs in $O(\operatorname{poly}\log(1/\delta))$ time (Appendix \ref{appendix:Stern-Brocot}). Set $\frac{v_u}{v_f}=\frac{p}{q}$.
\EndFor
\State \textbf{Return} the scaled values $\frac{v_u}{v_f}$ for every $u\in U$.
\end{algorithmic}
\end{algorithm}

\begin{restatable}{lemma}{FinalLemmaUValue}
\label{lemma:estimating-profit-of-items-in-U}
\Cref{alg:get-value-ratio-for-items-in-U} correctly recovers the scaled values $\frac{v_u}{v_f}$ for all items $u \in U$. Furthermore, the algorithm requires at most $O(n \log(1/\delta))$ comparison-oracle queries and runs in time polynomial in $n$ and $\log(1/\delta)$.
\end{restatable}

\begin{proof}
The algorithm iterates over all items in $U$. Fix an iteration corresponding to an item $u \in U$. By \Cref{lemma:unique-t*-for-U-switching}, there exists a unique integer $t^*_u \in \{1,\ldots,\lfloor 2/\delta^2 \rfloor\}$ such that $\frac{v_u}{v_f} \in \left(1-t^*_u\gamma,\ 1-(t^*_u-1)\gamma\right]. $

Moreover, by \Cref{claim:for-items-in-U-equal-iff-t-less-than-t*}, $\mathcal O_{\mathbf w}(\mathbf 1_n,\mathbf 1_n-t\gamma\mathbf e_u)=\mathsf{0}$ if and only if $t<t^*_u$. Therefore, a binary search over $\{1,\ldots,\lfloor 2/\delta^2 \rfloor\}$ for the smallest value of $t$ satisfying $\mathcal O_{\mathbf w}(\mathbf 1_n,\mathbf 1_n-t\gamma\mathbf e_u)\neq \mathsf{0}$ correctly identifies $t^*_u$. This requires $O(\log(1/\delta))$ comparison-oracle queries.

Note that $\frac{v_u}{v_f}=\frac{p'}{q'}$ for some integers $p',q'\in\{1,\ldots,\lfloor 1/\delta\rfloor\}$, and this ratio lies in the interval $\left(1-t_u^*\gamma,\ 1-(t_u^*-1)\gamma\right]$. To ensure that the interval endpoints have polynomial bit complexity, we instead consider the interval $\left(1-t_u^*\gamma_2,\ 1-(t_u^*-1)\gamma_1\right]$, which contains $\left(1-t_u^*\gamma, \ 1-(t_u^*-1)\gamma\right]$ since $\gamma\in[\gamma_1,\gamma_2]$. The enlarged interval still has length at most $\delta^2$:
\begin{align*}
    \left(1-(t^*_u-1)\gamma_1\right)
    -
    \left(1-t^*_u\gamma_2\right) 
    &= \gamma_1+t^*_u(\gamma_2-\gamma_1) \\
    &\le
    \gamma_1+\frac{2}{\delta^2}(\gamma_2-\gamma_1) \\
    &=
    \frac{\delta^2}{2}\left(1-\frac{\delta^2}{2}\right)
    +
    \frac{2}{\delta^2}\cdot \frac{\delta^4}{12} \\
    &=    \delta^2\left(\frac{2}{3}-\frac{\delta^2}{4}\right)
    <
    \delta^2.
\end{align*}
Hence, by \Cref{lemma:unique-p/q-in-interval-of-length-delta^2}, we have that $\frac{v_u}{v_f}$ is the unique ratio of the form $\frac{p}{q}$ with $p,q \in \{1,\ldots,\lfloor 1/\delta \rfloor\}$ contained in this interval. Consequently, the unique ratio returned by the $O(\operatorname{poly}\log(1/\delta))$-time simplest intervening rational algorithm (Appendix~\ref{appendix:Stern-Brocot}) is precisely $\frac{v_u}{v_f}$. Therefore, the algorithm correctly recovers $\frac{v_u}{v_f}$ for every $u \in U$.

We complete the proof by analyzing the query complexity and running time of the algorithm. Each iteration performs $O(\log(1/\delta))$ comparison-oracle queries and runs in $O(\operatorname{poly} \log(1/\delta))$ time. Since the algorithm executes one iteration for each item in $U$, with $|U| \leq n$, the total number of comparison-oracle queries is $O(n\log(1/\delta))$, and the overall running time is also $O(n\log(1/\delta))$. The lemma stands proved. 
\end{proof}

\subsection{Recovering Scaled Values of Packed Items}
\label{subsection:Svalues}

This section presents the third and final phase (\Cref{alg:get-value-ratio-of-items-in-S}) for recovering the scaled values of items in $S$, which are packed integrally under the price vector $\mathbf 1_n$. While the queried price vectors in the previous sections were relatively easy to state, the price vectors for this phase are intricate. This section is organized as follows. First, we introduce a family of price vectors, termed \emph{triad prices}, parameterized by a scalar, and establish their key properties (\Cref{subsection:triad}). Next, we characterize the oracle's responses to queries involving triad prices (\Cref{sec:characterization-of-oracle-response-under-triad-prices}). Finally, in \Cref{subsection:alg-for-packed}, we use these characterizations to obtain \Cref{alg:get-value-ratio-of-items-in-S}.

Throughout this section, we work with the perturbation parameters 
\begin{align} \beta := \frac{\delta^3}{3n^2}, \qquad \varepsilon_1 := \frac{\delta^8}{6n^4}, \qquad \varepsilon_2 := \frac{\delta^8}{4n^4}, \qquad \varepsilon \sim \operatorname{Unif}[\varepsilon_1,\varepsilon_2]. \label{eq:defn-beta} 
\end{align}

\subsubsection{Triad Prices}
\label{subsection:triad}
Towards recovering the scaled values of items in $S$, we first define a family of price vectors. In particular, for specific item subsets $R \subseteq [n]$ and scalars $\alpha \in (0,1)$, triad price vectors are constructed so that every item in $R$ has one of three possible value-to-price ratios: $\frac{\widehat{v}}{\alpha - 2 \varepsilon}$, $\frac{\widehat{v}}{\alpha-\varepsilon}$, $\frac{\widehat{v}}{\alpha}$, where $\widehat{v} \coloneqq \min_{j \in [n]} v_j$. Our algorithm begins with $R=U$, a set of items whose scaled values are known, and iteratively recovers the values of the remaining items.

\begin{definition}[Triad Prices]
\label{Definition:TriadPrices}
Let $R =\{r_1, r_2, \ldots, r_{|R|} \}$ be a set of items indexed in decreasing order of their (known scaled) values, and let $\alpha \in (0,1)$. Further, let $\widehat{v} \coloneqq \min_{j \in [n]} v_j$. Define $t$ to be the smallest index satisfying $$\sum_{i=1}^{t} \min \ \left\{ 1, (\alpha-2\varepsilon)\frac{v_{r_i}}{\widehat v} \right\} \ge B-n+|R|,$$ provided such an index exists. 

With $R_1 \coloneqq \{r_1, \ldots, r_t\}$, $k_\alpha \coloneqq r_{t+1}$, and $R_2 \coloneqq \{r_{t+2}, \ldots, r_{|R|}\}$, the triad price vector $\mathbf p(\alpha) \in [0,1]^n$ is defined by 
    \begin{align*}
        p_i(\alpha) = 
            \begin{cases}
                1, & i \in [n] \setminus R,\\
                \min\{1, (\alpha - 2\varepsilon)\frac{v_i}{\widehat{v}}\}, &i \in R_1, \\
                \min\{1, (\alpha - \varepsilon) \frac{v_i}{\widehat{v}}\}, &i = k_\alpha,\\
                \min\{1, \alpha \frac{v_i}{\widehat{v}}\}, & i \in R_2. 
            \end{cases}
    \end{align*}
The item $k_\alpha$ is called the triad item.
\end{definition}
\medskip
Note that if no index $t < |R|$ satisfies the threshold condition, then we say that the triad price vector at $\alpha$ does not exist.\footnote{This may occur when $\alpha$ is sufficiently small so that
$\sum_{r \in R} \min \ \left\{1, (\alpha-2\varepsilon)\frac{v_r}{\widehat v} \right\} < B-n+|R|$, and hence the threshold $B-n+|R|$ is never attained.} Moreover, by construction, items in $R_1$, $\{k_\alpha\}$, and $R_2$ have value-to-price ratios $\frac{\widehat{v}}{\alpha - 2 \varepsilon}$, $\frac{\widehat{v}}{\alpha-\varepsilon}$, $\frac{\widehat{v}}{\alpha}$, respectively, except when the corresponding price is capped at $1$. For notational convenience, the dependence of $\mathbf p(\alpha)$ on $R$ is left implicit.

\Cref{lemma:computation-of-triad-prices} below notes that the triad prices can be computed efficiently and \Cref{lemma:monotonicity-of-triad-prices} provides a useful monotonicity property of triad prices.\footnote{Recall that parameter $\beta = \frac{\delta^3}{3n^2}$.} 
\begin{lemma}
\label{lemma:computation-of-triad-prices} %
Let $R \subseteq [n]$ be a set of items with known scaled values and suppose that $\arg\min_{j\in[n]} v_j \in R$. Then, for every $\alpha \in \{\beta, 2\beta, \ldots, 1\}$, one can determine in polynomial time whether the triad price vector $\mathbf p(\alpha)$ exists. Furthermore, whenever it exists, the vector $\mathbf p(\alpha)$ and the associated triad item $k_\alpha$ can both be computed in polynomial time.
\end{lemma}

\begin{lemma}
\label{lemma:monotonicity-of-triad-prices}
For any fixed set $R \subseteq [n]$, if the triad price vector $\mathbf p(\widehat{\alpha})$ exists for some $\widehat{\alpha}\in(0,1)$, then $\mathbf p(\alpha)$ also exists for every $\alpha\in[\widehat{\alpha},1)$.
\end{lemma}

The proofs of \Cref{lemma:computation-of-triad-prices,lemma:monotonicity-of-triad-prices} are deferred to \Cref{sec:proof-of-computation-of-triad-prices}. Having defined the triad prices, we next develop a characterization of the oracle response that will enable us to recover the values of the items in $S$.

\subsubsection{Characterization of Oracle Responses under Triad Prices}
\label{sec:characterization-of-oracle-response-under-triad-prices}
Our algorithm maintains a set $R \subseteq [n]$ of items whose scaled values are known\footnote{Initially, $R=U \cup \{f\}.$ } and, in each iteration, recovers the scaled value of the minimum-valued item outside $R$. That is, it recovers the scaled value of $\overline{s} = \arg\min_{s\in[n]\setminus R} v_s$. Hence, by construction, the algorithm will maintain the invariant that $R$ contains the lowest-valued items: $\max_{r \in R} v_r < \min_{s \in [n] \setminus R} v_s$. 

Recall that $\widehat v \coloneqq \min_{j \in [n]} v_j$ and that the parameter $\beta$ (defined in \Cref{eq:defn-beta}) satisfies $\beta < \delta^2$. Consequently, by an argument analogous to that of \Cref{lemma:unique-t*-for-U-switching}, one can construct intervals of length $\beta$ that uniquely isolate the value ratios $\frac{\widehat v}{v_s}$ for all items $s \in [n]\setminus R$.

\begin{proposition}
\label{lemma:unique-alpha*-for-S-switching}
For any set of items $R \subseteq [n]$ and every item $s \in [n] \setminus R$, there is a unique scalar $\alpha^*_s \in \{\beta, 2 \beta, \ldots, 1 \}$ such that the value ratio $\frac{\widehat{v}}{v_s}$ belongs to the interval $\left[\alpha^*_s - \frac{3\varepsilon}{2}, \alpha^*_s + \beta - \frac{3\varepsilon}{2} \right)$. 
\end{proposition}
\begin{proof}
Since the item values $v_i \in (0,1]$ lie on the $\delta$-grid, we have $\frac{\widehat{v}}{v_s} \in [\delta, 1]$ for every $s \in [n] \setminus R$. Moreover, since $\beta < \delta^2/3$ and $\varepsilon < \delta^3/3$, the interval $[\delta, 1]$ is covered by the collection of $\beta$-length intervals: $[\delta, 1] \subset \bigcup_{\alpha' \in \{\beta, 2 \beta, \ldots, 1 \} } \ \left[\alpha' - \frac{3\varepsilon}{2}, \alpha' + \beta - \frac{3\varepsilon}{2} \right)$. 

For each $s \in [n]\setminus R$, let $\alpha_s^* \in \{\beta,2\beta,\ldots,1\}$ denote the unique scalar whose corresponding interval contains the ratio $\frac{\widehat v}{v_s}$; that is, $\frac{\widehat v}{v_s} \in \left[\alpha^*_s - \frac{3\varepsilon}{2}, \alpha^*_s + \beta - \frac{3\varepsilon}{2} \right)$. The uniqueness of $\alpha_s^*$ follows from the distinctness of the item values. Analogous to \Cref{lemma:unique-p/q-in-interval-of-length-delta^2}, for any distinct items $s$,$s'$, the ratios $\frac{\widehat v}{v_s}$ and $\frac{\widehat v}{v_{s'}}$ are separated by at least $\delta^2$. Since $\beta < \delta^2$, no interval of length $\beta$ can contain more than one such ratio. Therefore, each ratio $\frac{\widehat v}{v_s}$ belongs to a unique interval, and hence determines a unique scalar $\alpha_s^*$. This completes the proof.
\end{proof}
Following \Cref{lemma:unique-alpha*-for-S-switching}, we write throughout $\alpha^*_{\overline{s}}$ for the unique scalar associated with the value ratio $\frac{\widehat v}{v_{\overline{s}}}$, where
$\overline{s} = \arg\min_{s\in[n]\setminus R} v_s$ is the minimum-valued item outside $R$.

As $\alpha$ varies, the optimal solution $\mathbf x^*(\mathbf p(\alpha))$ may exhibit a complicated dependence on $\alpha$, making it difficult to obtain closed-form expressions for $\mathbf x^*(\mathbf p(\alpha))$. Nevertheless, the following lemma yields a key structural guarantee that will subsequently enable us to determine $\alpha^*_{\overline{s}}$, and hence the unique interval containing the value ratio of $\overline{s}$. The lemma essentially states that, for scalar $\alpha$, under a suitably small perturbation of the triad price vector, the triad item $k_\alpha$ is packed in the optimal solution if and only if $\alpha \le \alpha^*_{\overline{s}}$.

\begin{lemma}
\label{lemma:k-packed-only-at-alpha-lesser-than-optimal}
Let $R \subseteq [n]$ be a set satisfying $\max_{r \in R} v_r < \min_{s \in [n] \setminus R} v_s$ and $\overline{s} = \arg\min_{s\in[n]\setminus R} v_s$. Further, suppose that for $\alpha \in \{\beta, 2\beta, \ldots, 1\}$ the triad price vector $\mathbf p(\alpha)$ exists, and let $k_\alpha$ be the corresponding triad item. Then
$$x^*_{k_\alpha} \left(\mathbf p(\alpha)- \frac{\varepsilon}{2} \frac{v_{k_\alpha}}{\widehat{v}} \mathbf e_{k_\alpha} \right) = 0$$ if and only if $\alpha > \alpha^*_{\overline{s}}$.
\end{lemma}

\begin{proof}
Fix the set $R$ and scalar $\alpha$, and let $R_1$, $R_2$, and $k_\alpha$ be as in Definition~\ref{Definition:TriadPrices}. Since $\mathbf p(\alpha)$ exists, the total price of the items in $R_1$ under $\mathbf p(\alpha)$ is at least $B-n+|R|$. We divide the proof into two cases, depending on whether $\alpha$ is at most or at least $\alpha^*_{\overline{s}}$. 

The defining property of $\alpha^*_{\overline{s}}$ (\Cref{lemma:unique-alpha*-for-S-switching}) gives us $\frac{\widehat v}{v_{\overline{s}}} \in \left[\alpha^*_{\overline{s}}-\frac{3\varepsilon}{2},\ \alpha^*_{\overline{s}}-\frac{3\varepsilon}{2}+\beta \right)$. Equivalently, $\alpha^*_{\overline{s}} \le \frac{\widehat v}{v_{\overline{s}}}+\frac{3\varepsilon}{2} < \alpha^*_{\overline{s}}+\beta$. Furthermore, for any $\alpha \in \{\beta,2\beta,\ldots,1\}$, exactly one of the following holds:
\begin{itemize}
    \item If $\alpha \le \alpha^*_{\overline{s}}$, then $ \alpha < \frac{\widehat v}{v_{\overline{s}}} + \frac{3\varepsilon}{2}$. The inequality is strict almost surely, since $\alpha$ is restricted to the grid $\{\beta,2\beta,\ldots,1\}$ and $\varepsilon$ is a continuous random variable.
    \item If $\alpha > \alpha^*_{\overline{s}}$, then $\alpha \ge \alpha^*_{\overline{s}}+\beta$, and therefore $\alpha > \frac{\widehat v}{v_{\overline{s}}} + \frac{3\varepsilon}{2}$. 
\end{itemize}

We now address the two cases corresponding to $\alpha > \alpha^*_{\overline{s}}$ and $\alpha \le \alpha^*_{\overline{s}}$.

\noindent
\noindent\textbf{Case {\rm I}:} $\alpha>\alpha^*_{\overline{s}}$. As mentioned above, here we have $\alpha > \frac{\widehat{v}}{v_{\overline{s}}} + \frac{3\varepsilon}{2}$. For such an $\alpha$, consider the price vector $\mathbf{q} \coloneqq \mathbf p(\alpha)-\frac{\varepsilon}{2} \frac{v_{k_\alpha}}{ \widehat{v}} \mathbf e_{k_\alpha}$. We will show that, under $\mathbf{q}$, the value-to-price ratio of $k_\alpha$ is strictly smaller than that of every item in $R_1 \cup \left([n] \setminus R \right)$. Hence, all of these items are considered before $k_\alpha$ in the greedy construction of the optimal solution $\mathbf{x}^*(\mathbf{q})$ and $k_\alpha$ remains unpacked in the optimal solution.

Specifically, under $\mathbf{q}$, the value-to-price ratio of $k_\alpha$ is
\begin{align}
\frac{v_{k_\alpha}}{\min\{1, (\alpha-\varepsilon)v_{k_\alpha}/\widehat{v}\} - \varepsilon v_{k_\alpha}/2 \widehat{v}} = \max\left\{\frac{v_{k_\alpha}}{(1 - \varepsilon v_{k_\alpha}/2 \widehat{v})},\ \frac{\widehat{v}}{(\alpha - 3\varepsilon/2)}\right\} \label{eq:k-alpha-v2p}    
\end{align}
Note that the first term in the maximum satisfies
\begin{align}
    \frac{v_{k_\alpha}} {1-\frac{\varepsilon v_{k_\alpha}}{2\widehat v}} \le v_{k_\alpha} \left(1+\frac{\varepsilon v_{k_\alpha}}{\widehat v} \right) < v_{k_\alpha}+\delta \label{ineq:first-term}
\end{align}
Here, the first inequality follows from the bound $1/(1-x)\le 1+2x$ for $x\le 1/2$, while the second uses $\varepsilon \le \delta^8/(3n^4)$ together with $v_{k_\alpha}/\widehat v \le 1/\delta$.

\Cref{ineq:first-term,eq:k-alpha-v2p} give us the following upper bound on the value-to-price ratio of $k_\alpha$ under $\mathbf{q}$:
\begin{align}
    \max\left\{v_{k_\alpha}+\delta, \ \frac{\widehat{v}}{(\alpha - 3\varepsilon/2)}\right\} \label{eq:ub-v2p}
\end{align}
Now, for each item $r \in R_1$, the value-to-price ratio under $\mathbf{q}$ is
\begin{align}
    \frac{v_r}{ \min \left\{1, (\alpha - 2\varepsilon)\frac{v_r}{\widehat{v}} \right\} } = \max \left\{v_r, \ \frac{\widehat{v}}{\alpha - 2\varepsilon}  \right\} \label{eq:r-v2p}
\end{align}
Since $v_r \geq v_{k_\alpha} + \delta$ for every $r \in R_1$, comparing \Cref{eq:r-v2p,eq:ub-v2p} shows that the value-to-price of $k_\alpha$ is strictly less than that of $r$. 

Every item $s \in [n]\setminus R$ has price $1$ under $\mathbf q$. Consequently, its value-to-price ratio is $v_s$, which is at least $v_{\overline{s}}$, by the definition of $\overline{s}$. Consider the upper bound on the value-to-price ratio of $k_\alpha$ given in \Cref{eq:ub-v2p}. For the second term in the maximum, via the bound $\alpha > \frac{\widehat v}{v_{\overline{s}}} + \frac{3\varepsilon}{2}$, we obtain
\begin{align}
    \frac{\widehat v}{\alpha - {3\varepsilon}/{2}} < v_{\overline{s}} \label{ineq:too-many}
\end{align}
 In addition, using the inequality $v_{\overline{s}} \geq v_{k_\alpha} + \delta$, the upper bound in \Cref{eq:ub-v2p} we obtain that the value-to-price ratio of every item in $[n]\setminus R$ is strictly greater than that of $k_\alpha$. 

The above observations imply that all items in $R_1 \cup ([n]\setminus R)$ are considered before $k_\alpha$ in the greedy construction of the optimal solution $\mathbf x^*(\mathbf q)$. Furthermore, these items have the same prices under both the triad price vector $\mathbf p(\alpha)$ and the perturbed price vector $\mathbf q$. Therefore,
\begin{align*}
\sum_{i \in R_1 \cup ([n]\setminus R)} q_i &= \sum_{i \in R_1 \cup ([n]\setminus R)} p_i(\alpha)\\
&= \sum_{i \in [n]\setminus R} 1 + \sum_{i \in R_1} p_i(\alpha)\\
&\ge n-|R| + (B-n+|R|) \tag{$\sum_{r \in R_1} p_r(\alpha)\ge B-n+|R|$ by construction} \\
&=B.
\end{align*}

Hence, the knapsack capacity is already exhausted by the items in $R_1 \cup ([n]\setminus R)$ before $k_\alpha$ is considered. Consequently, no additional item can be packed, and in particular $k_\alpha$ remains unpacked under $\mathbf q = \mathbf p(\alpha) - \frac{\varepsilon}{2} \frac{v_{k_\alpha}}{\widehat v}\mathbf e_{k_\alpha}.$ Therefore, $x^*_{k_\alpha}\ \left( \mathbf p(\alpha) - \frac{\varepsilon}{2} \frac{v_{k_\alpha}}{\widehat v}\mathbf e_{k_\alpha} \right) = 0$, as claimed.

\medskip

\noindent\textbf{Case {\rm II}:} $\alpha\le \alpha^*_{\overline{s}}$. Here, as previously mentioned, we have $\alpha < \frac{\widehat{v}}{v_{\overline{s}}} + \frac{3\varepsilon}{2}$. We will show that, in this case, the value-to-price ratio of $k_\alpha$ is greater than that of $\overline{s}$ and every item in $R_2$. It will follow that the triad item is packed for such values of $\alpha$.

In the current case, consider the perturbed price vector $\mathbf q' \coloneqq \mathbf p(\alpha) - \frac{\varepsilon}{2} \frac{v_{k_\alpha}}{\widehat v}\mathbf e_{k_\alpha}$. The price of the triad item $k_\alpha$ under $\mathbf q'$ is 
\begin{align}
    \min \ \left\{ 1,\, \left(\alpha- \varepsilon \right)\frac{v_{k_\alpha}}{\widehat v} \right\}  \ - \ \frac{\varepsilon}{2} \frac{v_{k_\alpha}}{\widehat{v}}  = \left(\alpha-\frac{3\varepsilon}{2}\right) \frac{v_{k_\alpha}}{\widehat v} \label{eq:k-a-price}
\end{align}
Indeed,
\begin{align}
    \left(\alpha- \varepsilon \right) \frac{v_{k_\alpha}}{\widehat v} \le \left(\alpha^*_{\overline{s}}- \varepsilon \right) \frac{v_{k_\alpha}}{\widehat v} \le \left( \frac{\widehat v}{v_{\overline{s}}} + \frac{\varepsilon}{2} \right) \ \frac{v_{k_\alpha}}{\widehat v} = \frac{v_{k_\alpha}}{v_{\overline{s}}} + \frac{\varepsilon v_{k_\alpha}}{2 \widehat{v}} \leq 1 - \delta + \frac{\delta^7}{8n^4} < 1  \label{ineq:rep-derive}
\end{align}
The second inequality follows from the defining property of $\alpha^*_{\overline{s}}$. The final inequality uses $v_{k_\alpha} \leq v_{\overline{s}} - \delta$, $v_{\overline{s}} \leq 1$, $\varepsilon \leq \delta^8/(4n^4)$ and $v_{k_\alpha}/\widehat{v} \leq 1/\delta$.

Since the price of $k_\alpha$ under $\mathbf q'$ is given by \Cref{eq:k-a-price}, its value-to-price ratio under $\mathbf q'$ is $\widehat{v}/(\alpha - 3\varepsilon/2)$.

Furthermore, for bounding the prices of items $r' \in R_2$ note that in the current case we have 
{\allowdisplaybreaks
\begin{align*}
\alpha \ \frac{v_{r'}}{\widehat v} &\le \left( \frac{\widehat v}{v_{\overline{s}}} + \frac{3\varepsilon}{2} \right) \frac{v_{r'}}{\widehat v} \\
&= \frac{v_{r'}}{v_{\overline{s}}} + \frac{3\varepsilon v_{r'}}{2\widehat v}\\
&\le  \frac{v_{r'}}{v_{\overline{s}}} +\frac{3\delta^7}{8n^4} \tag{$\varepsilon \leq \delta^8/4n^4$ and $v_{r'}/\widehat{v} \leq 1/\delta$} \\
& \leq 1 - \delta +\frac{3\delta^7}{8n^4} \tag{$v_{r'} + \delta \leq v_{\overline{s}} \leq 1 $} \\
&<1.
\end{align*}
}
Hence, in the current case, every item $r' \in R_2$ has price $ \min\ \left\{1,\alpha\frac{v_{r'}}{\widehat v}\right\} = \alpha \frac{v_{r'}}{\widehat v}$, and therefore value-to-price ratio $\widehat v/\alpha$. 

The above-mentioned observations imply that the value-to-price ratio of $k_\alpha$ (specifically, $\widehat{v}/(\alpha - 3\varepsilon/2)$) is greater than that of every item in $R_2$.

Moreover, the inequality $\alpha < \frac{\widehat{v}}{v_{\overline{s}}} + \frac{3\varepsilon}{2}$ implies $\widehat{v}/(\alpha - 3\varepsilon/2) > v_{\overline{s}}$. That is, the value-to-price ratio of $k_\alpha$ is greater than that of $\overline{s}$ as well. 

Consequently, $k_\alpha$ precedes $\overline{s}$ and all items in $R_2$ in the greedy ordering used to construct the optimal solution $\mathbf x^*(\mathbf q')$. 

Let $H \coloneqq [n]\setminus \bigl(R \cup \{\overline{s}\}\bigr)$. We note that the total price of the items in $H \cup R_1$ does not exceed the capacity $B$:
\begin{align} 
\sum_{i \in R_1 \cup H} q'_i &= \sum_{i \in R_1 \cup H} p_i(\alpha) \nonumber \\ 
&\le (n-|R|-1) + \sum_{i \in R_1} p_i(\alpha) \nonumber \\ 
&< (n-|R|-1) + (B-n+|R|+1) \nonumber \\ 
&= B \label{eq:H-R1-budget}. 
\end{align}
The strict inequality follows from the minimality of $R_1$ in the construction of the triad partition (Definition~\ref{Definition:TriadPrices}).
Indeed, since $R_1$ is the smallest prefix whose total price reaches at least $B-n+|R|$, removing its last item yields a prefix whose total price is strictly below $B-n+|R|$. Since every item price is at most $1$, it follows that $\sum_{i \in R_1} p_i(\alpha) < (B-n+|R|)+1$.

Consider the optimal solution under the price vector $\mathbf q' = \mathbf p(\alpha) - \frac{\varepsilon}{2} \frac{v_{k_\alpha}}{\widehat v}\mathbf e_{k_\alpha}. $ Even if all items in $H \cup R_1$ are packed before $k_\alpha$, the knapsack capacity is not exhausted; see \Cref{eq:H-R1-budget}. Moreover, as shown above, neither the item $\overline{s}$ nor any item in $R_2$ is considered before $k_\alpha$ in the greedy ordering. Therefore, $k_\alpha$ must be packed when it is encountered, and hence $x^*_{k_\alpha}\!\left( \mathbf p(\alpha) - \frac{\varepsilon v_{k_\alpha}}{2\widehat v}\mathbf e_{k_\alpha} \right) > 0$, as claimed.

The lemma stands proved.
\end{proof}

We next prove that, via comparison queries and at a price vector $\mathbf p$, one can determine whether an item is packed in the optimal solution by applying a sufficiently small perturbation that preserves value-to-price ratios.

\begin{lemma}
\label{lemma:oracle-response-characterization-general-one-coordinate-change}
Let $\mathbf p$ be a price vector such that $\mathbf p^\top \mathbf{1}_n > B$. For any $i \in [n]$, let $\tau > 0$ be any sufficiently small scalar so that the 
ordering of the items by value-to-price ratio remains unchanged under $\mathbf q^{(i)} \coloneqq \mathbf p - \tau \mathbf e_i$. Then $\mathcal{O}_{\mathbf w}(\mathbf p, \mathbf q^{(i)}) \neq \mathsf 0$ if and only if $x^*_i(\mathbf p) > 0$.
\end{lemma}

\begin{proof}
Fix an item $i \in [n]$ and a scalar $\tau>0$. We prove the lemma by considering separately the cases $x_i^*(\mathbf{p})=0$ and $x_i^*(\mathbf{p}) >0$. 

First, suppose that $x_i^*(\mathbf p)=0$. Then, in the construction of the optimal solution $\mathbf x^*(\mathbf p)$, the capacity $B$ is exhausted before item $i$ is considered in the greedy ordering, namely while packing a subset of the items that precede $i$ in the value-to-price ordering. Since the prices of these items are identical under $\mathbf p$ and $\mathbf q^{(i)}$, the same subset of items is packed to the same extent in the construction of $\mathbf x^*(\mathbf q^{(i)})$. Consequently, $\mathbf x^*(\mathbf p)=\mathbf x^*(\mathbf q^{(i)})$, and therefore $\mathcal O_{\mathbf w}(\mathbf p,\mathbf q^{(i)})=\mathsf 0$. 

Next, we consider the case $x_i^*(\mathbf p)>0$. Here, we will show that $\mathbf x^*(\mathbf q^{(i)}) \geq \mathbf x^*(\mathbf p)$ component-wise, and that the inequality is strict in at least one coordinate. Since the value-to-price ordering is identical under $\mathbf p$ and $\mathbf q^{(i)}$, and the prices of all items preceding $i$ are unchanged, these items are packed to exactly the same extent in the constructions of $\mathbf x^*(\mathbf p)$ and $\mathbf x^*(\mathbf q^{(i)})$. Since the price of item $i$ is strictly lower under $\mathbf q^{(i)}$, either $i$ is packed to a greater extent (if it was initially packed fractionally) and the capacity is exhausted, or the remaining capacity after processing $i$ is strictly larger under $\mathbf q^{(i)}$ than under $\mathbf p$. It follows inductively that every subsequent item is packed to at least the same extent under $\mathbf q^{(i)}$ as under $\mathbf p$. Therefore, $\mathbf x^*(\mathbf q^{(i)}) \ge \mathbf x^*(\mathbf p)$ component-wise. Furthermore, the additional capacity available after processing item $i$ must eventually be allocated to some item since $\mathbf p^\top \mathbf{1}_n > B$, implying that the above inequality is strict in at least one coordinate.

The analysis of the two cases shows that, as claimed, $\mathcal{O}_{\mathbf w}(\mathbf p, \mathbf q^{(i)}) \neq \mathsf 0$ if and only if $x^*_i(\mathbf p) > 0$. The lemma stands proved.  
\end{proof}

Next, using \Cref{lemma:k-packed-only-at-alpha-lesser-than-optimal,lemma:oracle-response-characterization-general-one-coordinate-change}, we prove the following characterization of the oracle response under triad prices.
\begin{lemma}
\label{corollary:orcale-response-at-alpha-lesser-than-optimal}
Let set $R \subseteq [n]$ satisfy $\max_{r \in R} v_r < \min_{s \in [n] \setminus R} v_s$ and let $\overline{s} = \arg\min_{s\in[n]\setminus R} v_s$. Suppose that, for $\alpha \in \{\beta, 2\beta, \ldots, 1\}$, the triad price vector $\mathbf p(\alpha)$ exists, and let $k_\alpha$ denote the corresponding triad item. Define the price vectors
$$
\mathbf q^1(\alpha) \coloneqq \mathbf p(\alpha) - \frac{\varepsilon}{2} \frac{v_{k_\alpha}}{\widehat v}\,\mathbf e_{k_\alpha}, \quad \mathbf{q}^2(\alpha) \coloneqq \mathbf p(\alpha) - \frac{5\varepsilon}{8} \frac{v_{k_\alpha}}{\widehat v}\,\mathbf e_{k_\alpha}, \quad \mathbf{q}^3(\alpha) \coloneqq \mathbf p(\alpha) - \frac{6\varepsilon}{8} \frac{v_{k_\alpha}}{\widehat v}\,\mathbf e_{k_\alpha}~.
$$
The following statements hold:
\begin{itemize} 
\item[(i)]  If $\alpha > \alpha^*_{\overline{s}}+\beta$, then $\mathcal O_{\mathbf w}\!\left(\mathbf q^1(\alpha),\mathbf q^2(\alpha)\right) = \mathcal O_{\mathbf w}\ \left(\mathbf q^2(\alpha),\mathbf q^3(\alpha)\right) = \mathsf 0$. 
\item[(ii)] If $\alpha \le \alpha^*_{\overline{s}}$, then at least one of $\mathcal O_{\mathbf w}  \left(\mathbf q^1(\alpha),\mathbf q^2(\alpha)\right)$ and $\mathcal O_{\mathbf w}  \left(\mathbf q^2(\alpha),\mathbf q^3(\alpha)\right)$ is not equal to $\mathsf{0}$.
\end{itemize}
\end{lemma}

\begin{proof}
The price vectors $\mathbf q^1(\alpha)$, $\mathbf q^2(\alpha)$, and $\mathbf q^3(\alpha)$ differ from $\mathbf p(\alpha)$ only in the coordinate corresponding to $k_\alpha$. Specifically, for every $j \in [n]\setminus\{k_\alpha\}$ and every $t \in \{1,2,3\}$, we have $q_j^t(\alpha)=p_j(\alpha)$, while $ q_{k_\alpha}^3(\alpha) < q_{k_\alpha}^2(\alpha) < q_{k_\alpha}^1(\alpha) < p_{k_\alpha}(\alpha). $

The proof is divided into two parts, corresponding to statements~(i) and~(ii) in the lemma.

\medskip
\noindent
{\bf Part (i)}: $\alpha > \alpha^*_{\overline{s}} + \beta$. In this regime, we show that, under each of the price vectors $\mathbf q^1(\alpha)$, $\mathbf q^2(\alpha)$, and $\mathbf q^3(\alpha)$, the value-to-price ratio of $k_\alpha$ is strictly smaller than that of every item in $([n]\setminus R)\cup R_1$. Furthermore, by \Cref{lemma:k-packed-only-at-alpha-lesser-than-optimal}, we have $x^*_{k_\alpha}(\mathbf q^1(\alpha))=0$. That is, in the optimal solution under $\mathbf q^1(\alpha)$, the capacity is exhausted by items preceding $k_\alpha$ in the value-to-price ordering. Since the prices of all items other than $k_\alpha$ are identical across the three price vectors, the same conclusion holds under $\mathbf q^2(\alpha)$ and $\mathbf q^3(\alpha)$ as well. Consequently, $\mathbf x^*(\mathbf q^1(\alpha)) = \mathbf x^*(\mathbf q^2(\alpha)) = \mathbf x^*(\mathbf q^3(\alpha))$, and therefore $\mathcal O_{\mathbf w}(\mathbf q^1(\alpha),\mathbf q^2(\alpha)) = \mathcal O_{\mathbf w}(\mathbf q^2(\alpha),\mathbf q^3(\alpha)) = \mathsf 0$.

Here, the fact that the value-to-price ratio of $k_\alpha$ is strictly smaller than that of every item in $([n]\setminus R)\cup R_1$ follows from arguments analogous to those used in Case~I of the proof of \Cref{lemma:k-packed-only-at-alpha-lesser-than-optimal}. Specifically, under the three price vectors the value-to-price ratio of $k_\alpha$ is at least 
\begin{align}
\frac{v_{k_\alpha}} { \min \left\{1, (\alpha-\varepsilon) \frac{v_{k_\alpha}}{\widehat{v}} \right\} - \frac{6\varepsilon}{8} \frac{v_{k_\alpha}}{\widehat{v}} } = \max\left\{\frac{v_{k_\alpha}}{ \left(1 - \frac{6\varepsilon}{8} \frac{v_{k_\alpha}}{\widehat{v}}\right) },\ \frac{\widehat{v}}{ \left(\alpha - \frac{7\varepsilon}{4} \right)}\right\}   
\end{align}
The first term in the maximum satisfies
   $\frac{v_{k_\alpha}} {1-\frac{6\varepsilon}{8} \frac{v_{k_\alpha}}{2\widehat v}} \le v_{k_\alpha} \left(1+\frac{3\varepsilon}{2} \frac{v_{k_\alpha}}{\widehat v} \right) < v_{k_\alpha}+\delta$. 
Hence, we obtain the following upper bound on the value-to-price ratio of $k_\alpha$ under the three price vectors: 
\begin{align}
    \max\left\{v_{k_\alpha}+\delta, \ \frac{\widehat{v}}{ \left(\alpha - \frac{7\varepsilon}{4} \right)} \right\} \label{eq:ub-v2pa}
\end{align}

For every item $r \in R_1$, the price is identical under $\mathbf q^1(\alpha)$, $\mathbf q^2(\alpha)$, and $\mathbf q^3(\alpha)$. Hence, its value-to-price ratio is the same under all three price vectors and is given by
\begin{align}
\frac{v_r}{\min  \left\{1,\;(\alpha-2\varepsilon)\frac{v_r}{\widehat v}\right\}} = \max \left\{v_r,\;\frac{\widehat v}{\alpha-2\varepsilon}\right\} \label{eq:r-v2pa}
\end{align}

Since $v_r \geq v_{k_\alpha} + \delta$ for every $r \in R_1$, comparing \Cref{eq:r-v2pa,eq:ub-v2pa} shows that the value-to-price of $k_\alpha$ is strictly less than that of $r$ under the three price vectors. 

Every item $s \in [n]\setminus R$ has price $1$ under $\mathbf q^1(\alpha)$, $\mathbf q^2(\alpha)$, and $\mathbf q^3(\alpha)$. Consequently, its value-to-price ratio is $v_s$, which is at least $v_{\overline{s}}$, by the definition of $\overline{s}$. Consider the upper bound on the value-to-price ratio of $k_\alpha$ given in \Cref{eq:ub-v2pa}. In the current case, since $\alpha \geq \alpha^*_{\overline{s}} + 2\beta$, we have $ \alpha > \frac{\widehat v}{v_{\overline{s}}} + \frac{3\varepsilon}{2} + \beta$. Hence, $\alpha$ exceeds $\frac{\widehat v}{v_{\overline{s}}}+\frac{3\varepsilon}{2}$ by a margin of at least $\beta$. Since $\varepsilon \ll \beta$, it follows that $\alpha > \frac{\widehat v}{v_{\overline{s}}} + \frac{7\varepsilon}{4}$, and therefore $\frac{\widehat v}{\alpha - {7\varepsilon}/{4}} < v_{\overline{s}}$. Combining this inequality with the upper bound in \Cref{eq:ub-v2pa} and the inequality $v_{\overline{s}} \geq v_{k_\alpha}+\delta$, we conclude that the value-to-price ratio of every item in $[n]\setminus R$ is strictly greater than that of $k_\alpha$ under each of the three price vectors.

Therefore, under each of the three price vectors, the items in $([n]\setminus R)\cup R_1$ are packed before $k_\alpha$ and exhaust the capacity, so that $k_\alpha$ is never reached in the construction of the optimal solution. Since the prices of these items are identical across the three price vectors, the same subset of items is packed to the same extent in all three cases. Hence, $ \mathbf x^*(\mathbf q^1(\alpha)) = \mathbf x^*(\mathbf q^2(\alpha)) = \mathbf x^*(\mathbf q^3(\alpha)), $ and we have $ \mathcal O_{\mathbf w}(\mathbf q^1(\alpha),\mathbf q^2(\alpha)) = \mathcal O_{\mathbf w}(\mathbf q^2(\alpha),\mathbf q^3(\alpha)) = \mathsf 0. $

\medskip
\noindent
{\bf Part (ii)}: $\alpha \leq \alpha^*_{\overline{s}}$. We fist note that, in this regime, the price of $k_\alpha$ is not capped at $1$ under any of the three price vectors. In particular, as in \Cref{ineq:rep-derive}, we have $(\alpha-\varepsilon)\frac{v_{k_\alpha}}{\widehat v}<1$. Hence, the price of $k_\alpha$ under $\mathbf q^1(\alpha)$, $\mathbf q^2(\alpha)$, and $\mathbf q^3(\alpha)$ is, respectively, $\left(\alpha-\frac{3\varepsilon}{2}\right)\frac{v_{k_\alpha}}{\widehat v}$, \ $\left(\alpha-\frac{13\varepsilon}{8}\right)\frac{v_{k_\alpha}}{\widehat v}$, and $\left(\alpha-\frac{14\varepsilon}{8}\right)\frac{v_{k_\alpha}}{\widehat v}$. Hence, the value-to-price ratio of $k_\alpha$ is strictly increasing from $\mathbf q^1(\alpha)$ to $\mathbf q^3(\alpha)$:
\begin{align}
\frac{\widehat v}{\alpha-\frac{3\varepsilon}{2}} < \frac{\widehat v}{\alpha-\frac{13\varepsilon}{8}} < \frac{\widehat v}{\alpha-\frac{14\varepsilon}{8}} \label{ineq:ratios}
\end{align} 
Also, note that these ratios are within $\delta$ of each other. Indeed, 
\begin{align}
    \frac{\widehat{v}}{\alpha - 14\varepsilon/8} - \frac{\widehat{v}}{\alpha - 3\varepsilon/2} = \widehat{v} \left( \frac{\varepsilon/4}{ (\alpha - 14\varepsilon/8) (\alpha - 3\varepsilon/2)} \right)
    \leq \frac{\delta^8/16n^4}{(\beta - 7\delta^8/16n^4)^2} 
    \leq 9\delta^2/4
    < \delta \label{ineq:close-delta}
\end{align}
Here, we use the bounds $\varepsilon \leq \frac{\delta^8}{4n^4}$ and $\alpha \geq \beta = \delta^3/3n^2$.

Towards completing the proof, and in preparation for invoking \Cref{lemma:oracle-response-characterization-general-one-coordinate-change}, we establish the claim below.
\begin{claim}
\label{claim:one-works}
The value-to-price ordering of the items remains unchanged either between $\mathbf q^1(\alpha)$ and $\mathbf q^2(\alpha)$ or between $\mathbf q^2(\alpha)$ and $\mathbf q^3(\alpha)$.
\end{claim}
\begin{proof}
Considering the value-to-price ratio of $k_\alpha$ in $\mathbf q^3(\alpha)$ we define $ L\coloneqq \left\{ s \in [n]\setminus R : v_s < \frac{\widehat v}{\alpha-\frac{14\varepsilon}{8}} \right\}$; since $ v_{\overline{s}} < \frac{\widehat v}{\alpha-\frac{14\varepsilon}{8}}$, we have $\overline{s}\in L$, and therefore $L \subseteq [n] \setminus R$ is nonempty. 

We analyze order preservation by separately considering the following four subsets of items: (1) $R_2$, (2) $R_1$, (3) $[n]\setminus (R\cup L)$, and (4) $L$. We first show that the value-to-price ordering is identical under $\mathbf q^1(\alpha)$, $\mathbf q^2(\alpha)$, and $\mathbf q^3(\alpha)$ for all items outside $L$.
\medskip

\noindent
(1) For the items in $R_2$, Case~II of the proof of \Cref{lemma:k-packed-only-at-alpha-lesser-than-optimal} shows that their value-to-price ratios are smaller than that of $k_\alpha$ under $\mathbf q^1(\alpha)$. Moreover, the value-to-price ratios of the items in $R_2$ remain unchanged across the three price vectors, whereas that of $k_\alpha$ increases. Hence, every item in $R_2$ follows $k_\alpha$ in the value-to-price ordering under all three price vectors, and the relative ordering among the items in $R_2$ is preserved.

\noindent
(2) The value-to-price ratio of every item in $R_1$ is $ \frac{\widehat v}{\alpha-2\varepsilon}, $ which, by \Cref{ineq:ratios}, exceeds the value-to-price ratios of $k_\alpha$ under each of the three price vectors. Hence, every item in $R_1$ precedes $k_\alpha$ in the value-to-price ordering under all three price vectors. Moreover, since the value-to-price ratios of the items in $R_1$ remain unchanged, their relative ordering is preserved as well.

\noindent
(3) For every item $s \in [n]\setminus (R\cup L)$, the price is equal to $1$ under each of the three price vectors, and hence its value-to-price ratio is equal to $v_s$. By the definition of $L$, we have
$v_s \geq \frac{\widehat v}{\alpha-\frac{14\varepsilon}{8}},$ which exceeds the value-to-price ratio of $k_\alpha$ under each of the three price vectors. Hence, every item in $[n]\setminus (R\cup L)$ precedes $k_\alpha$ in the value-to-price ordering under all three price vectors. Moreover, since the value-to-price ratios of these items remain unchanged, their relative ordering is preserved as well.

\noindent (4) It remains to analyze the set $L$. Let $ \ell^* \coloneqq \arg\max_{\ell\in L} v_\ell$. Note that every item in $L$ has price $1$ under each of the three price vectors, and hence its value-to-price ratio is equal to its value. We proceed by considering the possible relationships between $v_{\ell^*}$ and the value-to-price ratio of $k_\alpha$ given in \Cref{ineq:ratios}. First, suppose that $ v_{\ell^*} < \frac{\widehat v}{\alpha-\frac{3\varepsilon}{2}}$. In this case, the value-to-price ratio of every item in $L$ is smaller than that of $k_\alpha$ under all three price vectors. Hence, every item in $L$ follows $k_\alpha$ in the value-to-price ordering, and the relative ordering among the items in $L$ is preserved.

On the other hand, by the definition of $L$, the only remaining possibility is $ \frac{\widehat v}{\alpha-\frac{3\varepsilon}{2}} < v_{\ell^*} < \frac{\widehat v}{\alpha-\frac{14\varepsilon}{8}}$.\footnote{All the interval containment inequalities are strict almost surely, since $\varepsilon$ is a continuous random variable.} We further distinguish between the following two subcases: \begin{enumerate}[nosep]
\item[(i)] $\frac{\widehat v}{\alpha-\frac{3\varepsilon}{2}} < v_{\ell^*} < \frac{\widehat v}{\alpha-\frac{13\varepsilon}{8}}$,  
\item[(ii)] $\frac{\widehat v}{\alpha-\frac{13\varepsilon}{8}} < v_{\ell^*} < \frac{\widehat v}{\alpha-\frac{14\varepsilon}{8}}$. 
\end{enumerate}
In both subcases, since the ratios are less than $\delta$ apart (see \Cref{ineq:close-delta}) and item values lie on the $\delta$-grid, every item $\ell \in L\setminus\{\ell^*\}$ satisfies
$ v_\ell < \frac{\widehat v}{\alpha-\frac{3\varepsilon}{2}}$. Therefore, the relative ordering of the items in $L\setminus\{\ell^*\}$ is preserved across the three price vectors. The only possible change in the value-to-price ordering among the vectors hence only involves the relative positions of $k_\alpha$ and $\ell^*$. In subcase~(i), their ordering is the same under $\mathbf q^2(\alpha)$ and $\mathbf q^3(\alpha)$, whereas in subcase~(ii), it is the same under $\mathbf q^1(\alpha)$ and $\mathbf q^2(\alpha)$. Consequently, the relative ordering of all items in $L$ is preserved between one of the two consecutive pairs of price vectors.

Combining this with the analysis of the subsets in (1)--(3), we conclude that the value-to-price ordering of all items remains unchanged either between $\mathbf q^1(\alpha)$ and $\mathbf q^2(\alpha)$ or between $\mathbf q^2(\alpha)$ and $\mathbf q^3(\alpha)$. This establishes the claim.
\end{proof}

Before applying \Cref{lemma:oracle-response-characterization-general-one-coordinate-change}, we verify that, in addition to \Cref{claim:one-works}, each of the three price vectors satisfies
$(\mathbf q^t(\alpha))^\top \mathbf 1_n > B$. Towards this end, consider the (nonempty) set $L$ defined in the proof of \Cref{claim:one-works}. Recall that, under each price vector $\mathbf q^t(\alpha)$, every item in $R_1 \cup \bigl([n]\setminus (R\cup L)\bigr)$ precedes $k_\alpha$ in the value-to-price ordering. The budget remaining after packing these items is
{\allowdisplaybreaks
\begin{align*}
    B - \sum_{i \in R_1 \cup ([n] \setminus (R \cup L))} q^t_i(\alpha)
    &= B - \left(\sum_{i \in [n] \setminus (R \cup L)} 1 + \sum_{r \in R_1} p_r(\alpha) \right)\\
    &\leq B - \left( n - |L| - |R| + (B - n + |R|) \right)
    \tag{\Cref{Definition:TriadPrices}}\\
    & = |L|\\
    & < \sum_{j \in \{k_\alpha\}\cup L\cup R_2} q^t_j(\alpha).
\end{align*}
}
Therefore, each of the three price vectors satisfies $(\mathbf q^t(\alpha))^\top \mathbf 1_n > B$.

Finally, \Cref{lemma:k-packed-only-at-alpha-lesser-than-optimal} implies that $ x^*_{k_\alpha}(\mathbf q^1(\alpha)) > 0. $ Since the price of $k_\alpha$ is strictly smaller under $\mathbf q^2(\alpha)$ and $\mathbf q^3(\alpha)$, it follows that $ x^*_{k_\alpha}(\mathbf q^t(\alpha)) > 0 $ for every $t \in \{1,2,3\}$.

Overall, for whichever of the two pairs, $(\mathbf q^1(\alpha), \mathbf q^2(\alpha)) $ or $ (\mathbf q^2(\alpha), \mathbf q^3(\alpha))$, has identical value-to-price ordering (\Cref{claim:one-works}), all the conditions of \Cref{lemma:oracle-response-characterization-general-one-coordinate-change} hold. Therefore, by \Cref{lemma:oracle-response-characterization-general-one-coordinate-change}, either 
 $\mathcal O_{\mathbf w} \left(\mathbf q^1(\alpha),\mathbf q^2(\alpha)\right) \neq \mathsf{0}$ or $\mathcal O_{\mathbf w} \left(\mathbf q^2(\alpha),\mathbf q^3(\alpha)\right) \neq \mathsf{0}$. This proves statement (ii) of the lemma and completes the proof. 
\end{proof}

\subsubsection{Identifying the item whose scaled value is recovered}
The previous subsection established structural results culminating in a characterization of the oracle response for different values of $\alpha \in \{\beta,2\beta,\ldots,1\}$ (\Cref{corollary:orcale-response-at-alpha-lesser-than-optimal}). This characterization allows us to recover $\alpha^*_{\overline{s}}$ via a binary search using comparison queries and, consequently, to determine the ratio $\widehat v / v_{\overline{s}}$. However, the identity of the item $\overline{s}$ remains unknown.  To identify $\overline{s}$, we first establish \Cref{lemma:item-s*-can-not-be-integrally-packed-under-p-hat} and subsequently derive an oracle characterization that enables us to distinguish $\overline{s}$ using comparison queries.
 
\begin{lemma}
\label{lemma:item-s*-can-not-be-integrally-packed-under-p-hat}
Let $R \subseteq [n]$ satisfy $ \max_{r \in R} v_r < \min_{s \in [n]\setminus R} v_s$ and $|R| \geq n - B - 1$. Define the price vector $\widehat{\mathbf p}$ as follows
\begin{align*}
\widehat{p}_i = \begin{cases}
    1, & i \in [n] \setminus R,\\
    (\alpha^*_{\overline{s}} - 2 \varepsilon) v_i / \widehat{v}, & i \in R~.
\end{cases}
\end{align*}
Then, in the optimal solution under the price vector $\widehat{\mathbf p} \in [0,1]^n$, every item in $[n]\setminus (R \cup \{\overline{s}\})$ is packed integrally (i.e., packed to extent $1$) and item  $\overline{s}$ is unpacked (i.e., $\mathbf{x}^*_{\overline{s}} ( \widehat{\mathbf p} ) = 0$).
\end{lemma}
\begin{proof}
Note that $\widehat{\mathbf p}$ is a valid price vector, i.e., all of its components lie in the interval $[0,1]$. Indeed, for any $r\in R$, the definition of $\alpha^*_{\overline{s}}$ gives us 
\begin{align} 
(\alpha^*_{\overline{s}}-2\varepsilon)\frac{v_r}{\widehat v}  \le \left(\frac{\widehat v}{v_{\overline{s}}}-\frac{\varepsilon}{2}\right) \frac{v_r}{\widehat v} = \frac{v_r}{v_{\overline{s}}} - \frac{\varepsilon v_r}{2\widehat v} \le 1-\delta+\frac{\delta^7}{8n^4} <1 \label{ineq:price-cap}
\end{align} 
The second inequality uses $v_r \leq v_{\overline{s}}-\delta$, $v_{\overline{s}}\leq 1$, $\varepsilon \leq \delta^8/(4n^4)$, and $v_r/\widehat v \leq 1/\delta$. 

We will show that, under the price vector $\widehat{\mathbf p}$, the value-to-price ordering places all items in $[n]\setminus (R\cup\{\overline{s}\})$ ahead of the items in $R$, which in turn precede $\overline{s}$. That is, the greedy construction of the optimal solution considers every item in $[n]\setminus (R\cup\{\overline{s}\})$ then $R$ before reaching $\overline{s}$. This ordering enables us to establish that all items in $[n]\setminus (R\cup\{\overline{s}\})$ are packed integrally in the optimal solution, while $\overline{s}$ remains unpacked.

To obtain that the value-to-price ratio of $\overline{s}$, namely $v_{\overline{s}}$, is strictly smaller than that of every item in $R$, note that the definition of $\alpha^*_{\overline{s}}$ gives us $ \frac{\widehat v}{v_{\overline{s}}} > \alpha^*_{\overline{s}}-\frac{3\varepsilon}{2} \ge \alpha^*_{\overline{s}}-2\varepsilon$. Rearranging yields $v_{\overline{s}} < \frac{\widehat v}{\alpha^*_{\overline{s}}-2\varepsilon}$, which is precisely the value-to-price ratio of each item in $R$.

We next show that the items in $[n]\setminus (R\cup\{\overline{s}\})$ have higher value-to-price ratio than the items in $R$. Since every item in $R$ has value-to-price ratio $\frac{\widehat{v}}{\alpha^*_{\overline{s}} - 2\varepsilon }$, it suffices to show that $\frac{\widehat{v}}{\alpha^*_{\overline{s}} - 2\varepsilon } < \min_{s \in [n] \setminus (R \cup \{\overline{s}\})} v_s$. Towards this, note that 
{\allowdisplaybreaks
\begin{align*}
\frac{\widehat{v}}{\alpha^*_{\overline{s}} - 2\varepsilon} - v_{\overline{s}} &<
\frac{\widehat{v}}{\alpha^*_{\overline{s}} - 2\varepsilon} - \frac{\widehat{v}}{\alpha^*_{\overline{s}} + \beta - 3\varepsilon/2} \tag{$\frac{\widehat{v}}{\alpha^*_{\overline{s}} + \beta - 3\varepsilon/2} < v_{\overline{s}}$ by the definition of $\alpha^*_{\overline{s}}$} \\
&= \widehat{v} \left( \frac{\beta + \varepsilon/2}{\left(\alpha^*_{\overline{s}} - 2\varepsilon \right) \left( \alpha^*_{\overline{s}} + \beta - 3\varepsilon/2 \right)}  \right) \\
&< v_{\overline{s}} \  \frac{\beta + \varepsilon/2}{\left(\alpha^*_{\overline{s}} - 2\varepsilon \right)} \tag{$\frac{\widehat{v}}{\alpha^*_{\overline{s}} + \beta - 3\varepsilon/2} < v_{\overline{s}}$} \\
&< v_{\overline{s}} \  \frac{\beta + \varepsilon/2}{(\widehat{v}/v_{\overline{s}}) - \beta - \varepsilon/2} \tag{$\frac{\widehat{v}}{v_{\overline{s}}} - \beta - \frac{\varepsilon}{2} < \alpha^*_{\overline{s}}  - 2\varepsilon$} \\
&\le \frac{\beta + \varepsilon/2}{\delta - \beta - \varepsilon/2} \tag{$\widehat{v} \ge \delta$ and $v_{\overline{s}} \le 1$} \\
&\le 2\delta^2 \tag{$\beta = \frac{\delta^3}{3n^2}$ and $\varepsilon \le \frac{\delta^8}{4n^4}$}
\end{align*}
}
Since the item values lie on a $\delta$-grid, every item in $[n]\setminus (R\cup\{\overline{s}\})$ has value at least $v_{\overline{s}}+\delta$. Hence, the bound $ \frac{\widehat v}{\alpha^*_{\overline{s}}-2\varepsilon}-v_{\overline{s}} < 2\delta^2 < \delta$ leads to the stated claim: $ \frac{\widehat v}{\alpha^*_{\overline{s}}-2\varepsilon} < v_{\overline{s}}+\delta \le \min_{s\in[n]\setminus (R\cup\{\overline{s}\})} v_s.$

Therefore, in the greedy construction of the optimal solution, all items in $[n]\setminus (R\cup\{\overline{s}\})$ are considered first, followed by the items in $R$, and only afterwards is $\overline{s}$ considered. Since $|[n]\setminus (R\cup\{\overline{s}\})| = n-|R|-1$ and the knapsack capacity satisfies $B \ge n-|R|-1$, all items in $[n]\setminus (R\cup\{\overline{s}\})$ are packed integrally in the optimal solution.

Lastly, the inequalities below show that the total price of all items in $[n]\setminus\{\overline{s}\}$ exceeds the budget:
{\allowdisplaybreaks
\begin{align*}
    \sum_{i \in [n] \setminus \{\overline{s}\}} \widehat{p}_i &= \sum_{i \in [n] \setminus (R \cup \{\overline{s}\})} 1 + \sum_{r \in R} (\alpha^*_{\overline{s}} - 2\varepsilon) \frac{v_r}{\widehat{v}} \\
    &\geq n - |R| - 1 + \sum_{r \in R} \left(\frac{\widehat{v}}{v_{\overline{s}}} - \beta - \frac{\varepsilon}{2}\right) \frac{v_r}{\widehat{v}} \tag{$\alpha^*_{\overline{s}} + \beta - 3\varepsilon/2 > \widehat{v}/v_{\overline{s}} $} \\
    &\geq n - |R| - 1 + \sum_{r \in R} \left(\frac{v_r}{v_{\overline{s}}} - (\beta + \varepsilon/2) \frac{v_r}{\widehat{v}} \right) \\
    &\geq n - |R| - 1 + |R| \left(\frac{\min_{i \in [n]} v_i}{\max_{j \in [n]} v_j} - (\beta + \varepsilon/2) \frac{v_r}{\widehat{v}} \right) \\
    &\geq n - |R| - 1 + |R| \left(\frac{\min_{i \in [n]} v_i}{\max_{j \in [n]} v_j} - \frac{2 \delta^3}{3 n^2} \cdot \frac{1}{\delta}\right) \tag{$\beta + \frac{\varepsilon}{2} \leq \frac{2\delta^3}{3n^2}$ and $\frac{v_r}{\widehat{v}} \leq \frac{1}{\delta}$} \\
    & = (n - 1) \left(\frac{\min_{i \in [n]} v_i}{\max_{j \in [n]} v_j} - \frac{2\delta^2}{3n^2}\right)  \\
    &> B \tag{$B \leq (n - 1) \left(\frac{\min_{i \in [n]} v_i}{\max_{j \in [n]} v_j} - \frac{\delta^2}{n^2}\right)$}
\end{align*}
}
Since every item in $[n]\setminus\{\overline{s}\}$ precedes $\overline{s}$ in the value-to-price ordering, the greedy construction exhausts the budget before reaching $\overline{s}$. Therefore,
$x^*_{\overline{s}}(\widehat{\mathbf p}) = 0$, which completes the proof.
\end{proof}

As a consequence of \Cref{lemma:item-s*-can-not-be-integrally-packed-under-p-hat}, we obtain the following characterization of the oracle response. In particular, once $\alpha^*_{\overline{s}}$ is known, the characterization allows us to identify the item $\overline{s}$ using comparison queries.
\begin{corollary}
\label{lemma:oracle-response-for-determining-which-item-is-s*}
With the same notation and assumptions as in \Cref{lemma:item-s*-can-not-be-integrally-packed-under-p-hat}, for every item $k \in [n]\setminus R$, $\mathcal O_{\mathbf w}\!\left(\widehat{\mathbf p},
\widehat{\mathbf p}-\frac{\varepsilon}{3}\mathbf e_k\right) = \mathsf 0$ if and only if $k=\overline{s}$.
\end{corollary}
\begin{proof}
    By \Cref{lemma:item-s*-can-not-be-integrally-packed-under-p-hat}, we have that, under price vector $\widehat{\mathbf p}$, all items in $[n] \setminus (R \cup \{\overline{s}\})$ are integrally packed whereas item $\overline{s}$ is unpacked. Moreover, since $\overline{s}$ is not packed under $\widehat{\mathbf p}$, we have $ \widehat{\mathbf p}^{\top}\mathbf 1_n > B. $ We will also use the fact, established in the proof of \Cref{lemma:item-s*-can-not-be-integrally-packed-under-p-hat}, that under $\widehat{\mathbf p}$ the value-to-price ordering places all items in $[n]\setminus (R\cup\{\overline{s}\})$ before the items in $R$, which in turn precede $\overline{s}$. 

    Throughout the proof, we use the notation $\mathbf z^k \coloneqq \widehat{\mathbf p} - \frac{\varepsilon}{3}\mathbf e_k, $ for every $k \in [n] \setminus R$ .

    Fix an item $k \in [n]\setminus R$. The value-to-price ratios of all items $j \in [n]\setminus\{k\}$ remain unchanged under the price vector $\mathbf z^k$. On the other hand, the value-to-price ratio of $k$ under $\mathbf z^k$ is 
    \[ \frac{v_k}{1-\varepsilon/3} \leq v_k \left(1+\frac{2\varepsilon}{3}\right) \leq v_k + \frac{\delta^8}{6n^4},\] 
    where the first inequality uses $1/(1-x) \leq 1+2x$ for all $x \in [0,1/2]$, and the second follows from $\varepsilon \leq \delta^8/(4n^4)$ and $v_k\leq 1$.

    \medskip
    \noindent
    \textbf{Case {\rm I}}. $k \neq \overline{s}$. We claim that the greedy ordering of the items is unchanged between $\widehat{\mathbf p}$ and $\mathbf z^k$. If $k$ has the largest value-to-price ratio under $\widehat{\mathbf p}$, then lowering its price can only increase this ratio, and hence the ordering remains unchanged. Otherwise, let $i'$ denote the item immediately preceding $k \in [n] \setminus \left( R \cup \{\overline{s} \} \right)$ in the value-to-price ordering under $\widehat{\mathbf p}$. Since all items in $R \cup \{\overline{s}\}$ appear after the items in $[n]\setminus (R\cup\{\overline{s}\})$, we have $i' \in [n]\setminus (R\cup\{\overline{s}\})$. Therefore, the value-to-price ratio of $i'$ under $\mathbf z^k$ is $v_{i'}$. On the other hand, $\frac{v_k}{1-\varepsilon/3} \leq v_k+\frac{\delta^8}{6n^4} < v_{i'}$, where the last inequality follows from the fact that the values lie on a $\delta$-grid. Thus, even under $\mathbf z^k$, the item $k$ remains below $i'$ in the ordering. Since only the value-to-price ratio of $k$ changes, the entire ordering is preserved under $\mathbf z^{k}$.

    Furthermore, $x^*_k(\widehat{\mathbf p}) = 1 > 0$. Therefore, \Cref{lemma:oracle-response-characterization-general-one-coordinate-change} implies that $\mathcal O_{\mathbf w}(\widehat{\mathbf p}, \mathbf z^k) \neq \mathsf 0$.

    \medskip
    \noindent
    \textbf{Case {\rm II}}. $k=\overline{s}$. Here again we show that the greedy ordering of the items is unchanged between $\widehat{\mathbf p}$ and $\mathbf z^k$. Under $\widehat{\mathbf p}$, the value-to-price ratio of $\overline{s}$ is $v_{\overline{s}}$, and the item immediately preceding $\overline{s}$ in the value-to-price ordering is some item $r\in R$, whose value-to-price ratio is $\widehat v/(\alpha^*_{\overline{s}}-2\varepsilon)$. Under $\mathbf z^{\overline{s}}$, the value-to-price ratio of $r$ remains unchanged, while that of $\overline{s}$ becomes $v_{\overline{s}}/(1-\varepsilon/3)$.
    Furthermore, 
    \begin{align*}
        \frac{\widehat{v}}{\alpha^*_{\overline{s}} - 2\varepsilon} - \frac{v_{\overline{s}}}{1 - \varepsilon/3} &= v_{\overline{s}} \left(\frac{\widehat{v}/v_{\overline{s}}}{\alpha^*_{\overline{s}} - 2\varepsilon} - \frac{1}{1 - \varepsilon/3}\right) \\
        &\geq v_{\overline{s}} \left(\frac{\widehat{v}/v_{\overline{s}}}{\widehat{v}/v_{\overline{s}} - \varepsilon/2} - \frac{1}{1 - \varepsilon/3}\right) \tag{by definition, $\widehat{v}/v_{\overline{s}} \geq \alpha^*_{\overline{s}} - 3\varepsilon/2$} \\
        &\geq v_{\overline{s}} \left(\frac{1}{1 - \varepsilon/2} - \frac{1}{1 - \varepsilon/3}\right) \tag{$\widehat{v}/v_{\overline{s}} \leq 1$}\\
        &> 0~.
    \end{align*} 
    Therefore, $r$ continues to precede $\overline{s}$ in the value-to-price ordering under $\mathbf z^{\overline{s}}$. Since the value-to-price ratios of all other items are unchanged, the greedy ordering is the same between $\widehat{\mathbf p}$ and $\mathbf z^{\overline{s}}$. Moreover, $x^*_{\overline{s}}(\widehat{\mathbf p}) = 0$. Hence, by \Cref{lemma:oracle-response-characterization-general-one-coordinate-change}, $\mathcal O_{\mathbf w}(\widehat{\mathbf p}, \mathbf z^{\overline{s}}) = \mathsf 0$. 
    
    This corollary stands proved. 
\end{proof}

\subsubsection{Algorithm for Recovering the Scaled Values}
\label{subsection:alg-for-packed}
In the previous subsection we have characterized the oracle response under the assumption that the triad price $\mathbf p(\alpha)$ exists for the given value of $\alpha$. Although $\mathbf p(\alpha)$ may fail to exist for sufficiently small $\alpha$, \Cref{lem:triad-price-exists-implies} (established below) shows that this can occur only when $\alpha$ is smaller than $\alpha^*_{\overline{s}}-\beta$. These structural results and characterization allow us to recover the quantity $\alpha^*_{\overline{s}}$ and identify the item $\overline{s}$ as follows: The algorithm (\Cref{alg:get-value-ratio-of-items-in-S}) proceeds iteratively, maintaining a set $R \subseteq [n]$ of items whose scaled values have already been recovered. In each iteration, the goal is to recover the scaled value of the least-valued item $\overline{s}$ in $[n]\setminus R$. To this end, the algorithm performs a binary search over the grid $\{\beta,2\beta,\ldots,1\}$ to locate $\alpha^*_{\overline{s}}$. The search is guided by the oracle-response characterization in \Cref{corollary:orcale-response-at-alpha-lesser-than-optimal}. Whenever the triad price $\mathbf p(\alpha)$ does not exist, which can be checked in polynomial time by \Cref{lemma:computation-of-triad-prices}, the algorithm discards that value of $\alpha$ and, by \Cref{lem:triad-price-exists-implies}, continues the binary search in the upper half of the grid. Once the binary search identifies $\alpha^*_{\overline{s}}$, it remains to determine the item $\overline{s}$. To this end, the algorithm uses the oracle-response characterization in \Cref{lemma:oracle-response-for-determining-which-item-is-s*}. Having recovered both the identity of $\overline{s}$ and its scaled value ratio from $\alpha^*_{\overline{s}}$, it augments the set $R$ with $\overline{s}$ and proceeds to the next iteration.

In particular, for each maintained set $R$ in \Cref{alg:get-value-ratio-of-items-in-S}, Lines~\ref{alg:get-value-ratio-of-items-in-S.line:triad-price-dne-search-upper-half}--\ref{line:found-alpha-star} are used to recover $\alpha^*_{\overline{s}}$, and Lines~\ref{alg:get-value-S.line:construct-p-hat-for-finding-s*}--\ref{alg:get-value-S.line:figure-out-s*} are then used to identify the corresponding item $\overline{s}$.
 
\begin{algorithm}[!h]
\caption{Recovering scaled values of items in $S$}
\label{alg:get-value-ratio-of-items-in-S}

\begin{algorithmic}[1]

\Require Items $S$, value ratios $\{\frac{v_u}{v_f}:u\in U\}$, and perturbation parameters $\beta,\varepsilon,\varepsilon_1,\varepsilon_2>0$.
\Ensure Value ratios $\frac{v_s}{\widehat v}$ for every item $s\in S$, where $\widehat v=\min_{j\in[n]}v_j$.

\State Initialize $R=U\cup\{f\}$. Let $\frac{\widehat v}{v_f} =\min_{u\in U}\frac{v_u}{v_f}$, and compute $\frac{v_r}{\widehat v} = \frac{v_r/v_f}{\widehat v/v_f}$ for every $r \in R$. 

\While{$[n]\setminus R$ is not empty}\linelabel{alg:get-value-S.line:while-loop-starts}

    \State Perform a binary search over $\alpha\in\{\beta,2\beta,\ldots,1\}$: For each candidate $\alpha$, first check if triad price $\mathbf p(\alpha)$ exists. If $\mathbf p(\alpha)$ does not exist, search in the upper half. $\quad$// \texttt{\Cref{lem:triad-price-exists-implies}} \linelabel{alg:get-value-ratio-of-items-in-S.line:triad-price-dne-search-upper-half}
    
    \State Else, if $\mathbf p(\alpha)$ exists, then let $k_\alpha$ as the corresponding triad item, and define $\mathbf q^1 \coloneqq \mathbf p(\alpha) - \frac{\varepsilon v_{k_\alpha}}{2\widehat{v}}\,\mathbf e_{k_\alpha}, \quad \mathbf q^2 \coloneqq \mathbf p(\alpha) - \frac{5\varepsilon v_{k_\alpha}}{8\widehat{v}}\,\mathbf e_{k_\alpha}, \quad \mathbf q^3 \coloneqq \mathbf p(\alpha) - \frac{6\varepsilon v_{k_\alpha}}{8\widehat{v}}\,\mathbf e_{k_\alpha}.$
    
    Query $ r^1_\alpha \coloneqq \mathcal O_{\mathbf w}(\mathbf q^1,\mathbf q^2) $ and $ r^2_\alpha \coloneqq \mathcal O_{\mathbf w}(\mathbf q^2,\mathbf q^3). $ If $r^1_\alpha = r^2_\alpha = \mathsf 0$, continue the binary search in the lower half; otherwise, if $r^1_\alpha \vee r^2_\alpha \neq \mathsf 0$, continue the search in the upper half. $\quad$ // \texttt{\Cref{corollary:orcale-response-at-alpha-lesser-than-optimal}} \linelabel{alg:get-value-S.line:binary-search-alpha-with-comparison-query}
    
    \State Let $\widehat{\alpha}$ be the largest value of $\alpha$ for which $r^1_\alpha \vee r^2_\alpha \neq \mathsf{0}$.

    \State Find the unique rational number $p/q$ in the interval $ \left( \widehat{\alpha}- \beta - \frac{3\varepsilon_2}{2},\ \widehat{\alpha}+\beta-\frac{3\varepsilon_1}{2} \right)$ satisfying $p,q\in \{1,2,\ldots, \lfloor1/\delta\rfloor\}$ by invoking the simplest intervening rational algorithm, which runs in $O(\operatorname{poly}\log(1/\delta))$ time (Appendix \ref{appendix:Stern-Brocot}). \linelabel{alg:get-value-S.line:Stern-Brocot-search}

    \State Set $\widetilde{\alpha} = \widehat{\alpha}$ if $p/q \in (\widehat{\alpha} - 3\varepsilon/2, \widehat{\alpha} + \beta - 3\varepsilon/2)$, else set $\widetilde{\alpha} = \widehat{\alpha} - \beta$. \label{line:found-alpha-star}
    
    \State Define the price vector $\widehat{\mathbf p}$ by $\widehat{p}_i = 1$, for all $i \in [n] \setminus R$, and $\widehat{p}_r = (\widetilde{\alpha} - 2\varepsilon) v_r/\widehat{v}$, for all $r \in R$. \linelabel{alg:get-value-S.line:construct-p-hat-for-finding-s*}

    \State For every item $i\in [n] \setminus R$, query $ r_i := \mathcal O_{\mathbf w} \left( \widehat{\mathbf p},\ \widehat{\mathbf p}-(\varepsilon/3)\mathbf e_i \right)$. If $r_i=\mathsf{0}$, set $\widehat{s} = i$ and terminate this inner loop, else continue to the next item. $\quad$ // \texttt{\Cref{lemma:oracle-response-for-determining-which-item-is-s*}} \linelabel{alg:get-value-S.line:figure-out-s*}

    \State Set $\frac{v_{\widehat{s}}}{\widehat{v}}=\frac{q}{p}$, and update $R \leftarrow R \cup \{ \widehat{s} \}$.

\EndWhile

\State \textbf{Return} the scaled values $\frac{v_s}{\widehat v}$ for every $s\in S$.

\end{algorithmic}
\end{algorithm}

We begin with a lemma showing that if the triad price fails to exist, then the corresponding value of $\alpha$ must sufficiently small.

\begin{lemma}\label{lem:triad-price-exists-implies}
Let $R \subseteq [n]$ satisfy $\max_{r\in R} v_r < \min_{s\in [n]\setminus R} v_s$, and let $\overline{s} = \arg\min_{s\in [n]\setminus R} v_s$. If $\alpha \in \{\beta,2\beta,\ldots,1\}$ is such that the triad price vector $\mathbf p(\alpha)$ does not exist, then $\alpha < \alpha^*_{\overline{s}}-\beta$. Equivalently, if $\alpha \geq \alpha^*_{\overline{s}}-\beta$, then the triad price vector $\mathbf p(\alpha)$ exists.
\end{lemma}

\begin{proof}
    We first show that, for $\alpha' = \alpha^*_{\overline{s}} - \beta$, the triad price vector $\mathbf p(\alpha')$ exists. Then, by the existential monotonicity established in \Cref{lemma:monotonicity-of-triad-prices}, it follows that $\mathbf p(\alpha)$ exists for every $\alpha \in \{\beta,2\beta,\ldots,1\}$ satisfying $\alpha \geq \alpha^*_{\overline{s}}-\beta$, which proves the lemma.

    For $\alpha' =\alpha^*_{\overline{s}}-\beta$, consider the term $\min\left\{1,(\alpha'-2\varepsilon)\frac{v_{i}}{\widehat v}\right\} $ appearing in the definition of triad price vectors (\Cref{Definition:TriadPrices}). We claim that $(\alpha'-2\varepsilon)\frac{v_{r_i}}{\widehat v}\le 1$ for all $i \in R$. Indeed,  
    \[ (\alpha^*_{\overline{s}} - \beta - 2 \varepsilon) \frac{v_i}{\widehat{v}} \leq (\alpha^*_{\overline{s}} - 2 \varepsilon) \frac{v_i}{\widehat{v}} \leq \left( \frac{\widehat{v}}{v_{\overline{s}}} + \frac{3 \varepsilon}{2} - 2 \varepsilon \right) \frac{v_i}{\widehat{v}} = \frac{\widehat{v}}{v_{\overline{s}}} \frac{v_i}{\widehat{v}} - \frac{ \varepsilon}{2} \frac{v_i}{\widehat{v}} = \frac{v_i}{v_{\overline{s}}} - \frac{ \varepsilon v_i}{2 \widehat{v}} \leq 1.\]
    The last inequality follows from $\frac{v_i}{v_{\overline{s}}} \leq 1$ for $i \in R$. 

    Let $r^* \coloneqq \arg\min_{r\in R} v_r$. We next evaluate, for the set $R\setminus\{r^*\}$, the condition that specifies the subset $R_1 \subseteq R$ in \Cref{Definition:TriadPrices}. In parcticular, using the bound $\alpha^*_{\overline{s}} + \beta - \frac{3\varepsilon}{2} > \frac{\widehat{v}}{v_{\overline{s}}}$, we obtain 
    \begin{align*}
        \sum_{s \in [n] \setminus R} 1 + \sum_{i \in R \setminus \{r^*\}} (\alpha^*_{\overline{s}} - \beta - 2 \varepsilon) \frac{v_i}{\widehat{v}} &\geq n - |R| + \sum_{i \in R \setminus \{r^*\} } \left(\frac{v_i}{v_{\overline{s}}} - \left(2\beta +  \frac{\varepsilon}{2}\right) \frac{v_i}{\widehat{v}} \right) \\
        &\geq n - |R| + (|R| - 1) \frac{\min_{i \in [n]} v_i}{\max_{j \in [n]} v_j} - (|R| - 1) \frac{2\beta + \varepsilon/2}{\delta} \\
        &\geq (n - 1) \frac{\min_{i \in [n]} v_i}{\max_{j \in [n]} v_j} - \frac{(n-1) (2\beta + \varepsilon/2)}{\delta} \\
        &\geq (n - 1) \left( \frac{\min_{i \in [n]} v_i}{\max_{j \in [n]} v_j} - \frac{\delta^2}{n^2} \right) \tag{$2\beta + \frac{\varepsilon}{2} \leq \frac{\delta^3}{n^2}$} \\
        &\geq B \tag{$B \leq (n - 1) \left(\frac{\min_{i \in [n]} v_i}{\max_{j \in [n]} v_j} - \frac{\delta^2}{n^2}\right)$}.
    \end{align*}
 Therefore, for $\alpha' = \alpha^*_{\overline{s}} - \beta$, the triad price vector exists with some $R_1 \subseteq R \setminus \{r^*\}$ and the corresponding triad item $k_{\alpha'}$ is well defined. This proves the lemma.
 \end{proof}

The next lemma analyzes a single iteration of \Cref{alg:get-value-ratio-of-items-in-S}. It establishes that the binary search with triad-price queries correctly recovers the ratio $\widehat v / v_{\overline{s}}$ for the least-valued unrecovered item $\overline{s}$, and that the item $\widehat s$ returned by the algorithm is precisely $\overline{s}$. 

\begin{lemma}
\label{lemma:algorithm-get-values-S-analyzing-one-iteration}
Suppose that, at the start of an iteration of the while-loop in \Cref{alg:get-value-ratio-of-items-in-S} (Line~\ref{alg:get-value-S.line:while-loop-starts}), the set $R$ consists of items $r \in [n]$ whose value ratios $v_r/\widehat{v}$ are known. Further, suppose that $|R| \ge n-B-1$ and  $\max_{r\in R} v_r < v_{\overline{s}}$, where $\overline{s} = \argmin_{s \in [n] \setminus R} v_s$. Let $\widehat{\alpha}$, $\widetilde{\alpha}$ and $\widehat{s}$ be as defined in the algorithm. Then, at the end of the iteration, the following hold:
\begin{enumerate}
    \item[(i)] $\frac{\widehat{v}}{v_{\overline{s}}} \in (\widehat{\alpha} - \beta - \frac{3 \varepsilon}{2}, \widehat{\alpha} + \beta - \frac{3\varepsilon}{2})$.
    \item[(ii)] $\widetilde{\alpha} = \alpha^*_{\overline{s}}$.
    \item[(iii)] $\widehat{s} = \overline{s}$.
\end{enumerate}
\end{lemma}
\begin{proof}
Fix an iteration of the while-loop in \Cref{alg:get-value-ratio-of-items-in-S}. 

For \textbf{part (i)}, we first show that the binary search in the iteration necessarily returns a value $\widehat{\alpha} \in \{\alpha^*_{\overline{s}},\, \alpha^*_{\overline{s}}+\beta\}$. Since $\max_{r \in R} v_r < \min_{s \in [n]\setminus R} v_s = v_{\overline{s}}$, \Cref{lem:triad-price-exists-implies} ensures that the binary search correctly identifies the range of $\alpha \in \{\beta,2\beta,\ldots,1\}$ for which the triad price vector $\mathbf p(\alpha)$ exists. Now fix any $\alpha$ in this range. By \Cref{corollary:orcale-response-at-alpha-lesser-than-optimal}, if $\alpha > \alpha^*_{\overline{s}}+\beta$, then both $ r^1_\alpha=\mathcal O_{\mathbf w}(\mathbf q^1(\alpha),\mathbf q^2(\alpha)) $ and $ r^2_\alpha=\mathcal O_{\mathbf w}(\mathbf q^2(\alpha),\mathbf q^3(\alpha)) $ are equal to $\mathsf 0$, so the algorithm correctly continues the binary search in the lower half. On the other hand, if $\alpha \le \alpha^*_{\overline{s}}$, then at least one of $r^1_\alpha$ and $r^2_\alpha$ is nonzero, and hence $ r^1_\alpha \vee r^2_\alpha \neq \mathsf 0. $ In this case, the algorithm correctly continues the search in the upper half. 

The only unresolved value is $\alpha=\alpha^*_{\overline{s}}+\beta$. If $ r^1_{\alpha^*_{\overline{s}}+\beta}\vee r^2_{\alpha^*_{\overline{s}}+\beta}=\mathsf 1, $ then the binary search returns $\widehat{\alpha}=\alpha^*_{\overline{s}}+\beta$; otherwise it returns $\widehat{\alpha}=\alpha^*_{\overline{s}}$. Therefore, $ \widehat{\alpha}\in\{\alpha^*_{\overline{s}},\,\alpha^*_{\overline{s}}+\beta\} $, as required.

Finally, by the definition of $\alpha^*_{\overline{s}}$, we have $\widehat{v}/v_{\overline{s}} \in (\alpha^*_{\overline{s}} - 3\varepsilon/2, \alpha^*_{\overline{s}} + \beta - 3\varepsilon/2)$. Therefore, 
for either choice of $\widehat{\alpha} \in \{\alpha^*_{\overline{s}},\, \alpha^*_{\overline{s}}+\beta\}$, it holds that $\widehat{v}/v_{\overline{s}} \in  (\widehat{\alpha} - \beta - \frac{3 \varepsilon}{2}, \widehat{\alpha} + \beta - \frac{3\varepsilon}{2})$. This proves part (i).

\medskip
For \textbf{part (ii)}, consider the interval
$
    \left(
        \widehat{\alpha}-\beta-\frac{3\varepsilon_2}{2},\,
        \widehat{\alpha}+\beta-\frac{3\varepsilon_1}{2}
    \right).
$

Since $\widehat{\alpha}$ lies on the $\beta$-grid, and since $\beta,\varepsilon_1,\varepsilon_2$ have bounded bit complexity, the endpoints of this interval have bounded bit complexity. Also, since $\varepsilon\in[\varepsilon_1,\varepsilon_2]$, we have
$
    \left(
        \widehat{\alpha}-\beta-\frac{3\varepsilon}{2},\,
        \widehat{\alpha}+\beta-\frac{3\varepsilon}{2}
    \right)
    \subseteq
    \left(
        \widehat{\alpha}-\beta-\frac{3\varepsilon_2}{2},\,
        \widehat{\alpha}+\beta-\frac{3\varepsilon_1}{2}
    \right).
$
Hence, by part~(i),
$
    \frac{\widehat v}{v_{\overline{s}}}
    \in
    \left(
        \widehat{\alpha}-\beta-\frac{3\varepsilon_2}{2},\,
        \widehat{\alpha}+\beta-\frac{3\varepsilon_1}{2}
    \right).
$
Moreover, this interval has length
\begin{align*}
    &\left(\widehat{\alpha}+\beta-\frac{3\varepsilon_1}{2}\right)
    -
    \left(\widehat{\alpha}-\beta-\frac{3\varepsilon_2}{2}\right) \\
    &\qquad =
    2\beta+\frac{3}{2}(\varepsilon_2-\varepsilon_1) \\
    &\qquad =
    \frac{2\delta^3}{3n^2}
    +
    \frac{\delta^8}{8n^4}
    <
    \delta^2.
\end{align*}
Also, $\widehat v/v_{\overline{s}}$ is of the form $p/q$, where $p,q\in(0,1/\delta]\cap\mathbb Z$. Therefore, by \Cref{lemma:unique-p/q-in-interval-of-length-delta^2}, $\widehat v/v_{\overline{s}}$ is the unique rational of this form contained in the interval above. Hence, the simplest intervening rational algorithm (Appendix~\ref{appendix:Stern-Brocot}) in Line~\ref{alg:get-value-S.line:Stern-Brocot-search} correctly recovers the ratio $\widehat v/v_{\overline{s}}$ in $O(\operatorname{poly} \log(1/\delta))$ time. It remains to identify $\alpha^*_{\overline{s}}$. Recall that
$
    \frac{\widehat v}{v_{\overline{s}}}
    \in
    \left(
        \alpha^*_{\overline{s}}-\frac{3\varepsilon}{2},\,
        \alpha^*_{\overline{s}}+\beta-\frac{3\varepsilon}{2}
    \right),
$
and the algorithm checks whether the recovered ratio $\widehat v/v_{\overline{s}}$ lies in the interval
$
    \left(
        \widehat{\alpha}-\frac{3\varepsilon}{2},\,
        \widehat{\alpha}+\beta-\frac{3\varepsilon}{2}
    \right).
$
If it does, then necessarily $\widehat{\alpha}=\alpha^*_{\overline{s}}$, and the algorithm correctly sets $\widetilde{\alpha}=\widehat{\alpha}=\alpha^*_{\overline{s}}$.

Otherwise, we must have $\widehat{\alpha}=\alpha^*_{\overline{s}}+\beta$. In this case, the tested interval equals
$
    \left(
        \alpha^*_{\overline{s}}+\beta-\frac{3\varepsilon}{2},\,
        \alpha^*_{\overline{s}}+2\beta-\frac{3\varepsilon}{2}
    \right),
$
which does not contain $\widehat v/v_{\overline{s}}$. Hence, again the algorithm correctly sets $\widetilde{\alpha}=\widehat{\alpha}-\beta=\alpha^*_{\overline{s}}$. This proves part~(ii).

    \medskip
    For \textbf{part (iii)}, observe that the algorithm constructs a price vector $\widehat{\mathbf p}$ such that $\widehat p_i = 1$ for every $i \in [n]\setminus R$, and $\widehat p_r = (\widetilde{\alpha}-2\varepsilon)v_r/\widehat v$ for every $r \in R$ (Line~\ref{alg:get-value-S.line:construct-p-hat-for-finding-s*}). By part~(ii), we have $\widetilde{\alpha}=\alpha^*_{\overline{s}}$, and hence the constructed vector $\widehat{\mathbf p}$ coincides with the one in \Cref{lemma:oracle-response-for-determining-which-item-is-s*}. Moreover, by assumption, $\max_{r\in R} v_r < \min_{s\in [n]\setminus R} v_s$ and $|R| \ge n-B-1$. Therefore, \Cref{lemma:oracle-response-for-determining-which-item-is-s*} applies and yields that, for every $i \in [n]\setminus R$, we have $\mathcal O_{\mathbf w} \left(\widehat{\mathbf p},\, \widehat{\mathbf p} - \frac{\varepsilon}{3}\mathbf e_i\right)=\mathsf 0$ if and only if $i=\overline{s}$. Hence, the algorithm correctly identifies $\widehat{s}=\overline{s}$. 

    This completes the proof of the lemma.
\end{proof}

The next lemma establishes the invariant maintained throughout the execution of Algorithm~\ref{alg:get-value-ratio-of-items-in-S}. In particular, at the start of every iteration, the set $R$ consists of items whose scaled values are already known, and every item in $R$ has value strictly smaller than every item in $[n]\setminus R$. Consequently, each iteration correctly recovers the scaled value of the least-valued item in $[n]\setminus R$ while preserving this invariant. 

\begin{lemma}\label{lem:Per-iteration-correctness-and-invariant} 
Throughout the execution of \Cref{alg:get-value-ratio-of-items-in-S}, the following invariants hold: $\max_{r\in R} v_r < \min_{s\in [n]\setminus R} v_s$ and $|R| \ge n-B-1$.
Moreover, at the start of every iteration, the value ratios $v_r/\widehat v$ are known for all $r\in R$. 
\end{lemma}

\begin{proof}
We prove the lemma by induction on the number of iterations of the while-loop in \Cref{alg:get-value-ratio-of-items-in-S}. Initially, the algorithm sets $R = U \cup \{f\}$. By \Cref{alg:get-value-ratio-for-items-in-U}, the ratios $v_u/v_f$ are known for all $u \in U$. Moreover, under the price vector $\mathbf 1_n$, the greedy ordering coincides with the decreasing order of item values. Hence, every item in $S$ has value larger than $v_f$, whereas every item in $U$ has value smaller than $v_f$. Therefore, $\max_{r\in U\cup\{f\}} v_r < \min_{s\in S} v_s$. Since $\widehat v = \min_{u\in U} v_u$, it follows that the ratios $v_r/\widehat v$ are known for all $r \in U \cup \{f\}$. Also, at initialization, $|R| = |U \cup \{f\}| = n - |S| = n - \lfloor B \rfloor \ge n - B - 1$. Thus, the invariants hold in the base case.

Now suppose the invariant holds at the start of some iteration of the while-loop. Hence, all the conditions required to apply \Cref{lemma:algorithm-get-values-S-analyzing-one-iteration} holds. By that lemma, the current iteration correctly identifies the item $\overline{s}$ and recovers the ratio $v_{\overline{s}}/\widehat v$. The algorithm then updates $R\gets R\cup\{\overline{s}\}$. Consequently, the ratios $v_r/\widehat v$ are known for all items in the updated set $R$, and since $R$ only grows, the bound $|R|\ge n-B-1$ continues to hold. Also, note that since $\overline{s}$ is the minimum-valued item in the old set $[n]\setminus R$, the invariant $\max_{r\in R} v_r < \min_{s\in [n]\setminus R} v_s$ is preserved as well. 

Thus, by induction, the stated invariants hold throughout the execution of \Cref{alg:get-value-ratio-of-items-in-S}. This proves the lemma.
\end{proof}

\begin{lemma}
\label{lemma:performance-guarantee-of-alg-for-computing-value-ratio-S}
\Cref{alg:get-value-ratio-of-items-in-S} correctly recovers the value ratios $\frac{v_s}{\widehat v}$ for all $s\in S$ using at most $O(B\log(1/\delta)+B^2)$ comparison-oracle queries. Moreover, it runs in time polynomial in $n$, $B$, and $\log(1/\delta)$.
\end{lemma}

\begin{proof}
We begin by proving correctness. Algorithm~\ref{alg:get-value-ratio-of-items-in-S} initializes $R=U\cup\{f\}$ and then proceeds iteratively to recover the value ratios of items in $S$. By \Cref{lem:Per-iteration-correctness-and-invariant,lemma:algorithm-get-values-S-analyzing-one-iteration}, in each iteration, the algorithm correctly identifies the item $\overline{s}=\arg\min_{s\in [n]\setminus R} v_s$ and recovers the ratio $v_{\overline{s}}/\widehat v$. Consequently, after $|S|$ iterations, every item in $S$ has been added to $R$. It follows that \Cref{alg:get-value-ratio-of-items-in-S} correctly recovers the ratio $v_s/\widehat v$ for every $s\in S$.

We now establish the query complexity of the algorithm. The algorithm performs exactly $|S|$ while-loop iterations, and since $S$ is the set of integrally packed items under the price vector $\mathbf 1_n$, we have $|S|=\lfloor B\rfloor$. Thus, the algorithm performs exactly $\lfloor B\rfloor$ iterations, so it suffices to bound the number of comparison queries used in a single iteration.

In one iteration, the algorithm first performs a binary search over the grid $\{\beta,2\beta,\ldots,1\}$, which uses $O(\log(1/\beta))=O(\log(1/\delta))$ comparison queries; recall that $\beta= \delta^3/3n^2$ and we have $n\le 1/\delta$. The algorithm then executes the item-identification step. At this stage, the unresolved set is $[n]\setminus R$, and we have $[n]\setminus R\subseteq S$. Hence $|[n]\setminus R|\le |S|=\lfloor B\rfloor$. Since the algorithm makes one comparison query for each item in $[n]\setminus R$, the item-identification step uses at most $\lfloor B\rfloor$ comparison queries in the iteration.

Thus, each iteration requires $O(\log(1/\delta)+B)$ comparison queries. As \Cref{alg:get-value-ratio-of-items-in-S} runs for $\lfloor B\rfloor$ iterations, the total number of comparison queries made by \Cref{alg:get-value-ratio-of-items-in-S} is $O\bigl(B\log(1/\delta)+B^2\bigr)$.

Finally, for bounding the time complexity of the algorithm note that each triad price can be computed in polynomial time (\Cref{lemma:computation-of-triad-prices}), and the Stern-Brocot search runs in time polynomial in $\log(1/\delta)$. Therefore, the overall running time is polynomial in $n$, $B$ and $\log(1/\delta)$. This completes the proof. 
\end{proof}

\subsubsection{Proofs of \Cref{lemma:computation-of-triad-prices,lemma:monotonicity-of-triad-prices}}
\label{sec:proof-of-computation-of-triad-prices}

\begin{proof}[Proof of \Cref{lemma:computation-of-triad-prices}]
As in \Cref{Definition:TriadPrices}, let the given set $R=\{r_1,r_2,\dots,r_{|R|}\}$ be indexed in decreasing order of values. Since the scaled values of the items in $R$ are known and $\arg\min_{j\in[n]} v_j \in R$---i.e., $\widehat v=\min_{r\in R} v_r$---we can efficiently compute both this indexing and the ratios $\{v_r/\widehat v\}_{r\in R}$.

Using these ratios, define $\pi_i \coloneqq \min \left\{1, \ \left(\alpha - 2 \varepsilon \right) \frac{v_{r_i}}{\widehat{v}} \right\}$ for each $1 \leq i \leq |R|$. Each $\pi_i$ is nonnegative since $\alpha \geq \beta > 2 \varepsilon$. 

By definition, for a given $\alpha$, the triad price vector $\mathbf p(\alpha)$ exists if there is an index $t<|R|$ such that $\sum_{i=1}^t \pi_i \ge B-n+|R|$, where $t$ is the smallest such index. Since each $\pi_i$ is nonnegative, this condition can be checked by scanning the prefix sums of $\pi_1,\pi_2,\ldots,\pi_{|R|-1}$ in order and selecting the first index $t$ at which the threshold is reached. Thus, one can efficiently determine whether such an index $t$ exists, and hence whether the triad price vector $\mathbf p(\alpha)$ exists. Whenever it does, we set $k_\alpha=r_{t+1}$ and compute $\mathbf p(\alpha)$ as specified in \Cref{Definition:TriadPrices}.

Also, for every maintained set $R$, we have $|R| \ge n-|S| = n-\lfloor B\rfloor$, and hence $B-n+|R| \ge B-n+(n-\lfloor B\rfloor)=B-\lfloor B\rfloor>0$. In particular, whenever the triad price vector exists, the corresponding index $t$ must satisfy $t\ge 1$. 

Overall, one can determine the existence of triad prices and compute them in $O(n\log n)$ time. This completes the proof.

\end{proof}

\begin{proof}[Proof of \Cref{lemma:monotonicity-of-triad-prices}]
Since the triad price vector $\mathbf p(\widehat{\alpha})$ exists for the given $\widehat{\alpha}\in(0,1)$, there exists an index $\widehat t<|R|$ such that $\sum_{i=1}^{\widehat t}\min\left\{1,\;(\widehat{\alpha}-2\varepsilon)\frac{v_{r_i}}{\widehat v}\right\}\ge B-n+|R|$.

Now fix any $\alpha\in[\widehat{\alpha},1)$. For each $r\in R$, we have
\[
\min\left\{1,\;(\widehat{\alpha}-2\varepsilon)\frac{v_r}{\widehat v}\right\}
\le
\min\left\{1,\;(\alpha-2\varepsilon)\frac{v_r}{\widehat v}\right\},
\]
because $\widehat{\alpha}\le \alpha$ and the map $x\mapsto \min\{1,x\}$ is non-decreasing. Therefore, 
\begin{align*}
B - n + |R|
\leq
\sum_{i=1}^{\widehat{t}}
\min\left\{
1,\,
(\widehat{\alpha}
- 2\varepsilon )\frac{v_{r_i}}{\widehat{v}}
\right\} \leq
\sum_{i=1}^{\widehat{t}}
\min\left\{
1,\,
(\alpha
- 2\varepsilon )\frac{v_{r_i}}{\widehat{v}}
\right\}.
\end{align*}
It follows that the threshold condition required for the existence of the triad price vector is also satisfied at $\alpha$. Hence, $\mathbf p(\alpha)$ exists, and the lemma stands proved. 
\end{proof}

\subsection{Proof of \Cref{theorem:query-upper-bound-KwC}}

Finally, we restate and prove our main result for the Knapsack using Comparison Oracle (KCO) problem. 

\TheoremMainKCO*
\begin{proof}[Proof of \Cref{theorem:query-upper-bound-KwC}]
As outlined above, our algorithm for the KCO problem proceeds in three phases. The first phase (\Cref{alg:get-partition}) identifies the set $S$ of integrally packed items, the unique fractionally packed item $f$, and the set $U$ of unpacked items for the optimal solution under the price vector $\mathbf{1}_n$. The second  (\Cref{alg:get-value-ratio-for-items-in-U}) and third (\Cref{alg:get-value-ratio-of-items-in-S}) phases recover the scaled values of items in $U \cup \{f\}$ and $S$, respectively. The correctness of the three phases---and therefore of the overall algorithm---follows directly from \Cref{lemma:identification-of-sets-S-U-f,lemma:estimating-profit-of-items-in-U,lemma:performance-guarantee-of-alg-for-computing-value-ratio-S}.

Furthermore, summing the query bounds of the three phases, we obtain that the total query complexity of our KCO algorithm is $O\bigl(n + n\log(1/\delta) + B\log(1/\delta) + B^2\bigr) = O\bigl(n\log(1/\delta) + B^2\bigr)$, as stated. Similarly, the polynomial-time guarantees for the three phases established in \Cref{lemma:identification-of-sets-S-U-f,lemma:estimating-profit-of-items-in-U,lemma:performance-guarantee-of-alg-for-computing-value-ratio-S} imply that the overall algorithm also runs in polynomial time. 

Therefore, the KCO problem can be solved by our algorithm with the stated query complexity and running-time guarantees. This completes the proof of the theorem.
\end{proof}

\section{Lower Bound for Knapsack using Comparison Oracle}\label{sec:lower_bound}

In this section, we establish a lower bound on the query complexity of the {Knapsack using Comparison Oracle} (KCO) problem (see \Cref{definition:knapsack-with-comparison}) via a randomized reduction from the \textsc{Index} problem in the one-way communication model.

\subsection{One-Way Communication Complexity of \textsc{Index}}
In the \textsc{Index} problem, Alice holds a bit-string ${\bf x} \in \{0, 1\}^m$ and Bob holds an index $i^* \in [m]$. Alice sends a single message to Bob, after which Bob must output the bit $x_{i^*}$. The communication complexity of this problem is stated as follows.

\begin{restatable}[{\cite{roughgarden16communication}}]{theorem}{IndexComplexity}
\label{thm:index_complexity}
Any one-way protocol that solves the \textsc{Index} problem on a string of length $m$ with probability at least $\frac{5}{9}$ requires Alice to send a message of size $\Omega(m)$ bits.
\end{restatable}

\subsection{Lower Bound Construction}

\begin{restatable}{theorem}{QueryLowerBound}
\label{query-lower-bound}
Let $\delta < \frac{1}{6n^2}$ and $n \geq 3$. Any randomized algorithm $\mathcal{A}$ that solves the Knapsack using Comparison Oracle problem with success probability at least $2/3$ must make $\Omega \left(n\log(1/\delta)\right)$ queries to the comparison oracle $\mathcal{O}_{\bf w}$.
\end{restatable}

\begin{proof}
Let $k=\lceil \log_2(1/\delta)\rceil$ be the bit precision. We reduce from an instance of \textsc{Index} on an input of length $m=nk$. Alice and Bob share a public random string ${\bf R}\in\{0,1\}^m$, partitioned into $n$ consecutive blocks ${\bf r}^{(1)},\ldots,{\bf r}^{(n)}$, each consisting of $k$ bits. They also share the random bits used by the algorithm $\mathcal A$.

Alice receives ${\bf x}\in\{0,1\}^m$ and partitions it into equal-length blocks $\{{\bf x}^{(i)}\}_{i=1}^n$. For each $i\in[n]$, she computes the perturbed string ${\bf y}^{(i)}={\bf x}^{(i)}\oplus{\bf r}^{(i)}$. She then interprets ${\bf y}^{(i)}$ as the binary representation of an integer $z_i\in\{1,\ldots,1/\delta\}$ and sets $v_i=z_i\delta$. Additionally, Alice appends an anchor item with value $v_{n+1}=1$ to form the value vector ${\bf v}=(v_1,\ldots,v_n,v_{n+1})$. Finally, Alice sets the knapsack capacity $B = (n - 2) \cdot \min_{i \in [n + 1]} v_i \leq (n-1)(\min_{i \in [n + 1]} v_i - \delta^2/n^2)$. Note that $\max_{i \in [n+1]} v_i = 1$.

The reduction yields a valid KCO instance provided that all $n+1$ values in ${\bf v}$ are distinct and are non-zero. To establish this, note that the random strings ${\bf r}^{(1)},\ldots,{\bf r}^{(n)}$ are independent and uniformly distributed over $\{0,1\}^k$. Hence, for any distinct indices $i,j\in[n]$, we have $$\Pr \left\{ {\bf x}^{(i)}\oplus{\bf r}^{(i)} = {\bf x}^{(j)}\oplus{\bf r}^{(j)}\right\} = 1/2^k \le \delta.$$ Therefore, $\Pr \left\{ v_i = v_j \right\} \leq \delta$ for all distinct $i, j \in [n]$. In addition, for any $i\in[n]$, the probability that $v_i=v_{n+1}=1$ or that $v_i = 0$ is at most $2\delta$. Applying a union bound over all pairs of items, we conclude that the probability that the reduction fails to produce a valid instance (i.e., that two values coincide or some value is $0$) is at most
$
\Pr \bigl\{ \text{there exists } i\neq j \in [n+1] \text{ s.t. } v_i=v_j\bigr\} \leq \binom{n}{2}\delta + 2n\delta = \left(\frac{n^2-n}{2}+n\right)\delta = \frac{(n^2+3n)\delta}{2}.
$

Since $\delta < \frac{1}{9n^2}$, the above probability is at most $\frac{(n^2+3n)\delta}{2} < \frac{n^2+3n}{18n^2} = \frac{1+3/n}{18} \leq \frac{1}{9}$, $n \geq 3$. Thus, with probability at least $8/9$, the reduction maps the given \textsc{Index} instance to a valid KCO instance in which all $n+1$ values are distinct and the knapsack capacity is appropriately bounded.

Now suppose, towards a contradiction, that there exists an algorithm $\mathcal{A}$ that solves KCO using $q=o(n\log(1/\delta))$ comparison-oracle queries and the knowledge of the knapsack capacity $B$. Alice inputs $B$ to $\mathcal{A}$ and then simulates the execution of $\mathcal{A}$ locally. For each comparison query $({\bf p},{\bf p'})$ issued by $\mathcal{A}$, she explicitly computes the oracle response $b = \operatorname{sign} \left( {\bf w}^\top {\bf x}^*({\bf p}) - {\bf w}^\top {\bf x}^*({\bf p'}) \right)$ and returns it to $\mathcal{A}$. Once the simulation is complete, Alice sends to Bob the transcript $\Pi$ consisting of the sequence of oracle responses observed during the execution and the binary representation $\Gamma$ of $\min_{i \in [n+1]} v_i$. Note that since the oracle responses lie in $\{-\mathsf 1, \mathsf 0, \mathsf 1\}$, the length of $\Pi$ is at most $2q$ bits. Furthermore, the length of $\Gamma$ is $\log(1/\delta)$ since $\min_{i \in [n+1]} v_i$ is of the form $j\delta$, where $j \in \{1,2,\ldots,k-1\}$. Hence, the length of the message sent by Alice to Bob is at most $2q + \log(1/\delta)$.

Bob recovers the knapsack capacity $B = (n-2) \min_{i \in [n+1]} v_i$ using $\Gamma$, and then re-simulates $\mathcal{A}$ using the transcript $\Pi$ and $B$. Upon termination, the algorithm---correctly with probability at least $\frac{2}{3}$---outputs a scaled value vector $\widehat{{\bf v}}={\bf v}/\lambda$. Since $v_{n+1}=1$, Bob can recover the scaling factor as $\lambda=1/\widehat{v}_{n+1}$ and hence reconstruct the entire value vector via ${\bf v}=\lambda\widehat{{\bf v}}$. From the recovered values $v_i$, he obtains the corresponding perturbed strings ${\bf y}^{(i)}$ and computes ${\bf x}^{(i)}={\bf y}^{(i)}\oplus{\bf r}^{(i)}$. In particular, Bob recovers the desired bit $x_{i^*}$ from the recovered blocks.

This protocol enables Bob to solve \textsc{Index} using only at most $2q + \log(1/\delta)=o(m)$ bits of communication, where $m=nk=n\lceil \log_2(1/\delta)\rceil$. Since the reduction produces a valid KCO instance with probability at least $8/9$ and the algorithm $\mathcal{A}$ succeeds with probability at least $2/3$, Bob correctly recovers $x_{i^*}$ with probability at least $1-\left(\frac{1}{9}+\frac{1}{3}\right)=\frac{5}{9}$. This contradicts the $\Omega(m)$ one-way communication lower bound for \textsc{Index} (Theorem~\ref{thm:index_complexity}). We conclude that any algorithm solving KCO must make $q=\Omega\left(n\log(1/\delta)\right)$ comparison-oracle queries, thereby proving the theorem.
\end{proof}

\begin{remark}
Our lower-bound construction continues to hold even when the algorithm is given the comparison weight vector $\mathbf w$ used by the oracle. In fact, the lower bound is more general: for the problem of recovering the value vector, the same argument applies to any oracle that returns only a constant number of bits per query. This demonstrates that the fundamental difficulty of the problem stems from the limited amount of information available to the algorithm from each query.
\end{remark}

\section{Packing Linear Program using Comparison Oracle}
\label{section:plco}
This section extends our fractional knapsack result to the more general setting of packing linear programs (LPs). We consider a model where the algorithm may query two constraint matrices $\mathbf{P}, \mathbf{Q} \in [0,1]^{m\times n}$. For an unknown but fixed objective vector $\mathbf{v} \in (0,1]^n$ and a known capacity vector $\mathbf{b} = (b_1, \ldots, b_m) \in \mathbb{R}^m_{>0}$, let $\mathbf x^*(\mathbf P)$ and $\mathbf x^*(\mathbf Q)$ denote the optimal solutions, selected according to an unknown but fixed tie-breaking rule, to the packing programs 
\[
\max_{\mathbf x\in[0,1]^n} \ \mathbf v^\top \mathbf x \quad \text{s.t.} \quad \mathbf{P} \mathbf{x} \le \mathbf{b} 
\qquad\text{and}\qquad
\max_{\mathbf x\in[0,1]^n} \ \mathbf v^\top \mathbf x \quad \text{s.t.} \quad \mathbf{Q} \mathbf{x} \le \mathbf{b},
\]
respectively. The oracle then reports 
\begin{align*}
    \mathcal{O}_{\mathbf w}(\mathbf P, \mathbf Q) = \operatorname{sign}\big(\mathbf w^\top \mathbf x^*(\mathbf P) - \mathbf w^\top \mathbf x^*(\mathbf Q)\big),
\end{align*} 

The \emph{Packing Linear Program using Comparison Oracle} problem is to recover the value vector $\mathbf v$ by making comparison queries with pairs of constraint matrices $\mathbf P, \mathbf Q \in [0,1]^{m \times n}$. Formally, 
\begin{definition}[Packing Linear Program using Comparison Oracle (PLCO)]
\label{defn:plco}
For an unknown value vector $\mathbf{v} \in (0,1]^n$ with distinct coordinates $v_i \in \delta \mathbb{Z}_{>0}$, known capacity vector $\mathbf b \in \mathbb R^m_{>0}$ satisfying $\min_{i \in [m]} \ b_i \leq (n - 1) \left(\frac{\min_{i \in [n]} v_i}{\max_{j \in [n]} v_j} - \frac{\delta^2}{n^2}\right)$, and given query access to the comparison oracle $\mathcal O_{\mathbf w}: [0,1]^{m \times n} \times [0,1]^{m \times n} \to \{-\mathsf 1, \mathsf 0, \mathsf 1\}$, the \emph{Packing Linear Program using Comparison Oracle} ({PLCO}) is to recover a vector $\widehat {\mathbf v} = \mathbf v / \lambda$ for some scalar $\lambda > 0$.
\end{definition}

Throughout, for any constraint matrix $\mathbf P\in[0,1]^{m\times n}$ and any index $i\in[m]$, we will write $\mathbf p^{(i)}\in[0,1]^n$ to denote the $i$th row of $\mathbf P$, viewed as a vector.

The following theorem establishes our query upper bound for the PLCO problem.

\begin{restatable}[Upper Bound]{theorem}{PLCOUpperBound}
\label{thm:plco-upper-bound}
    The Packing Linear Program using Comparison Oracle problem admits a polynomial-time algorithm with query complexity $O(n \log(1/\delta) + \min_{i \in [m]} b_i^2)$.
\end{restatable}

\begin{proof}
Let $\widehat{i} \in \argmin_{i \in [n]} b_i$ and note that $b_{\widehat i} \leq (n - 1)(\frac{\min_{i \in [n]} v_i}{\max_{j \in [n]} v_j} - \frac{\delta^2}{n^2})$. Also, we have that the values $v_i \in (0,1]$ lie on the $\delta$-grid and are distinct.

 Consider the following reduction to the {KCO} problem with value vector same as the PLCO instance's value vector $\mathbf v$, and the knapsack capacity as $b_{\widehat i}$. For a comparison query $\mathbf p, \mathbf q$ issued by the KCO algorithm, we construct the matrices $\mathbf P$ and $\mathbf Q$ specified row-wise as
    \begin{align*}
        \mathbf p^{(i)} = \begin{cases}
            \mathbf 0, & i \neq \widehat i\\
            \mathbf p, & i = \widehat i
        \end{cases}\ ,
        &&
        \mathbf q^{(i)} = \begin{cases}
            \mathbf 0, & i \neq \widehat i\\
            \mathbf q, & i = \widehat i
        \end{cases}.
    \end{align*}
    The PLCO comparison oracle is then queried with the matrices $\mathbf P$ and $\mathbf Q$, and its response, $\mathcal O_{\mathbf w}(\mathbf P,\mathbf Q)$, is forwarded to the KCO algorithm.

    We next focus only on the matrix $\mathbf P$, as the same argument applies to $\mathbf Q$. For every row $i\neq\widehat{i}$, we have $\mathbf p^{(i)\top}\mathbf x=0\le b_i$ for all $\mathbf x\in[0,1]^n$. Hence, all constraints corresponding to rows $i\neq\widehat{i}$ can be omitted. Therefore, under the constraint matrix $\mathbf P$, the packing LP is equivalent to the fractional knapsack problem
\[\max_{\mathbf x\in[0,1]^n}\ \mathbf v^\top\mathbf x \qquad \text{s.t.} \qquad \mathbf p^\top\mathbf x\le b_{\widehat{i}}.\]
Therefore, the PLCO oracle's response to the query $(\mathbf P,\mathbf Q)$ is identical to the KCO oracle's response to the query $(\mathbf p,\mathbf q)$.

Therefore, by \Cref{theorem:query-upper-bound-KwC}, the KCO algorithm recovers the value vector $\mathbf v$, up to a scaling factor, in polynomial time using $O(n\log(1/\delta)+b_{\widehat{i}}^2)$ comparison queries. Since the reduction from PLCO to KCO takes polynomial time, the claimed computational efficiency follows. This completes the proof of the theorem.
\end{proof}

\begin{restatable}[Lower Bound]{theorem}{PLCOLowerBound}
\label{thm:plco-lower-bound}
    Let $\delta < \frac{1}{6n^2}$ and $n \geq 3$. Any randomized algorithm $\mathcal A$ that solves the Packing Linear Program using Comparison Oracle problem with success probability at least $2/3$ must make $\Omega(n \log(1/\delta))$ queries to the comparison oracle $\mathcal O_{\bf w}$.
    
\end{restatable}

\begin{proof}
    The proof closely follows that of \Cref{query-lower-bound}. Alice uses the same randomized reduction from \textsc{Index} to KCO to obtain a value vector $\mathbf v$ and knapsack capacity $B$. She then constructs a PLCO instance with value vector $\mathbf v$ and capacity vector $\mathbf b=B\mathbf 1_m$. As in the proof of \Cref{query-lower-bound}, the vector $\mathbf v \in (0,1]^n$ has distinct coordinates with constant probability, and each coordinate satisfies $v_i\in \delta\mathbb Z_{>0}$. Moreover, $B\le (n-1)\left(\frac{\min_{i\in[n+1]}v_i}{\max_{j\in[n+1]}v_j}-\frac{\delta^2}{n^2}\right)$. Therefore, $\mathbf v$ and $\mathbf b$ constitutes a valid PLCO instance.

Assume, towards a contradiction, that there exists a PLCO algorithm $\mathcal A$ with query complexity $q=o(n\log(1/\delta))$. Using $\mathcal A$, we construct a communication protocol that closely follows the proof of \Cref{query-lower-bound}. For each query $(\mathbf P,\mathbf Q)$ issued by $\mathcal A$, Alice simulates the comparison oracle $\mathcal O_{\mathbf w}(\mathbf P,\mathbf Q)$ and returns the resulting oracle response to $\mathcal A$. Once $\mathcal A$ terminates, Alice collects the oracle responses into a transcript $\Pi$ and sends it to Bob, together with the $\log_2(1/\delta)$-bit representation of $\min_{i\in[n+1]}\ v_i$. As shown in the proof of \Cref{query-lower-bound}, this information suffices for Bob to reconstruct $B$, and hence the capacity vector $\mathbf b=B\mathbf 1_m$. Finally, since each oracle response belongs to $\{ -\mathsf1,\mathsf0,\mathsf1 \}$, the transcript satisfies $|\Pi|\le 2q$.

Bob then re-simulates $\mathcal A$ using the reconstructed capacity vector $\mathbf b$ and the oracle responses from $\Pi$. Therefore, Bob recovers the value vector $\mathbf v$ and thereby solves the \textsc{Index} instance. By \Cref{thm:index_complexity}, Alice must communicate $\Omega(n\log(1/\delta))$ bits to Bob. However, since $|\Pi|\le 2q$, the total communication is at most $|\Pi|+\log_2(1/\delta)\le 2q+\log_2(1/\delta) = o(n\log(1/\delta))$. Since this bound contradicts \Cref{thm:index_complexity}, no such PLCO algorithm $\mathcal A$ exists, completing the proof.
\end{proof}

\section{Profit Maximization using Comparison Oracle}
\label{section:profit-maximization}
This section addresses the profit-maximization problem under comparison oracles. In the profit-maximization problem, a seller posts prices $\mathbf p \in [0,1]^n$ for $n$ divisible items to a buyer with unknown additive valuations $\mathbf v \in [0,1]^n$. Given $\mathbf p$, the buyer selects a value-maximizing bundle $\mathbf x^*(\mathbf p)$ subject to her budget, i.e., solves a fractional knapsack problem. The seller seeks to maximize the resulting profit, $(\mathbf p-\mathbf c)^\top \mathbf x^*(\mathbf p)$, where $\mathbf c\in[0,1]^n$ is a known cost vector. Throughout this section, as previously, $\mathbf x^*(\mathbf p)$ denotes the optimal fractional knapsack solution corresponding to the price vector $\mathbf p$, where ties in the value-to-price ratios are broken according to an unknown but fixed ordering of the items. Also, for every pair of queried price vectors $\mathbf p,\mathbf q\in[0,1]^n$, the comparison oracle returns $\mathcal O_{\mathbf w}(\mathbf p,\mathbf q) \coloneqq \operatorname{sign} \left( \mathbf w^\top\mathbf x^*(\mathbf p) - \mathbf w^\top\mathbf x^*(\mathbf q) \right)$.

Formally, the \emph{Profit Maximization using Comparison Oracle} problem is defined as follows.

\begin{definition}[Profit Maximization using Comparison Oracle]
\label{definition:profit-maximization}
Let $\mathbf{v}\in[0,1]^n$ be an unknown value vector with distinct coordinates $v_i\in\delta\mathbb{Z}_{>0}$, let $B>0$ be a known budget satisfying $B \le (n-1)\left(\frac{\min_{i\in[n]}v_i}{\max_{i\in[n]}v_i}-\frac{\delta^2}{n^2}\right)$, let $\mathbf c\in[0,1]^n$ be a known cost vector, and suppose we are given access to the comparison oracle $\mathcal O_{\mathbf w}(\cdot,\cdot)$. The \emph{Profit Maximization using Comparison Oracle} problem is to find a price vector $\mathbf p^* \in \argmax_{\mathbf p\in[0,1]^n} \ (\mathbf p-\mathbf c)^\top\mathbf x^*(\mathbf p)$.
\end{definition}

Fix a value vector $\mathbf v$ and budget (capacity) $B$. Let $\mathsf{OPT} \coloneqq \max_{\mathbf p\in[0,1]^n} \ (\mathbf p-\mathbf c)^\top\mathbf x^*(\mathbf p)$ denote the optimal profit. Also, for price vectors $\mathbf p \in [0,1]^n$, let $\mathcal X^*(\mathbf p)$ denote the set of optimal solutions of the corresponding fractional knapsack problem, i.e., $\mathcal X^*(\mathbf p) \coloneqq \argmax_{\mathbf x\in[0,1]^n} \left\{ \mathbf v^\top\mathbf x : \mathbf p^\top\mathbf x\le B \right\}$. 

The following theorem is the main result of this section. It shows that the {Profit Maximization using Comparison Oracle} problem can be solved to arbitrarily small additive error by a polynomial-time algorithm using $O(n\log(1/\delta)+B^2)$ comparison queries.

\begin{restatable}{theorem}{theoremprofitmax}
\label{theorem:profit-maximization}
For every $\varepsilon>0$, the Profit Maximization using Comparison Oracle problem admits a polynomial-time algorithm that, using $O(n\log(1/\delta)+B^2)$ comparison queries, computes a price vector $\widehat{\mathbf p}$ satisfying $(\widehat{\mathbf p}-\mathbf c)^\top \mathbf{x}^*(\widehat{\mathbf p}) \ge \mathsf{OPT}-\varepsilon.$ 
\end{restatable}

The proof of \Cref{theorem:profit-maximization} closely follows the approach of \cite{amin2015onlineProfitMax}. We include the details here for completeness. We begin by restating {\cite[Lemma 1]{amin2015onlineProfitMax}} and providing an alternative proof. The lemma asserts that there exists a set $\mathcal P$ of $n$ price vectors, one of which achieves the optimal profit. Note that, for price vector $\overline{\mathbf p}$, the lemma considers the best profit attainable over all (optimal) solutions in $\mathcal X^*(\overline{\mathbf p})$, rather than the particular tie-broken solution $\mathbf x^*(\overline{\mathbf p})$. Furthermore, the set $\mathcal P$ can be constructed once the value ratios are known.

\begin{lemma}[{\cite[Lemma 1]{amin2015onlineProfitMax}}]
\label{lemma:optimal-price-lies-in-P-structural-result}
For each index $k\in[n]$, define the price vector $\mathbf p^k$ by $p_i^k=\min\left\{\frac{v_i}{v_k},\,1\right\}$, for all $i\in[n]$. Let $\mathcal P\coloneqq\{\mathbf p^k:k\in[n]\}$. Then, there exists  $\overline{\mathbf p}\in\mathcal P$ such that 
\[
\max_{\mathbf x\in\mathcal X^*(\overline{\mathbf p})}
\ (\overline{\mathbf p}-\mathbf c)^\top\mathbf x
\ge
\mathsf{OPT}.
\]
\end{lemma}

\begin{proof}
Let $\mathbf p^* \in \argmax_{\mathbf p\in[0,1]^n} \ (\mathbf p-\mathbf c)^\top\mathbf x^*(\mathbf p)$ denote an optimal price vector for the profit-maximization problem. For the associated optimal solution $x^*(\mathbf p^*) \in \mathcal{X}^* (\mathbf p^*)$, let $O^* \coloneqq \{i \in [n] \ : \ x^*_i(\mathbf p^*) >0 \}$ be the set of (integrally or fractionally) packed items. 

We will identify a price vector $\overline{\mathbf p}\in\mathcal P$ and a solution $\overline{\mathbf y}\in\mathcal X^*(\overline{\mathbf p})$ such that $ (\overline{\mathbf p}-\mathbf c)^\top\overline{\mathbf y} \geq (\mathbf p^*-\mathbf c)^\top\mathbf x^*(\mathbf p^*) = \mathsf{OPT}$, thereby establishing the lemma. Toward this end, let $k$ be the highest-valued item outside $O^*$, i.e., $k=\argmax_{i\notin O^*}v_i$, and let $f\in O^*$ denote a packed item with the smallest value-to-price ratio under $\mathbf p^*$. Also, let threshold $\tau^* \coloneqq \max\{v_f, v_k\}$. Define the price vector $\overline{\mathbf p}$ by
\[
\overline p_i=\min\left\{\frac{v_i}{\tau^*},\,1\right\} 
\qquad \text{for all } i\in[n].
\]
Since $\tau^*\in\{v_f,v_k\}$, we have $\overline{\mathbf p}\in\mathcal P$.

We next show that under the price vector $\overline{\mathbf p}$, every item $i\in O^*$ has value-to-price ratio at least $\tau^*$, whereas every item $j\notin O^*$ has value-to-price ratio at most $\tau^*$. Furthermore, $\overline{p}_i \geq p^*_i$ for every item $i \in O^*$.

Note that every item $j\notin O^*$ satisfies $v_j\le v_k \leq \tau^*$. Hence, $\overline p_j=v_j/\tau^*$, and therefore every such item has value-to-price ratio equal to $\tau^*$. Now consider an item $i\in O^*$. If $v_i\le\tau^*$, then $\overline p_i=v_i/\tau^*$, and hence its value-to-price ratio is $\tau^*$. Otherwise, $v_i>\tau^*$, in which case $\overline p_i=1$, and its value-to-price ratio is $v_i>\tau^*$. Therefore, under $\overline{\mathbf p}$, every item in $O^*$ has value-to-price ratio at least $\tau^*$, whereas every item outside $O^*$ has value-to-price ratio at most $\tau^*$. Moreover, the value-to-price ratio of $f$ is exactly $\tau^*$. Consequently,
\begin{align} 
\frac{v_i}{\overline p_i} \ge \frac{v_f}{\overline p_f} & \ge \frac{v_j}{\overline p_j}, \qquad \text{for all } i\in O^*,\ j\notin O^* \label{ineq:ties} 
\end{align}

Next, observe that the greedy construction of $\mathbf x^*(\mathbf p^*)$ ensures that every packed item $i\in O^*$ satisfies $v_i/p_i^*\ge v_f/p_f^*\ge v_k/p_k^*$. Since $p_f^*,p_k^*\le1$, it follows that $v_i/p_i^*\ge\max\{v_f,v_k\}=\tau^*$. Hence, $p^*_i \leq v_i/\tau^*$. Combining this with the bound $p^*_i \leq 1$, we obtain 
\begin{align}
    p^*_i & \leq \min \left\{ \frac{v_i}{\tau^*}, 1 \right\} = \overline{p}_i \qquad \text{for all } i \in O^*\label{ineq:go}
\end{align} 

Inequalities \Cref{ineq:ties,ineq:go} show that, by breaking ties appropriately, we can construct an optimal solution $\overline{\mathbf y}\in\mathcal X^*(\overline{\mathbf p})$ such that $\overline{\mathbf y}\le\mathbf x^*(\mathbf p^*)$ coordinate-wise. Indeed, $x_i^*(\mathbf p^*)=1$ for every $i\in O^*\setminus\{f\}$ and $0<x_f^*(\mathbf p^*)\le1$. Therefore, under $\overline{\mathbf p}$, by breaking ties in favor of items in $O^*\setminus\{f\}$ and then $f$, we obtain an optimal solution $\overline{\mathbf y}$. Since the prices of items in $O^*$ do not decrease when moving from $\mathbf p^*$ to $\overline{\mathbf p}$, and ties are broken so that $f$ is considered last, every item is packed to no greater extent than under $\mathbf p^*$. Hence, $\overline{\mathbf y}\le\mathbf x^*(\mathbf p^*)$ coordinate-wise. Therefore, 
\begin{align*}
(\overline{\mathbf p}-\mathbf c)^\top\overline{\mathbf y} = B-\mathbf c^\top\overline{\mathbf y} \ge B-\mathbf c^\top\mathbf x^*(\mathbf p^*) = (\mathbf p^*-\mathbf c)^\top\mathbf x^*(\mathbf p^*) = \mathsf{OPT}.
\end{align*}
This completes the proof of the lemma.
\end{proof}

\begin{proof}[Proof of \Cref{theorem:profit-maximization}] 
For the given profit-maximization instance, we first recover the scaled values of all items using Algorithms~\ref{alg:get-partition}, \ref{alg:get-value-ratio-for-items-in-U}, and \ref{alg:get-value-ratio-of-items-in-S}. That is, we recover the value ratios specified in \Cref{theorem:query-upper-bound-KwC} using $O(n\log(1/\delta)+B^2)$ comparison queries.

Next, observe that each price vector in the set $\mathcal P$, defined in \Cref{lemma:optimal-price-lies-in-P-structural-result}, can be constructed from the recovered value ratios. We then iterate over the $n$ price vectors in $\mathcal P$ and, for each $\mathbf p\in\mathcal P$, compute $\max_{\mathbf x\in\mathcal X^*(\mathbf p)} \mathbf x^\top(\mathbf p-\mathbf c)$ in $O(n\log n)$ time using {\cite[Lemma 2]{amin2015onlineProfitMax}}. Consequently, we identify a price vector $\overline{\mathbf p}\in\mathcal P$ satisfying \Cref{lemma:optimal-price-lies-in-P-structural-result} in $O(n^2\log n)$ time. That is, we efficiently identify a price vector $\overline{\mathbf p}\in\mathcal P$ satisfying $\max_{\mathbf x\in\mathcal X^*(\overline{\mathbf p})} (\overline{\mathbf p}-\mathbf c)^\top\mathbf x \ge \mathsf{OPT}$.

Note, however, that this profit guarantee holds for some solution in $\mathcal X^*(\overline{\mathbf p})$, which need not coincide with the solution $\mathbf x^*(\overline{\mathbf p})$ selected under the oracle's tie-breaking rule. We address this issue via \Cref{lemma:price-perturbation-to-induce-optimal-bundle} (stated and proved below), which allows us to construct, in $O(n\log n)$ time, a price vector $\widehat{\mathbf p}$ satisfying
\begin{align*}
(\widehat{\mathbf p}-\mathbf c)^\top\mathbf x^*(\widehat{\mathbf p}) &\ge \max_{\mathbf x\in\mathcal X^*(\overline{\mathbf p})} (\overline{\mathbf p}-\mathbf c)^\top\mathbf x -\varepsilon \\ &\ge \mathsf{OPT}-\varepsilon.
\end{align*}
Therefore, we obtain a polynomial-time algorithm that computes a price vector whose profit is within $\varepsilon$ of the optimum using $O(n\log(1/\delta)+B^2)$ comparison queries. This establishes the theorem.
\end{proof}

\begin{lemma}[{\cite[Lemma 3]{amin2015onlineProfitMax}}]
\label{lemma:price-perturbation-to-induce-optimal-bundle}
    For price vector $\overline{\mathbf p}$ define in \Cref{lemma:optimal-price-lies-in-P-structural-result}, and given $\varepsilon > 0$, there exists a perturbation $\widehat{\mathbf p}$ of $\overline{\mathbf p}$ such that $\sum_{i=1}^n (\overline p_i - \widehat p_i) \leq \varepsilon$, and the following holds:
    \begin{itemize}
    \item[(i)] For any two items $i$ and $j$, $i \neq j$, with the same value-to-price ratio under $\overline{\mathbf p}$, it holds that under $\widehat{\mathbf p}$, the value-to-price of $i$ is larger than that of $j$ if $c_i < c_j$.
    \item[(ii)] For any two items $i$ and $j$ such that $v_i/\overline p_i > v_j/\overline p_j$, it holds that $v_i/\widehat p_i > v_j / \widehat p_j$.
    \item[(iii)] $(\widehat{\bf p} - \mathbf c)^\top \mathbf x^*(\widehat{\mathbf p}) \geq \max_{\mathbf x \in \mathcal{X}^*(\overline{\mathbf p})}\ (\overline{\bf p} - \mathbf c)^\top \mathbf x - \varepsilon$.
    \end{itemize}
    Moreover, $\widehat{\bf p}$ can be computed in polynomial-time.
\end{lemma}

\begin{proof}
    Under $\overline{\mathbf p}$, let the non-increasing value-to-price ordering of the $n$ items be $\{i_1, i_2, \ldots, i_n\}$ where ties are broken in favour of larger cost $c_i$, i.e., if two items $i$ and $j$ have the same value-to-price ratio, then $i$ gets a smaller index in the ordering than $j$ if $c_i > c_j$ (in case of $c_i = c_j$ ties are broken arbitrarily).

Consider the following perturbation of price vector $\overline{\mathbf p}$. For $\ell \in [n]$, define
\[
\widehat p_{i_{ \ell}} \coloneqq \overline{p}_{i_{\ell}} \big(1 - (\ell \varepsilon \delta/n^3)\big).
\]
In other words, we reduce the price of item $i_{\ell}$ by an amount $\ell \varepsilon \delta/n^2$. Hence, we have
\begin{align*}
    \sum_{\ell=1}^n \overline p_{i_\ell} - \widehat p_{i_\ell} = \sum_{\ell=1}^n \overline p_{i_\ell} (\ell \varepsilon \delta/n^3) \leq \sum_{\ell=1}^n \overline p_{i_\ell} (\varepsilon \delta / n^2) \leq \varepsilon \delta / n,
\end{align*}
where the first inequality uses that $\ell \leq n$ for all $i \in [n]$ and the second inequality is due to $\overline p_i \leq 1$ for all $i \in [n]$. Therefore, indeed $\sum_{i \in [n]} \overline{p}_i - \widehat p_i \leq \varepsilon/2$.

For \textbf{part (i)}, observe that if for items $i_\ell$ and $i_{\ell'}$, $v_{i_\ell}/\overline p_{i_\ell} = v_{i_{\ell'}}/\overline p_{i_{\ell'}}$, then $\ell < \ell'$ if $c_{i_\ell} > c_{i_{\ell'}}$. However, we have
\begin{align*}
\frac{v_{i_{\ell'}}}{\widehat p_{i_{\ell'}}} - \frac{v_{i_{\ell}}}{\widehat p_{i_{\ell}}} &= \frac{v_{i_{\ell'}}}{\overline p_{i_{\ell'}}(1 - (\ell' \varepsilon \delta/n^3))} - \frac{v_{i_{\ell}}}{\overline p_{i_{\ell}} (1 - (\ell \varepsilon \delta / n^3))} \\
&= \frac{v_{i_\ell}}{\overline{p}_{i_\ell}} \left(\frac{1}{1 - (\ell' \varepsilon \delta/ n^3)} - \frac{1}{1 - (\ell \varepsilon \delta/ n^3})\right) \\
&>0 \tag{$\ell < \ell'$}
\end{align*}
Hence, for two items $i$ and $j$ with the same value-to-price ratio under $\overline{\bf p}$, it holds that item $i$ has higher value-to-price ratio than item $j$ under $\widehat{\bf p}$ if $c_i < c_j$.

For \textbf{part (ii)}, note that for any $i \in [n]$, we have $v_i/\overline{p}_i = \max\{v_i, v_{j^*}\}$ for some $j^* \in [n]$. Thus, it holds that if $v_i/\overline{p}_i > v_{i'}/\overline{p}_{i'}$, then $(v_i/\overline{p}_i) - (v_{i'}/\overline{p}_{i'}) \ge \delta$. Let $\widehat p_i = \overline p_i (1 - (\ell \varepsilon \delta/n^2))$ and $\widehat p_{i'} = \overline p_{i'} (1 - (\ell' \varepsilon \delta/n^2))$ for some $\ell, \ell' \in [n]$. Hence, we have
\begin{align*}
    \frac{v_i}{\widehat p_i} - \frac{v_{i'}}{\widehat p_{i'}} &= \frac{v_i}{\overline p_i (1 - (\ell \varepsilon \delta/n^3))} - \frac{v_{i'}}{\overline p_{i'} (1 - (\ell' \varepsilon \delta/n^3))}\\
    &\geq \frac{v_i}{\overline{p}_i} - \frac{v_{i'}}{\overline p_{i'} (1 - (\varepsilon \delta/n^2))}  \tag{$1 \leq \ell, \ell' \leq n$}\\
    &\geq \frac{v_i}{\overline{p}_i} - \frac{v_{i'}}{\overline p_{i'}} (1 + (2 \varepsilon \delta/n^2)) \tag{$1/(1 - x) \leq 1 + 2x$ for $x \in [0, 1/2]$}\\
    &= \frac{v_i}{\overline{p}_i} - \frac{v_{i'}}{\overline p_{i'}} - \max\{v_{i'}, v_{j^*}\} \frac{2 \varepsilon \delta}{n^2} \\
    &\geq \delta - \frac{2 \varepsilon \delta}{n^2} \tag{$v_t \leq 1$ for all $t \in [n]$}\\
    &>0.
\end{align*}

For \textbf{part (iii)}, consider two distinct elements $\mathbf x $ and $ \mathbf x'$ in the set $\mathcal X(\mathbf v, \overline{\mathbf p}, B)$. Both of them, by virtue of being the optimal solution for the knapsack problem, starts by packing the items in non-increasing order of value-to-price ratio. The packing in $\mathbf x$ and $\mathbf x'$ remains identical as items are packed integrally, till the remaining knapsack capacity $B' < \sum_{\ell = 1}^t \overline p_{j_t}$, where all items in $\{j_1, \ldots, j_t\}$ have the same value-to-price ratio, which is different from items in $[n] \setminus \{j_1, \ldots, j_t\}$. At this point $\mathbf x$ and $\mathbf x'$ pack items differently. For the cost optimal solution $\overline{\mathbf y} \in \argmax_{\mathbf x \in \mathcal{X}(\mathbf v, \overline{\mathbf p}, B)}$, the order of packing items in $\{j_1, \ldots, j_t\}$ must be in non-decreasing order of cost $c_{j_i}$, $i \in [t]$. Between items $j_i$ and $j_{i'}$ with the same cost, let $\overline{\mathbf y}$ break tie in favour of the item that appears before the other in the greedy ordering under $\widehat{\mathbf p}$.

Consider the value-to-price ordering of $\{j_1, \ldots, j_t\}$ under $\widehat{\bf p}$. Note that items that had strictly higher and strictly lower value-to-price ratio than $j_1$ under $\overline{\bf p}$ continues to maintain this relative ordering under $\widehat{\bf p}$ (shown in part (ii)). For each item in $\{j_1, \ldots, j_t\}$, the value-to-price ratio under $\widehat{\bf p}$ is unique, hence the packing order of these items are unique. Moreover, this order is nothing but the increasing order of cost $c_{j_i}$, $i \in [t]$, as shown in part (i). Therefore, $\mathcal X(\mathbf v, \widehat{\bf p}, B)$ is a singleton containing $\mathbf x^*(\widehat{\bf p})$. Note that upto the last integrally packed item in $\overline{\bf y}$, the packing of $\mathbf x^*(\widehat{\bf p})$ is identical. This is because the ordering of items is the same in both the cases, whereas $\widehat{p}_i \leq \overline{p}_i$. Note that in $\mathbf x^*(\widehat{\bf p})$, after packing the same set of items integrally as $\mathbf y$, an extra capacity of at most $\sum_{i=1}^n \overline{p}_i - \widehat{p}_i \leq \varepsilon \delta/n$ remains. Since for any $i$, $\widehat{p}_i = \overline{p}_i(1 - (\ell \varepsilon \delta/ n^3))$ for some $\ell \in [n]$, it holds that
\begin{align}
\label{eq:max-packing-possible}
    \frac{\varepsilon \delta / n}{\widehat{p}_i} = \frac{\varepsilon \delta / n}{\overline{p}_i(1 - (\ell \varepsilon \delta/ n^3))} \leq \frac{\varepsilon}{n (1 - (\varepsilon \delta/ n^2))} \leq 2\varepsilon/n,
\end{align}
where the first inequality is because $\overline{p}_i = \min\{v_i/v_{j^*}, 1\} \geq \delta$ and $\ell \leq n$, and the last inequality is via $\varepsilon \delta/ n^2 \leq 1/2$. Hence, a total of at most $2 \varepsilon/n$ quantity of any item can be packed using the remaining capacity of at most $\varepsilon \delta/n$.

Let $I = \{i \in [n]: \overline y_i = 1\}$ and $f$ be the item such that $\overline y_f \in (0,1)$. Further, let $u$ be the item appearing right after $f$ in the value-to-price ordering under $\widehat{\bf p}$. Note that after packing all items in $I$ integrally, $\mathbf x^*(\widehat{\bf p})$ is left with a capacity of at most 
$$\overline p_f \overline y_f + \varepsilon\delta/n \leq \widehat{p}_f + 3 \varepsilon \delta/n \leq \widehat{p}_f + \widehat{p}_u,$$ 
where we have used $\overline{p}_i = \widehat{p}_i/(1 - (\ell \varepsilon \delta/n^3)) \leq \widehat{p}_i + 2\delta\varepsilon/n$. The last step $\widehat{p}_u = \overline{p}_u (1 - (\ell \varepsilon \delta/n^3)) \geq \overline{p}_u/2 \geq \delta/2$. Hence, at most $f$ and $u$ can be integrally packed, and no other items appearing after $u$ can be packed in $\mathbf x^*(\widehat{\bf p})$. 

Next, we have $x^*_f(\widehat{\bf p}) \geq \overline{y}_f$ since prices are only lower under $\widehat{\bf p}$, and by \Cref{eq:max-packing-possible}, $x^*_u(\widehat{\bf p})$. Moreover, $x^*_f \leq \frac{\overline{p}_f \overline{y}_f + \varepsilon \delta/n}{\widehat{p}_f} \leq \overline{y}_f + (2 \overline{y}_f \varepsilon \delta/n^2) + 2\varepsilon/n$, where we have used $1/(1 - x) \leq 1 + 2x$ for $x \in [0,1/2]$ and \Cref{eq:max-packing-possible}. Then, we have
\begin{align*}
    (\widehat{\mathbf p} - \mathbf c)^\top \mathbf x^*(\widehat{\mathbf p}) - (\overline{\mathbf p} - \mathbf c)^\top \overline{\mathbf y} &= \left(\sum_{i \in I} \widehat{p}_i - \overline{p}_i\right) + (\widehat{p}_f - c_f)  x^*_f(\widehat{\bf p}) - (\overline{p}_f - c_f) \overline{y}_f + (\widehat{p}_u - c_u)  x^*_u(\widehat{\bf p})\\
    &\geq -\varepsilon\delta/n + (\overline{p}_f - c_f)(x^*_f(\widehat{\bf p}) - \overline{y}_f) - (\ell \varepsilon \delta/n^3) x^*_f(\widehat{\bf p}) - c_u x^*_u(\widehat{\bf p}) \\
    &\geq -\varepsilon\delta/n - c_f(x^*_f(\widehat{\bf y}) - \overline{y}_f) - \varepsilon \delta/n^2 - c_u x^*_u(\widehat{\bf p}) \tag{$x^*_f(\widehat{\bf p}) \geq \overline{y}_f$, $\ell \leq n$} \\
    &\geq -\varepsilon\delta/n - \varepsilon \delta/n^2 - c_f \left((2 \overline{y}_f \varepsilon \delta/n^2) + 2\varepsilon/n\right) - c_u (2\varepsilon/n) \tag{$x^*_u(\widehat{\bf p}) \leq 2\varepsilon/n$} \\
    &\geq - 5 \varepsilon/n \tag{$0 \leq c_f, c_u \leq 1$, $3\delta / n \leq 1$}
\end{align*}

Finally, note that the construction of $\widehat{\bf p}$ only requires sorting the items by value-to-price ratio, and among the items with equal value-to-price ratio, by cost. This step requires $O(n \log n)$ time. Lastly, to assign $\widehat{p}_i$ to every item $i \in [n]$ requires $O(n)$ time. Hence, a total of $O(n \log n)$ time is required to construct $\widehat{\bf p}$.
This concludes the proof of the lemma.
\end{proof}

\section{Conclusion and Future Work}
\label{section:conclusion}
This work establishes that, despite observing only pairwise comparisons between optimal solutions, it is possible to recover the underlying objective of the fractional knapsack problem and, more generally, packing linear programs in a query-efficient manner. Two particularly relevant instantiations of the comparison-oracle model are those in which the oracle compares either the total sum of the coordinates of the optimal solutions or their objective values. Our algorithms are built upon structural properties of optimal packings with carefully designed price vectors for comparison queries. We complement this algorithmic result with a query lower bound that matches the upper bound up to a linear factor.

Beyond these specific algorithmic insights, our results highlight the power of ordinal information in important linear optimization problems.

An interesting direction for future work is to understand whether similar query-efficient algorithms can be obtained for general linear programs and convex optimization problems. It would also be interesting to understand the role of adaptivity in this setup and to study noisy oracle models.

\section*{Acknowledgments}
Siddharth Barman and Nirjhar Das gratefully acknowledges the support of the Walmart Centre for Tech Excellence (CSR WMGT-23-0001) and Ittiam CSR Grant (OD/OTHR-24-0032).

\addcontentsline{toc}{section}{References}
\bibliographystyle{alpha}
\bibliography{references}

\newpage
\appendix
\section{Simplest Intervening Rational via Continued Fractions}
\label{appendix:Stern-Brocot}

\begin{theorem} 
\label{thm:rational-reconstruction-short-interval}
Let $N\in\mathbb Z_{>0}$, and let $I=[a,b]$ be an interval satisfying $b-a<1/N^2$. Suppose that $I$ contains a rational number $p^*/q^*$ with $p^*,q^*\in\{1,2,\ldots,N\}$. Then $p^*/q^*$ is unique. Moreover, the unique rational $p^*/q^*$ can be identified in time polynomial in $\log N$ and the bit complexity of the endpoints $a$ and $b$.
\end{theorem}

\begin{proof}
The uniqueness of the rational $p^*/q^*$ follows from \Cref{lemma:unique-p/q-in-interval-of-length-delta^2}; in particular, for distinct rationals $p/q\neq p'/q'$ with $p,q,p',q' \in \{1, 2, \ldots, N\}$ it holds that 
\[
    \left|\frac pq-\frac{p'}{q'}\right| = \frac{|pq'-p'q|}{qq'} \ge \frac1{N^2}.
\]
Hence, interval $I$ contains at most one rational number with numerator and denominator at most $N$. Therefore, recovering $p^*/q^*$ 
is equivalent to computing the rational number of minimum denominator contained in $I$. The claimed recovery then follows from the continued-fraction based method of Murakami~\cite{murakami2010continued} (see also Stern-Brocot algorithm \cite{bosma2012continued}), which finds a rational number of minimum denominator contained in a given interval. This rational is referred to as the simplest intervening rational. The time complexity of this algorithm in polynomial in $\log N$ and the bit complexity of the interval's endpoints. This completes the proof.  
\end{proof}

\section{Towards Necessity of Bound on Capacity}
\label{appendix:bounded-capacity}
In this section, we present two examples showing that the bounded-capacity assumption is essentially necessary. For a value vector $\mathbf v$, let $ v_{\min}:=\min_{i\in[n]} v_i \text{ and } v_{\max}:=\max_{i\in[n]} v_i .$ We show that if the capacity $B$ exceeds the threshold $(n-1)v_{\min}/v_{\max}$ by even a small multiplicative factor, then value recovery from comparison queries may become impossible. More precisely, we construct KCO instances with $B=(1+o(1))(n-1)\frac{v_{\min}}{v_{\max}}$ for which there exist two non-proportional value vectors that induce identical comparison-oracle responses. Hence, no algorithm can recover the value vector $\mathbf v$ up to scaling in these instances.

\begin{example}
\label{example:large-capacity-indistinguishability}
There exists a family of KCO instances with $ \frac{v_{\min}}{v_{\max}}=1-o(1)$ and $B=(1+o(1))(n-1)\frac{v_{\min}}{v_{\max}},$ for which the value vector cannot be recovered from comparison queries.
\begin{proof}
Consider an instance in which
$
    v_{\max}:=\max_{i\in[n]} v_i = 1\text{ and } 
    v_{\min}:=\min_{i\in[n]} v_i = 1-\frac{1}{n^2},
$
with all item values distinct and lying in the interval
$
    \left[1-\frac{1}{n^2},1\right].
$
For this instance, the bound on the capacity equals
$
    (n-1)\left(\frac{v_{\min}}{v_{\max}}-\frac{\delta^2}{n^2}\right)
    =
    (n-1)\left(1-\frac{1+\delta^2}{n^2}\right)
    =
    n-1-O\left(\frac1n\right).
$

With the multiplicative factor $1 + o(1)$ chosen as $1 + O(1/n)$, we have
$
    B
    =
    (1+o(1))(n-1)\left(\frac{v_{\min}}{v_{\max}}-\frac{\delta^2}{n^2}\right)
    \ge n.
$

Consequently, for every price vector $\mathbf p\in[0,1]^n$, the all-ones vector $\mathbf 1_n$ is feasible for the knapsack LP, since
$
    \mathbf p^\top \mathbf 1_n
    =
    \sum_{i=1}^n p_i
    \le n
    \le B.
$
Furthermore, since all item values are positive, the unique optimal solution is to pack every item fully. Hence
$
    \mathbf x^*(\mathbf p)=\mathbf 1_n
     \text{ for every }\mathbf p\in[0,1]^n.
$
Therefore, for any queried price vectors $\mathbf p,\mathbf q\in[0,1]^n$,
\[
    \mathbf w^\top \mathbf x^*(\mathbf p)
    =
    \mathbf w^\top \mathbf 1_n
    =
    \mathbf w^\top \mathbf x^*(\mathbf q).
\]
Thus, every comparison query returns $\mathsf 0$, independent of the underlying value vector. Consequently, the oracle responses are identical for all positive value vectors satisfying this feasibility condition (i.e., all item values being distinct and lying in the interval $[1, 1-1/n^2]$ with $v_{\max} = 1$ and $v_{\min} = 1 - 1/n^2$), and hence, the item values cannot be recovered.
\end{proof}
\end{example}

\begin{example}
\label{example:small-value-ratio-indistinguishability}
There exists a family of KCO instances with $\frac{v_{\min}}{v_{\max}}=o(1)$ and $B=(1+o(1))(n-1)\frac{v_{\min}}{v_{\max}},$ for which the value vector cannot be recovered from comparison queries.

\begin{proof}
Let $T=T(n)$ be a parameter satisfying $T\to\infty$ as $n \to \infty$ and $T + 3\le n/2$; for instance, one may take $T=\sqrt n$ or, more generally, $T=n^{1-\eta}$ for any fixed $\eta\in(0,1)$. Consider two value vectors $\mathbf v^1,\mathbf v^2\in[0,1]^n$ defined as follows:
$
    v^1_1=\frac1T,
    v^2_1=\frac{1-\frac1n}{T},$
and, for every $i\in\{2,\ldots,n\}$,
$
    v^1_i=v^2_i=\frac1n+\frac{n-i}{n^3}.
$

Thus, the two value vectors agree on all coordinates except the first one, and hence they are not positive scalar multiples of each other.

For both instances, the minimum value is
$
    v_{\min}=\frac1n,
$
while, since $T\le n/2$, the maximum value is attained by item $1$ and is
$
    v_{\max}=\frac1T.
$
Therefore, due to the fact that $\delta^2/T \ll 1$, we have,
$$
    (n-1)\left(\frac{v_{\min}}{v_{\max}} - \frac{\delta^2}{n^2}\right)
    =
    (n-1)\left(\frac{T}{n} - \frac{\delta^2}{n^2}\right)
    \geq
    T\left(1-\frac2n\right).
$$
Hence,
$
    B
    =
    \left(1+o(1)\right)(n-1)\left(\frac{v_{\min}}{v_{\max}} - \frac{\delta^2}{n^2} \right) \geq (1 + o(1)) T (1 - 2/n).
$
An appropriate choice of $o(1)$, e.g., $3/T$, (note that $T \to \infty$ as $n \to \infty$) allows the value to $B$
\[
B = \left(1 + \frac{3}{T}\right) T \left(1 - \frac{2}{n}\right) = T + 3 - \frac{2(T + 3)}{n} \geq T + 2.
\]
where we have used $T + 3 \leq n/2$.

We now show that the two value vectors $\mathbf v^1$ and $\mathbf v^2$ are indistinguishable using comparison queries. Fix an arbitrary price vector $\mathbf p\in[0,1]^n$. Let $\mathbf x^1(\mathbf p)$ and $\mathbf x^2(\mathbf p)$ denote the greedy optimal fractional-knapsack solutions under the value vectors $\mathbf v^1$ and $\mathbf v^2$, respectively.

The relative ordering of the items in $\{2,\ldots,n\}$ by value-to-price ratio is the same in both instances, since their values are identical in the two instances. The only difference is that item $1$ has value $1/T$ in the first instance and value $(1-1/n)/T$ in the second instance. Hence, relative to the common ordering of the items in $\{2,\ldots,n\}$, item $1$ can only move later in the second instance.

Thus, for some ordering $i_1,\ldots,i_{n-1}$ of the items in $\{2,\ldots,n\}$ and some integers $k,\ell\ge 0$, the greedy orderings in the two instances have the form
$
    \text{Instance }1:
    i_1,i_2,\ldots,i_k,1,i_{k+1},\ldots,i_{n-1},
$
and
$
    \text{Instance }2:  i_1,i_2,\ldots,i_k,i_{k+1},\ldots,i_{k+\ell},1,
    i_{k+\ell+1},\ldots,i_{n-1}.
$
We claim that all items in
    $\{i_1,\ldots,i_{k+\ell},1\}$
can be packed integrally in both instances.

For any $j\le k$, since $i_j$ appears before item $1$ in the first instance, we have
\(
    \frac{\frac1n+\frac{n-i_j}{n^3}}{p_{i_j}}
    \ge
    \frac{1/T}{p_1},
\) or equivalently,
\[
    p_{i_j}
    \le
    T\left(\frac1n+\frac{n-i_j}{n^3}\right)p_1 \le 1,
\]
where the last inequality is via $T \leq n/2$. Similarly, for any $j\in\{k+1,\ldots,k+\ell\}$, since $i_j$ appears before item $1$ in the second instance, we have
\(
    \frac{\frac1n+\frac{n-i_j}{n^3}}{p_{i_j}}
    \ge
    \frac{(1-\frac1n)/T}{p_1},
\)
or equivalently,
\[
    p_{i_j}
    \le
    \frac{T}{1-\frac1n}
    \left(\frac1n+\frac{n-i_j}{n^3}\right)p_1 \leq 1,
\]
where the last inequality is via $T \leq n/2$. Combining the two bounds and using $p_1\le 1$, we obtain
\begin{align*}
    p_1+\sum_{j=1}^{k+\ell}p_{i_j}
    &\le
    p_1\left(
        1+
        \frac{T}{1-\frac1n}
        \sum_{j=1}^{k+\ell}
        \left(\frac1n+\frac{n-i_j}{n^3}\right)
    \right) \\
    &\le
    1+
    \frac{T}{1-\frac1n}
    \sum_{i=2}^{n}
    \left(\frac1n+\frac{n-i}{n^3}\right) \\
    &= 1 + \frac{T}{1 - 1/n} \left(\frac{n-1}{n}
    +
    \frac{(n-1)(n-2)}{2n^3}\right) \\
    &= 1 + T + \frac{T(n-2)}{2n^2}\\
    &\leq T + 5/4 \tag{$T \leq n/2$}\\
    &< B.
\end{align*}
Hence the items $\{i_1,\ldots,i_{k+\ell},1\}$ are packed integrally in both instances.

After these items have been packed, the remaining items appear in the same order in both instances, namely
$
    i_{k+\ell+1},\ldots,i_{n-1},
$
and the remaining capacity is also the same in both instances. Since these remaining items have identical values in the two instances, the greedy algorithm proceeds identically from this point onward. Thus, for every price vector $\mathbf p\in[0,1]^n$,
$
    \mathbf x^1(\mathbf p)=\mathbf x^2(\mathbf p).
$

Consequently, for any two queried price vectors $\mathbf p,\mathbf q\in[0,1]^n$ and any comparison weight vector $\mathbf w\in\mathbb R^n_{>0}$, we have
\[
    \operatorname{sign}\left(
        \mathbf w^\top \mathbf x^1(\mathbf p)
        -
        \mathbf w^\top \mathbf x^1(\mathbf q)
    \right)
    =
    \operatorname{sign}\left(
        \mathbf w^\top \mathbf x^2(\mathbf p)
        -
        \mathbf w^\top \mathbf x^2(\mathbf q)
    \right).
\]
Thus, the comparison oracle's responses are identical on $\mathbf v^1$ and $\mathbf v^2$, even though $\mathbf v^1 \neq \lambda \mathbf v^2$ for any $\lambda > 0$. Therefore, in this regime, the value vector cannot be recovered from comparison queries.
\end{proof}
\end{example}

\section{Existence of a Fractionally Packed Item under Uniform Perturbation}
\label{appendix:BZwlog}

The algorithms in this work are described using
the all-ones price vector \(\mathbf 1_n\). When \(B\notin\mathbb Z\), the
greedy solution under \(\mathbf 1_n\) contains a unique fractionally packed
item. If \(B\in\mathbb Z\), this need not be true, since the buyer may pack an
integer number of items exactly. The following lemma shows that this issue can
be avoided by replacing \(\mathbf 1_n\) with a uniform price vector
\(\rho\mathbf 1_n\), where \(\rho\) is chosen arbitrarily close to \(1\).

\begin{lemma} \label{lemma:fractional-item-exists-at-uniform-size} Let \(B\in\mathbb Z\) with \(B<n\). Then there exists a scalar \(\rho\in(0,1)\), arbitrarily close to \(1\), such that \(B/\rho\notin \mathbb Z\). Consequently, under the price vector \(\rho\mathbf 1_n\), the optimal fractional solution \(\mathbf x^*(\rho\mathbf 1_n)\) contains a unique fractionally packed item. Moreover, if the items are ordered so that \(v_{(1)}>v_{(2)}>\cdots>v_{(n)}\), and \(m\coloneqq \lfloor B/\rho\rfloor\), then the unique fractionally packed item is \((m+1)\), with \[ x^*_{(m+1)}(\rho\mathbf 1_n) = \frac{B-m\rho}{\rho} \in(0,1). \] \end{lemma} \begin{proof} Set \( \rho=\frac{aB+1}{aB+2}, \) where \(a\) is a sufficiently large integer. Then \(\rho\in(0,1)\), and \(\rho\) can be chosen arbitrarily close to \(1\) by choosing \(a\) appropriately. Moreover, \[ \frac{B}{\rho} = B\left(\frac{aB+2}{aB+1}\right) = B+\frac{B}{aB+1}. \] For sufficiently large \(a\), we have \(0<\frac{B}{aB+1}<1\), and hence \(B/\rho\notin\mathbb Z\). Now consider the optimal solution \(\mathbf x^*(\rho\mathbf 1_n)\). Since every item has price \(\rho\), the value-to-price ratio of item \(i\) is \(v_i/\rho\). Hence the greedy ordering is the same as the ordering of values, \(v_{(1)}>v_{(2)}>\cdots>v_{(n)}\). Let \(m=\lfloor B/\rho\rfloor\). Then the first \(m\) items are packed integrally. The remaining budget is \(B-m\rho\in(0,\rho)\), where the inclusion follows from \(B/\rho\notin\mathbb Z\). Therefore item \((m+1)\) is packed fractionally with \[ x^*_{(m+1)}(\rho\mathbf 1_n) = \frac{B-m\rho}{\rho} \in(0,1). \] All subsequent items remain unpacked. This proves the lemma. 
\end{proof}

\end{document}